\documentclass{article}
\usepackage{amssymb}

\usepackage{graphicx}
\usepackage{amsmath}

\newtheorem{theorem}{Theorem}

\newtheorem{algorithm}[theorem]{Algorithm}

\newtheorem{definition}[theorem]{Definition}

\newtheorem{lemma}[theorem]{Lemma}

\newtheorem{remark}[theorem]{Remark}

\newenvironment{proof}[1][Proof]{\textbf{#1.} }{\ \rule{0.5em}{0.5em}}
\input{tcilatex}

\begin{document}

\author{V. S. Borisov \thanks{%
The Pearlstone Center for Aeronautical Engineering Studies, Department of
Mechanical Engineering, Ben-Gurion University of the Negev, Beer-Sheva,
Israel. E-mail: viatslav@bgu.ac.il} \and M. Mond \thanks{%
The Pearlstone Center for Aeronautical Engineering Studies, Department of
Mechanical Engineering, Ben-Gurion University of the Negev, Beer-Sheva,
Israel. E-mail: mond@bgu.ac.il}}
\title{On stability of difference schemes. Central schemes for hyperbolic
conservation laws with source terms}
\maketitle

\begin{abstract}
The stability of nonlinear explicit difference schemes with not, in general,
open domains of the scheme operators are studied. For the case of
path-connected, bounded, and Lipschitz domains, we establish the notion that
a multi-level nonlinear explicit scheme is stable \emph{iff} (if and only
if) the corresponding scheme in variations is stable. A new modification of
the central Lax-Friedrichs (LxF) scheme is developed to be of the second
order accuracy. The modified scheme is based on nonstaggered grids. A
monotone piecewise cubic interpolation is used in the central scheme to give
an accurate approximation for the model in question. The stability of the
modified scheme is investigated. Some versions of the modified scheme are
tested on several conservation laws, and the scheme is found to be accurate
and robust. As applied to hyperbolic conservation laws with, in general,
stiff source terms, it is constructed a second order nonstaggered central
scheme based on operator-splitting techniques.
\end{abstract}

%\date{}

\section{Introduction\label{Introduction}}

We are mainly concerned with the stability of nonlinear explicit difference
schemes arising, e.g., in numerical analysis of nonlinear PDE (partial
differential equation, such as parabolic, hyperbolic, etc.) systems,
generally, in multidimensional space. Stability is the central and the most
pressing problem in any algorithm \cite{Samarskii 2001}. Nowadays, there
exists a few methods for stability analysis of some classes of difference
schemes approximating PDE systems (see, e.g., \cite{Ganzha and Vorozhtsov
1996b}, \cite{Gil' 2007}, \cite{LeVeque 2002}, \cite{Morton 1996}, \cite
{Naterer and Camberos 2008}, \cite{Samarskii 2001} and references therein).
However, the problem of stability analysis for nonlinear schemes is still
one of the most burning problems, because of the absence of its complete
solution \cite{Gil' 2007}. In particular, the vast majority of difference
schemes, currently in use, have still not been analyzed \cite{Ganzha and
Vorozhtsov 1996b}, and, in general, no numerical method for non-linear
systems of equations has been proven to be stable \cite{LeVeque 2002}. A new
approach to testing scheme stability has been suggested in \cite{Borisov and
Mond 2010a} to prove convergence of non-linear schemes for systems of PDEs.
It was demonstrated \cite{Borisov and Mond 2010a} that the notion of scheme
in variations (or variational scheme \cite{Borisov and Sorek 2004}, \cite
{Borisov V.S. 2003}) has much potential to be an effective tool for studying
stability of nonlinear schemes. The scheme in variations for a difference
scheme represents a tangent space at a point of the manifold associated to
the difference scheme \cite{Borisov and Mond 2010a}, \cite{Borisov and Sorek
2004}. Thus, the scheme in variations will always be linear and, hence,
enables the investigation of the stability for nonlinear operators using, in
general, an infinite family of linear patterns. It was established in \cite
{Borisov and Mond 2010a} the notion that the stability of a scheme in
variations implies the stability of its original scheme, and that a
nonlinear explicit scheme will be stable \emph{iff} (if and only if) its
scheme in variations will be stable. In this paper we consider difference
schemes with not, in general, open domains of the scheme operators. For the
case of internal path-connected, bounded, and Lipschitz domains, we
establish the notion that a multi-level nonlinear explicit scheme is stable 
\emph{iff} the corresponding scheme in variations is stable.

By way of illustration of the developed approach we are concerned with the
stability analysis of central difference schemes for hyperbolic systems of
conservation laws with source terms. Such systems are used to describe many
physical problems of great practical importance in magneto-hydrodynamics,
kinetic theory of rarefied gases, linear and nonlinear waves,
viscoelasticity, multi-phase flows and phase transitions, shallow waters,
etc. (see, e.g., \cite{Bereux and Sainsaulieu 1997}, \cite{Caflisch at el.
1997}, \cite{Godlewski and Raviart 1996}, \cite{Jin Shi 1995}, \cite
{Kurganov and Tadmor 2000}, \cite{LeVeque 2002}, \cite{Monthe 2003}, \cite
{Naldi and Pareschi 2000}, \cite{Pareschi Lorenzo 2001}, \cite{Pareschi and
Russo 2005}). We will consider a system of hyperbolic conservation laws
written as follows (e.g., \cite{Godlewski and Raviart 1996}, \cite{LeVeque
2002}) 
\begin{equation}
\frac{\partial \mathbf{u}}{\partial t}+\sum_{j=1}^{N}\frac{\partial }{%
\partial x_{j}}\mathbf{f}_{j}\left( \mathbf{u}\right) =\frac{1}{\tau }%
\mathbf{q}\left( \mathbf{u}\right) ,\ 0<t\leq T_{\max },\ \left. \mathbf{u}%
\left( \mathbf{x},t\right) \right| _{t=0}=\mathbf{u}^{0}\left( \mathbf{x}%
\right) ,  \label{INA10}
\end{equation}
where the column-vector $\mathbf{x\equiv }\left\{ x_{1},x_{2},\ldots
,x_{N}\right\} ^{T}\in \mathbb{R}^{N}$, $\mathbf{u}=\left\{
u_{1},u_{2},\ldots ,u_{M}\right\} ^{T}$ is a vector-valued function from $%
\mathbb{R}^{N}$ $\times $ $[0,+\infty )$ into a subset $\Omega _{\mathbf{u}%
}\subset \mathbb{R}^{M}$, $\mathbf{f}_{j}\left( \mathbf{u}\right) =\left\{
f_{1j}\left( \mathbf{u}\right) \right. ,$ $f_{2j}\left( \mathbf{u}\right) ,$ 
$\ldots ,$ $\left. f_{Mj}\left( \mathbf{u}\right) \right\} ^{T}$ is a smooth
function (flux-function) from $\Omega _{\mathbf{u}}$ into $\mathbb{R}^{M}$, $%
\mathbf{q}\left( \mathbf{u}\right) =\left\{ q_{1}\left( \mathbf{u}\right)
,q_{2}\left( \mathbf{u}\right) ,\ldots ,q_{M}\left( \mathbf{u}\right)
\right\} ^{T}$ denotes the source term, $\tau >0$ denotes the stiffness
parameter, $\mathbf{u}^{0}\left( \mathbf{x}\right) $\ is of compact support.
It is assumed that $\tau =const$ without loss of generality. It is also
assumed that all eigenvalues, $\xi _{k}=\xi _{k}\left( \mathbf{u}\right) ,$
of the Jacobian matrix $\mathbf{G}$ ($\mathbf{=}\partial \mathbf{q}\left( 
\mathbf{u}\right) \diagup \partial \mathbf{u}$) have non-positive real
parts, i.e. 
\begin{equation}
R_{e}\xi _{k}\left( \mathbf{u}\right) \leq 0,\quad \forall k,\quad \forall 
\mathbf{u\in }\Omega _{\mathbf{u}}.  \label{INA20}
\end{equation}
In what follows $\left\| \mathbf{M}\right\| _{p}$ denotes the matrix norm of
a matrix $\mathbf{M}$ induced by the vector norm $\left\| \mathbf{v}\right\|
_{p}$ $=\left( \sum_{i}\left| v_{i}\right| ^{p}\right) ^{1/p}$, and $\left\| 
\mathbf{M}\right\| $ denotes the matrix norm induced by a prescribed vector
norm. $\mathbb{R}$ denotes the field of real numbers. The transpose of a
matrix (or vector) $\mathbf{S}$ is denoted by $\mathbf{S}^{T}$. The null
element in any linear space, as well as the number zero, will be denoted by
the same symbol $0$.

Central schemes are attractive for various reasons: no Riemann solvers,
characteristic decompositions, complicated flux splittings, etc., must be
involved in construction of a central scheme (see, e.g., \cite{Balaguer and
Conde 2005}, \cite{Borisov and Mond 2010b}, \cite{Jiang et al. 1998}, \cite
{Kurganov and Tadmor 2000}, \cite{Kurganov and Levy 2000}, \cite{LeVeque
2002}, \cite{Nessyahu and Tadmor 1990}, \cite{Pareschi Lorenzo 2001}, \cite
{Pareschi et al. 2005} and references therein). However, there is the need
for staggered central schemes to alternate between two staggered grids,
which can be cumbersome, e.g., near the boundaries \cite{Jiang et al. 1998}.
Moreover, there is a risk for every central scheme to exhibit spurious
solutions \cite{Borisov and Mond 2010b} in spite of sufficiently small CFL
(Courant-Friedrichs-Lewy \cite{LeVeque 2002}) number. To simplify
implementation of central schemes in the case of complex geometries and
boundary conditions, it was developed nonstaggered central schemes in \cite
{Jiang et al. 1998}. A new second-order modification of the staggered LxF
scheme was developed in \cite{Borisov and Mond 2010b} to avoid the risk of
spurious oscillations. It was demonstrated in \cite{Borisov and Mond 2010b}
that the higher order versions of LxF scheme can produce spurious
oscillations because of a negative numerical viscosity introduced to
increase accuracy of the scheme up to $O(\Delta t+\left( \Delta x\right)
^{2})$. To reduce the risk of spurious solutions, it was suggested in \cite
{Borisov and Mond 2010b} to introduce an additional non-negative numerical
viscosity such that the scheme's order of accuracy becomes $O((\Delta
t)^{2}+\left( \Delta x\right) ^{2})$. The developed central scheme was
tested on several conservation laws taking a CFL number equal or close to
unity, and the scheme was found to be accurate and robust. In this paper, we
extend the second-order modification to a new nonstaggered central scheme.
The stability of this scheme is proven in Section \ref{Stability of the
developed schemes} on the basis of the approach developed in Section \ref
{Stability of difference schemes} as well as in \cite{Borisov and Mond 2010a}%
, \cite{Borisov and Mond 2010b}. The scheme is tested on several
conservation laws in Section \ref{Exemplification and discussion}.

A stable numerical scheme may yield spurious results when applied to a stiff
hyperbolic system with relaxation (see, e.g., \cite{Ahmad and Berzins 2001}, 
\cite{Aves Mark A. et al. 2000}, \cite{Bereux and Sainsaulieu 1997}, \cite
{Caflisch at el. 1997}, \cite{Du Tao et al. 2003}, \cite{Jin Shi 1995}, \cite
{Pember 1993}, \cite{Pember 1993a}). It is significant that a numerical
scheme for relaxation systems must possess a discrete analogy to the
continuous asymptotic limit, because any scheme violating the correct
asymptotic limit leads to spurious or poor solutions (see, e.g., \cite
{Caflisch at el. 1997}, \cite{Jin Shi 1995}, \cite{Jin Shi et al. 2000}, 
\cite{Naldi and Pareschi 2000}, \cite{Pareschi Lorenzo 2001}). Most methods
for solving such systems can be described as operator splitting ones, \cite
{Du Tao et al. 2003}, or methods of fractional steps, \cite{Bereux and
Sainsaulieu 1997}. After operator splitting, one solves the advection
homogeneous system, and then the ordinary differential equations associated
with the source terms. We are mainly concerned with such an approach in
Section \ref{OSS}.

\section{Stability of nonlinear explicit schemes\label%
{Stability of difference schemes}}

We consider the following ($q+1$)-level difference scheme: 
\begin{equation*}
\mathbf{w}_{i}^{n+1}=\mathbf{G}_{i}^{n}(\mathbf{w}_{1}^{n},\mathbf{w}%
_{1}^{n-1},\ldots ,\mathbf{w}_{1}^{n-q+1},\mathbf{w}_{2}^{n},\mathbf{w}%
_{2}^{n-1},\ldots ,\mathbf{w}_{2}^{n-q+1},\ldots ,
\end{equation*}
\begin{equation}
\mathbf{w}_{I}^{n},\mathbf{w}_{I}^{n-1},\ldots ,\mathbf{w}%
_{I}^{n-q+1}),\quad \mathbf{G}_{i}^{n}:\Omega _{n}\subseteq \mathbb{R}%
^{N}\rightarrow \mathbb{R}^{N_{0}},  \label{NSCS04}
\end{equation}
where $q,N_{0},I\geq 1$ are finite integer constants, $N=qN_{0}I$, $\mathbf{w%
}_{i}^{k}\in \mathbb{R}^{N_{0}}$ denotes a vector-valued grid function, $%
i\in \omega _{1}$ denotes a node of the grid $\omega _{1}\equiv \left\{
1,2,\ldots ,I\right\} $, $k\in \omega _{2}$ denotes a node (time level) of
the grid $\omega _{2}\equiv \left\{ 0,1,\ldots ,M\right\} $, $M\geq q$ is a
finite integer constant. $\mathbf{G}_{i}^{n}$$=$ $\left\{
G_{i,1}^{n},\right. G_{i,2}^{n},$ $\ldots ,$ $\left. G_{i,N_{0}}^{n}\right\}
^{T}$ denotes the vector-valued function with the domain and range belonging
to $\mathbb{R}^{N}$ and $\mathbb{R}^{N_{0}}$, respectively. We emphasize
that the domain, $\Omega _{n}$, of the function $\mathbf{G}_{i}^{n}$ is the
set of argument values, which need not be open in $\mathbb{R}^{N}$. We
assume that $n$\ in (\ref{NSCS04}) denotes the time level, $t_{n}$ $\left(
=n\Delta t\right) $. Thus, the time increment will be represented by $\Delta
t=t_{\max }/M=const$, where $t_{\max }$\ denotes some finite time over which
we wish to compute. Using the notation 
\begin{equation*}
\mathbf{H}_{i}^{n}=\left\{ \left( \mathbf{G}_{i}^{n}\right) ^{T},\left( 
\mathbf{w}_{i}^{n}\right) ^{T},\left( \mathbf{w}_{i}^{n-1}\right)
^{T},\ldots ,\left( \mathbf{w}_{i}^{n-q+2}\right) ^{T}\right\} ^{T},
\end{equation*}
\begin{equation}
\mathbf{v}_{i}^{n}=\left\{ \left( \mathbf{w}_{i}^{n}\right) ^{T},\left( 
\mathbf{w}_{i}^{n-1}\right) ^{T},\ldots ,\left( \mathbf{w}%
_{i}^{n-q+1}\right) ^{T}\right\} ^{T},\quad i\in \omega _{1},\ n\in \omega
_{2},  \label{NSCS03}
\end{equation}
we reduce (see, e.g., \cite{Ganzha and Vorozhtsov 1996}, \cite{Richtmyer and
Morton 1967}) the multi-level scheme, (\ref{NSCS04}), to the following
two-level scheme: 
\begin{equation}
\mathbf{v}_{i}^{n+1}=\mathbf{H}_{i}^{n}(\mathbf{v}_{1}^{n},\mathbf{v}%
_{2}^{n},\ldots ,\mathbf{v}_{I}^{n}),\ \mathbf{H}_{i}^{n}:\Omega
_{n}\subseteq \mathbb{R}^{N}\rightarrow \mathbb{R}^{qN_{0}},\ i\in \omega
_{1},\ n,n+1\in \omega _{2}.  \label{NSCS10}
\end{equation}
If we define 
\begin{equation}
\mathbf{v}^{n}\mathbf{=}\left\{ \left( \mathbf{v}_{1}^{n}\right) ^{T},\left( 
\mathbf{v}_{2}^{n}\right) ^{T},\ldots ,\left( \mathbf{v}_{I}^{n}\right)
^{T}\right\} ^{T},\ \mathbf{H}^{n}\mathbf{=}\left\{ \left( \mathbf{H}%
_{1}^{n}\right) ^{T},\left( \mathbf{H}_{2}^{n}\right) ^{T},\ldots ,\left( 
\mathbf{H}_{I}^{n}\right) ^{T}\right\} ^{T},  \label{NSCS05}
\end{equation}
then the scheme (\ref{NSCS10}) becomes 
\begin{equation}
\mathbf{v}^{n+1}=\mathbf{H}^{n}(\mathbf{v}^{n}),\quad \mathbf{H}^{n}:\Omega
_{n}\subseteq \mathbb{R}^{N}\rightarrow \mathbb{R}^{N},\quad n,n+1\in \omega
_{2}\equiv \left\{ 0,1,\ldots ,M\right\} .  \label{NSCS20}
\end{equation}

As usual (e.g., \cite[p. 62]{Ortega and Rheinboldt 1970}), for mappings $%
\mathbf{f}:\Omega _{f}\subseteq \mathbb{R}^{N}\rightarrow \mathbb{R}^{N}$
and $\mathbf{g}:\Omega _{g}\subseteq \mathbb{R}^{N}\rightarrow \mathbb{R}%
^{N} $, the composite mapping $\mathbf{h}=\mathbf{g}\circ \mathbf{f}$ is
defined by $\mathbf{h}\left( \mathbf{v}\right) =\mathbf{g}\left( \mathbf{f}%
\left( \mathbf{v}\right) \right) $ for all $\mathbf{v\in }\Omega
_{h}=\left\{ \mathbf{v\in }\Omega _{f}\mid \mathbf{f}\left( \mathbf{v}%
\right) \in \Omega _{g}\right\} $. Using the composite mapping approach, we
rewrite Scheme (\ref{NSCS20}) to read 
\begin{equation}
\mathbf{y}=\mathbf{F}\left( \mathbf{x}\right) ,\quad \mathbf{F}:\Omega
_{F}\subseteq \mathbb{R}^{N}\rightarrow \mathbb{R}^{N},  \label{NSCS30}
\end{equation}
where the following notation is used: $\mathbf{x=v}^{0}$, $\mathbf{y}=%
\mathbf{v}^{M}$, $\mathbf{F=H}^{M-1}\circ \mathbf{H}^{M-2}\circ \ldots \circ 
\mathbf{H}^{0}$, $\Omega _{F}=\left\{ \mathbf{v}^{0}\mathbf{\in }\Omega
_{0}\mid \right. $ $\mathbf{v}^{1}=\mathbf{H}^{0}\left( \mathbf{v}%
^{0}\right) \in \Omega _{1}\mid $ $\ldots $ $\mid \mathbf{v}^{M-1}=$ $\left. 
\mathbf{H}^{M-2}\left( \mathbf{v}^{M-2}\right) \in \Omega _{M-1}\right\} $.
Notice, $\mathbf{F}$ in (\ref{NSCS30}) depends also on scheme parameters
(e.g., space and time increments), however, this dependence is usually not
included in the notation. Let the scheme parameters (including time
increments) be represented by a vector $\mathbf{s}$ belonging to some normed
space with the norm $\left| \mathbf{s}\right| $.

Let us remind the well-known definition of stability of an explicit scheme
(see, e.g., \cite{Ganzha and Vorozhtsov 1996b}, \cite{Godlewski and Raviart
1996}, \cite{LeVeque 2002}, \cite{Richtmyer and Morton 1967}, \cite
{Samarskii 2001}, \cite{Samarskiy and Gulin 1973}).

\begin{definition}
\label{DefStab1}Scheme (\ref{NSCS30}) is said to be stable if there exist
positive constants $s_{0}$ and $C$ such that for all $\mathbf{x,}$ $\mathbf{x%
}_{\ast }\in \Omega _{F}$ the following inequality is valid 
\begin{equation}
\left\| \mathbf{F}\left( \mathbf{x}_{\ast }\right) -\mathbf{F}\left( \mathbf{%
x}\right) \right\| \leq C\left\| \mathbf{x}_{\ast }-\mathbf{x}\right\|
,\quad \forall \ \mathbf{s}:\ \left| \mathbf{s}\right| \leq s_{0}.
\label{NSCS40}
\end{equation}
\end{definition}

Thus, scheme (\ref{NSCS30}) and, hence, scheme (\ref{NSCS04}) will be stable 
\emph{iff} the function $\mathbf{F}$ will be Lipschitz for a constant $C$.
Notice that this well-known definition of stability may lead to wrong
conclusion. Actually, let us consider the ``slit plane'' \cite{Heinonen Juha
2005} in polar coordinates $\left( r,\theta \right) $ 
\begin{equation}
\Omega _{F}=\left\{ \left( r,\theta \right) \mid 0<r<\infty ,\ -\pi <\theta
<\pi \right\} \subset \mathbb{R}^{2},  \label{NSCS42}
\end{equation}
and the vector-valued function $\mathbf{F}=\left\{ F_{1}\left( r,\theta
\right) ,F_{2}\left( r,\theta \right) \right\} ^{T}$ such that \cite
{Heinonen Juha 2005} 
\begin{equation}
F_{1}=r,\ F_{2}=\theta \diagup 2.  \label{NSCS44}
\end{equation}
We take $\left( r,\theta \right) _{\ast }=\left( r_{0},-\pi +\varepsilon
\right) $, $\left( r,\theta \right) =\left( r_{0},\pi -\varepsilon \right) $%
, and $r_{0}=const$. Considering the Euclidean distance between these two
points of the plane with Cartesian coordinates $\mathbf{x}_{\ast }\equiv
\left( -r_{0}\cos \varepsilon ,-r_{0}\sin \varepsilon \right) $ and $\mathbf{%
x\equiv }\left( -r_{0}\cos \varepsilon ,r_{0}\sin \varepsilon \right) $ we
find that the inequality (\ref{NSCS40}) takes the form: $\cos \left(
\varepsilon \diagup 2\right) \leq C\sin \varepsilon $. Obviously the mapping
(\ref{NSCS44}) is not Lipschitz, since $C$ in (\ref{NSCS40}) tends to
infinity\ as $\varepsilon \rightarrow 0$. Therefore, we have to conclude, in
view of the above definition, that the scheme (\ref{NSCS44}) is not stable,
even though for every point $\left( r,\theta \right) $ in $\Omega _{F}$, (%
\ref{NSCS42}), there exists a neighborhood of $\left( r,\theta \right) $
such that the function $\mathbf{F}$, (\ref{NSCS44}), restricted to the
neighborhood is Lipschitz for $C=1$, i.e. the function $\mathbf{F}$ is
locally Lipschitz. Further, $\mathbf{F}$ is the non-stretching mapping of
the ``slit plane'' (\ref{NSCS42}) into the right semi-plane. Hence, the
preceding definition of stability, i.e. Definition \ref{DefStab1}, needs to
be slightly improved. Such a problem has been considered in \cite{Borisov
and Mond 2010a} under the assumption that the domain $\Omega _{F}$ of the
vector-valued function $\mathbf{F}$, (\ref{NSCS30}), is open in $\mathbb{R}%
^{N}$. We will assume that the Lebesgue-measurable domain $\Omega _{F}$ of
the function $\mathbf{F}$, (\ref{NSCS30}), is a not, in general, open subset
of $\mathbb{R}^{N}$ with non-empty interior. 
%In order to give a more precise definition of difference scheme stability,

We shall start with some basic notions. A set $\Omega \subseteq \mathbb{R}%
^{N}$ is said to be internal path-connected (or path-connected, in the case
that $\Omega $ is open) if every two points $\mathbf{x}$, $\mathbf{x}_{\ast
} $ $\in \Omega $ can be joined by a continuous curve $\gamma :\left[ 0,1%
\right] \subset \mathbb{R}\rightarrow \Omega $ of finite length, $L\left(
\gamma \right) $, with $int\left( \gamma \right) \subseteq int\left( \Omega
\right) $. Here and in what follows $int\left( \Omega \right) $ denotes the\
interior of $\Omega \subseteq \mathbb{R}^{N}$ and $int\left( \gamma \right)
\equiv \gamma \diagdown (\gamma \left( 0\right) \cup \gamma \left( 1\right)
) $ denotes the interior of the continuous curve $\gamma $. The intrinsic
metric $\Lambda _{\Omega }$ on an internal path-connected set $\Omega $\ is
defined as 
\begin{equation}
\Lambda _{\Omega }\left( \mathbf{x},\mathbf{x}_{\ast }\right) =\underset{%
\gamma \subseteq \Omega }{\inf }L\left( \gamma \right) ,\quad \mathbf{x}%
=\gamma \left( 0\right) ,\ \mathbf{x}_{\ast }=\gamma \left( 1\right) ,\ 
\mathbf{x,x}_{\ast }\mathbf{\in }\Omega .  \label{NSCS46}
\end{equation}
An open ball (of radius $r>0$) about $\mathbf{x\in }\mathbb{R}^{N}$ is
denoted by $B\left( \mathbf{x},r\right) $ (or just $B_{\mathbf{x}}$). A
closed ball is denoted by $\overline{B}\left( \mathbf{x},r\right) $ (or just 
$\overline{B}_{\mathbf{x}}$). Given points $\mathbf{x},\mathbf{y\in }\mathbb{%
R}^{N}$, given an open ball $B\left( \mathbf{x},r_{\mathbf{x}}\right) $, and
given an open ball $B\left( \mathbf{y},r_{\mathbf{y}}\right) $ such that $%
\mathbf{x\notin }B_{\mathbf{y}}$, the convex open set $Q_{\mathbf{x}}=\{%
\mathbf{z}\mid \mathbf{x}+\lambda \left( \mathbf{z-x}\right) ,\ \lambda >0,\ 
\mathbf{z\in }B_{\mathbf{y}}\}\cap B_{\mathbf{x}}$ will be called a finite
cone with vertex $\mathbf{x}$. A set $\Omega \subseteq \mathbb{R}^{N}$ is
said to have the cone property if there exists a finite cone $Q_{\ast }$
such that each boundary point $\mathbf{x\in \partial }\Omega $ is the vertex
of a finite cone $Q_{\mathbf{x}}\subseteq int\left( \Omega \right) $
congruent to $Q_{\ast }$. Suppose that $\Omega \subseteq \mathbb{R}^{N}$ has
the cone property, and let $K_{\mathbf{x}}$ denote the union of all finite
cones with the vertex $\mathbf{x\in }\partial \Omega $, i.e. $K_{\mathbf{x}%
}=\cup Q_{\mathbf{x}}$. Then a function $\mathbf{F}:\Omega \subseteq \mathbb{%
R}^{N}\rightarrow \mathbb{R}^{N}$ will be called locally Lipschitz (or
locally Lipschitz continuous) if for every $\mathbf{x\in }\Omega $ there
exist $B_{\mathbf{x}}$ and a constant $c_{\mathbf{x}}>0$ such that $\left\| 
\mathbf{F}\left( \mathbf{x}\right) -\mathbf{F}\left( \mathbf{y}\right)
\right\| \leq c_{\mathbf{x}}\left\| \mathbf{x}-\mathbf{y}\right\| $ for each 
$\mathbf{y\in }B_{\mathbf{x}}$ provided $\mathbf{x\in }int\left( \Omega
\right) $, or $\mathbf{y\in }B_{\mathbf{x}}\cap K_{\mathbf{x}}$ otherwise. $%
\Omega \subset \mathbb{R}^{N}$ is Lipschitz (has a Lipschitz boundary), if
for each point $\mathbf{x\in }\partial \Omega $ there exists an open ball $%
B_{\mathbf{x}}$ such that $\partial \Omega \cap B_{\mathbf{x}}$ is the graph
of a Lipschitz function and $int\left( \Omega \right) $ lies on one side of
its boundary 
%The last requirement will be satisfied if $int\left( \Omega \right) \cap
%B\left( \mathbf{x},r\right) $ will be path-connected for all $r\leq r_{%
%\mathbf{x}}$
(for more details see, e.g., \cite[p. 149]{Burenkov V.I. 1998}, \cite[p. 354]
{Leoni Giovanni 2009}). Let us note that $\Omega \subset \mathbb{R}^{N}$
satisfies the cone property if it has a Lipschitz boundary (see, e.g., 
\cite[p. 151]{Burenkov V.I. 1998}, \cite[p. 355]{Leoni Giovanni 2009}), but
not vice-versa as the example (\ref{NSCS42}) shows. $W^{1,\infty }\left(
\Omega \right) $ denotes the Sobolev space (see Definition 5 in \cite[p. 28]
{Burenkov V.I. 1998}, see also \cite{Heinonen Juha 2005} and \cite{Leoni
Giovanni 2009}), i.e., the space consisting of all bounded functions on open 
$\Omega \subset \mathbb{R}^{N}$ whose distributional derivatives are bounded
functions as well. Let $\mathbf{F\equiv }\left\{ F_{1}\right. ,$ $F_{2},$ $%
\ldots ,$ $\left. F_{N}\right\} ^{T}$ in (\ref{NSCS30}), let $\nabla
F_{i}\equiv \{\partial _{1}F_{i},$ $\ldots ,$ $\partial _{N}F_{i}\}$ denote
the distributional gradient of $F_{i}$, and let $\delta \mathbf{F,}$%
\thinspace $\delta \mathbf{x}$ $\mathbf{\in }$ $\mathbb{R}^{N}$ denote
variations. The following equality 
\begin{equation}
\delta \mathbf{F}=\mathbf{F}^{\prime }\cdot \delta \mathbf{x,\quad F}%
^{\prime }\equiv \left\{ \left( \nabla F_{1}\right) ^{T},\left( \nabla
F_{2}\right) ^{T},\ldots ,\left( \nabla F_{N}\right) ^{T}\right\} ^{T},
\label{NDS55}
\end{equation}
will be viewed as the scheme in variations for (\ref{NSCS30}). The matrix
representation of $\mathbf{F}^{\prime }$ in (\ref{NDS55})\ is given by the
following Jacobian matrix\ \cite{Ortega and Rheinboldt 1970}: $\mathbf{F}%
^{\prime }=\left\{ \partial F_{i}\diagup \partial x_{j}\right\} $, $%
i,j=1,2,\ldots ,N$.

\begin{definition}
\label{DefStability}Let $\Omega _{F}$ in (\ref{NSCS30}) be internal
path-connected. Scheme (\ref{NSCS30}) is said to be stable if there exist
positive constants $s_{0}$ and $C$ such that the following inequality holds 
\begin{equation}
\left\| \mathbf{F}\left( \mathbf{x}_{\ast }\right) -\mathbf{F}\left( \mathbf{%
x}\right) \right\| \leq C\Lambda _{\Omega _{F}}\left( \mathbf{x},\mathbf{x}%
_{\ast }\right) ,\quad \forall \ \mathbf{x,x}_{\ast }\in \Omega _{F},\quad
\forall \ \mathbf{s}:\ \left| \mathbf{s}\right| \leq s_{0}.  \label{NDS20}
\end{equation}
\end{definition}

Notice that the scheme (\ref{NSCS44}) is stable in terms of Definition \ref
{DefStability}, since the inequality (\ref{NDS20}) holds for $C=1$. It is
significant that the domain $\Omega _{F}$, (\ref{NSCS42}), of the function $%
\mathbf{F}$, (\ref{NSCS44}), is open in $\mathbb{R}^{2}$. However, if the
domain of the function $\mathbf{F}$, (\ref{NSCS44}), is not open, then the
situation can be in complete contrast to the previous one. Actually, let us
add to $\Omega _{F}$, (\ref{NSCS42}), the only boundary point $\left( 1,\pi
\right) $, i.e. let $\Omega _{F}^{\ast }=\Omega _{F}\cup \left( 1,\pi
\right) $ will be the domain of the function $\mathbf{F}^{\ast }=\mathbf{F}$%
, (\ref{NSCS44}), if $\left( r,\theta \right) \in \Omega _{F}$ and $\mathbf{F%
}^{\ast }=\left( 1,\pi \diagup 2\right) $ if $\left( r,\theta \right) 
\mathbf{=}\left( 1,\pi \right) $. We take $\left( r,\theta \right) _{\ast
}=\left( 1,-\pi +\varepsilon \right) $, $\left( r,\theta \right) =\left(
1,\pi \right) $. Considering the Euclidean distance between these two points
of the plane with Cartesian coordinates $\mathbf{x}_{\ast }\equiv \left(
-\cos \varepsilon ,-r_{0}\sin \varepsilon \right) $ and $\mathbf{x\equiv }%
\left( -1,0\right) $ we find that the inequality (\ref{NDS20}) takes the
form: $\cos \left( \varepsilon \diagup 4\right) \leq C\sin \left(
\varepsilon \diagup 2\right) $. Thus, the scheme (\ref{NSCS44}), namely $%
\mathbf{F}:\Omega _{F}^{\ast }\subset \mathbb{R}^{2}\rightarrow \mathbb{R}%
^{2}$, is not stable, since $C$ in (\ref{NDS20}) tends to infinity\ as $%
\varepsilon \rightarrow 0$.

\begin{lemma}
\label{LipschitzStability}Let the path-connected $\Omega _{F}$ of (\ref
{NSCS30}) be open in $\mathbb{R}^{N}$. Scheme (\ref{NSCS30}) will be stable
in terms of Definition \ref{DefStability} \emph{iff}\ $\mathbf{F}$\ in (\ref
{NSCS30}) will be locally Lipschitz for a common constant $C$, for all
scheme parameters $\mathbf{s}$ such\ that $\left| \mathbf{s}\right| \leq
s_{0}$.
\end{lemma}

\begin{proof}
Suppose Scheme (\ref{NSCS30}) is stable, i.e. (\ref{NDS20}) is valid. Choose
any point $\mathbf{x\in }\Omega _{F}$. Since $\Omega _{F}$ is open, there
exists a radius $r$ such that $B\left( \mathbf{x},r\right) \subset \Omega
_{F}$. Choose any point $\mathbf{x}_{\ast }\in B\left( \mathbf{x},r\right) $%
, and let $\gamma _{\ast }$ be the straight line segment joining the points $%
\mathbf{x}$, $\mathbf{x}_{\ast }$ $\in $ $B\left( \mathbf{x},r\right) $. In
view of (\ref{NDS20}), $\mathbf{F}$\ in (\ref{NSCS30}) will be locally
Lipschitz for a common constant $C$, for all $\mathbf{s}:\ \left| \mathbf{s}%
\right| \leq s_{0}$, since $\Lambda _{\Omega _{F}}\left( \mathbf{x},\mathbf{x%
}_{\ast }\right) =L\left( \gamma _{\ast }\right) =\left\| \mathbf{x}_{\ast }-%
\mathbf{x}\right\| $.

Conversely, suppose that $\mathbf{F}$\ in (\ref{NSCS30}) is locally
Lipschitz for a common constant $C$, for all $\mathbf{s}:\ \left| \mathbf{s}%
\right| \leq s_{0}$. Let some points $\mathbf{x}$, $\mathbf{x}_{\ast }$ $\in
\Omega _{F}$ be joined by a continuous curve $\gamma $. In view of (\ref
{NSCS46}), the curve $\gamma $ can be taken such that $L\left( \gamma
\right) \leq \Lambda _{\Omega _{F}}\left( \mathbf{x},\mathbf{x}_{\ast
}\right) +\varepsilon $ for an arbitrary $\varepsilon >0$. Given any point $%
\mathbf{z}\in \gamma $, there is a ball $B_{\mathbf{z}}\subset \Omega _{F}$
such that $\mathbf{F}$\ restricted to $B_{\mathbf{z}}$ is Lipschitz for the
common constant $C$ (locally Lipschitz continuity). The balls $\left\{ B_{%
\mathbf{z}}\right\} $ form an open cover of $\gamma $. Since the mapping $%
\gamma :\left[ 0,1\right] \subset \mathbb{R}\rightarrow \mathbb{R}^{N}$ is
continuous, the curve $\gamma $ is compact \cite[p. 94]{Kolmogorov and Fomin
1970}. Hence, by the compactness of $\gamma $, $\left\{ B_{\mathbf{z}%
}\right\} $ has a finite subcover consisting of balls $B_{\mathbf{x}}=B_{%
\mathbf{z}_{1}}$, $B_{\mathbf{z}_{2}}$, $\ldots $, $B_{\mathbf{z}_{K}}=B_{%
\mathbf{x}_{\ast }}$. Since $\mathbf{F}$\ is locally Lipschitz, we find 
\begin{equation}
\left\| \mathbf{F}\left( \mathbf{z}_{k+1}\right) -\mathbf{F}\left( \mathbf{z}%
_{k}\right) \right\| \leq C\left\| \mathbf{z}_{k+1}-\mathbf{z}_{k}\right\|
,\quad k=1,2,\ldots K-1,\ \forall \ \mathbf{s}:\ \left| \mathbf{s}\right|
\leq s_{0}.  \label{NDS30}
\end{equation}
Then, by virtue of (\ref{NDS30}), we find 
\begin{equation*}
\left\| \mathbf{F}\left( \mathbf{x}_{\ast }\right) -\mathbf{F}\left( \mathbf{%
x}\right) \right\| =\left\| \sum_{k}\left[ \mathbf{F}\left( \mathbf{z}%
_{k+1}\right) -\mathbf{F}\left( \mathbf{z}_{k}\right) \right] \right\| \leq
C\sum_{k}\left\| \mathbf{z}_{k+1}-\mathbf{z}_{k}\right\| \leq
\end{equation*}
\begin{equation}
CL\left( \gamma \right) \leq C\Lambda _{\Omega _{F}}\left( \mathbf{x},%
\mathbf{x}_{\ast }\right) +\varepsilon C,\quad \forall \ \mathbf{s}:\ \left| 
\mathbf{s}\right| \leq s_{0}.  \label{NDS40}
\end{equation}
By letting $\varepsilon \rightarrow 0$, we find that (\ref{NDS20}) holds.
\end{proof}

Notice, in view of Lemma \ref{LipschitzStability}, the scheme (\ref{NSCS44})
is stable in terms of Definition \ref{DefStability}, since the domain $%
\Omega _{F}$, (\ref{NSCS42}), is open in $\mathbb{R}^{2}$ and
path-connected, and the function $\mathbf{F}$, (\ref{NSCS44}), is locally
Lipschitz for the common constant $C=1$.

\begin{theorem}
\label{LipStabTheorem}Let the internal path-connected $\Omega _{F}$ of (\ref
{NSCS30}) have the cone property, then the scheme (\ref{NSCS30}) will be
stable in terms of Definition \ref{DefStability} \emph{iff}\ $\mathbf{F}$\
in (\ref{NSCS30}) will be locally Lipschitz for a common constant $C$, for
all $\mathbf{s}$ such\ that $\left| \mathbf{s}\right| \leq s_{0}$.
\end{theorem}

\begin{proof}
Let $K_{\mathbf{x}}$ denote the union of all finite cones with the vertex $%
\mathbf{x\in }\partial \Omega _{F}$. Suppose Scheme (\ref{NSCS30}) is
stable. Choose any point $\mathbf{x\in }\Omega _{F}$. If $\mathbf{x\in }%
int\left( \Omega _{F}\right) $, then, in view of Lemma \ref
{LipschitzStability}, $\mathbf{F}$\ in (\ref{NSCS30}) will be locally
Lipschitz. If $\mathbf{x\notin }int\left( \Omega _{H}\right) $, then (cone
property) for each $\mathbf{x}_{\ast }\in K_{\mathbf{x}}$ there is a
straight line segment joining the points $\mathbf{x}$, $\mathbf{x}_{\ast }$.
Since $\Lambda _{\Omega _{F}}\left( \mathbf{x},\mathbf{x}_{\ast }\right)
=\left\| \mathbf{x}_{\ast }-\mathbf{x}\right\| $, the necessity, in view of (%
\ref{NDS20}), is proven.

Conversely, suppose $\mathbf{F}$\ in (\ref{NSCS30}) is locally Lipschitz.
Let some points $\mathbf{x}$, $\mathbf{x}_{\ast }$ $\in \Omega _{F}$ be
joined by a continuous curve $\gamma $. If $\mathbf{x}$, $\mathbf{x}_{\ast }$
$\in int\left( \Omega _{F}\right) $, then, in view of Lemma \ref
{LipschitzStability}, we find that (\ref{NDS20}) holds. Let $\mathbf{x\notin 
}int\left( \Omega _{F}\right) $, however $\mathbf{x}_{\ast }\in int\left(
\Omega _{F}\right) $. Choose any point $\mathbf{z}\in K_{\mathbf{x}}$ such
that $\left\| \mathbf{x}-\mathbf{z}\right\| =\varepsilon $. Then, since $%
\mathbf{F}$\ is locally Lipschitz, we have $\left\| \mathbf{F}\left( \mathbf{%
x}\right) -\mathbf{F}\left( \mathbf{z}\right) \right\| \leq C\left\| \mathbf{%
x}-\mathbf{z}\right\| $ for a sufficiently small $\varepsilon $. Since $%
\mathbf{z}\in int\left( \Omega _{F}\right) $, we find, in view of Lemma \ref
{LipschitzStability}, that $\left\| \mathbf{F}\left( \mathbf{z}\right) -%
\mathbf{F}\left( \mathbf{x}_{\ast }\right) \right\| \leq C\Lambda _{\Omega
_{F}}\left( \mathbf{z},\mathbf{x}_{\ast }\right) $. Then 
\begin{equation*}
\left\| \mathbf{F}\left( \mathbf{x}\right) -\mathbf{F}\left( \mathbf{x}%
_{\ast }\right) \right\| \leq \left\| \mathbf{F}\left( \mathbf{x}\right) -%
\mathbf{F}\left( \mathbf{z}\right) \right\| +\left\| \mathbf{F}\left( 
\mathbf{z}\right) -\mathbf{F}\left( \mathbf{x}_{\ast }\right) \right\|
\end{equation*}
\begin{equation}
\leq \varepsilon C+\left\| \mathbf{F}\left( \mathbf{z}\right) -\mathbf{F}%
\left( \mathbf{x}_{\ast }\right) \right\| ,\quad \forall \ \mathbf{z}:\
\left\| \mathbf{x}-\mathbf{z}\right\| =\varepsilon ,\ \forall \ \mathbf{s}:\
\left| \mathbf{s}\right| \leq s_{0}.  \label{NDS47}
\end{equation}
Since $\left\| \mathbf{F}\left( \mathbf{x}\right) -\mathbf{F}\left( \mathbf{x%
}_{\ast }\right) \right\| $ does not depend on $\mathbf{z}$, we find 
\begin{equation}
\left\| \mathbf{F}\left( \mathbf{x}\right) -\mathbf{F}\left( \mathbf{x}%
_{\ast }\right) \right\| \leq \underset{\mathbf{z}:\,\left\| \mathbf{x}-%
\mathbf{z}\right\| =\varepsilon }{\inf }\left\{ \varepsilon C+\left\| 
\mathbf{F}\left( \mathbf{z}\right) -\mathbf{F}\left( \mathbf{x}_{\ast
}\right) \right\| \right\}  \label{NDS48}
\end{equation}
\begin{equation}
\leq \varepsilon C+C\underset{\mathbf{z}:\,\left\| \mathbf{x}-\mathbf{z}%
\right\| =\varepsilon }{\,\inf }\Lambda _{\Omega _{F}}\left( \mathbf{z},%
\mathbf{x}_{\ast }\right) ,\quad \forall \ \mathbf{s}:\ \left| \mathbf{s}%
\right| \leq s_{0}.  \label{NDS50}
\end{equation}
By letting $\varepsilon \rightarrow 0$, i.e. $\mathbf{z}\rightarrow \mathbf{x%
}$, we find that (\ref{NDS20}) holds. The proof in the case $\mathbf{x}$, $%
\mathbf{x}_{\ast }$ $\notin int\left( \Omega _{F}\right) $ is reduced to the
previous one by choosing $\mathbf{z}\in K_{\mathbf{x}_{\ast }}$.
\end{proof}

We shall need the following theorem that identifies the Sobolev space $%
W^{1,\infty }\left( \Omega \right) $ as a space of locally Lipschitz
functions.

\begin{theorem}
\label{SobolevLipschitz}Let $\Omega _{F}\subseteq \mathbb{R}^{N}$ be open,
and let $\mathbf{F=}\left\{ F_{1}\right. ,$ $F_{2},$ $\ldots ,$ $\left.
F_{N}\right\} ^{T}$ in (\ref{NSCS30}). Then, $F_{i}$, $i=1,2,\ldots ,N$,
(and, hence, $\mathbf{F}$) is locally Lipschitz (in the sense of having
representatives) \emph{iff}\ $F_{i}\in W^{1,\infty }\left( \Omega
_{F}\right) $.
\end{theorem}

The proof of Theorem \ref{SobolevLipschitz} can be found, e.g., in 
\cite[Theorem 4.1]{Heinonen Juha 2005}, \cite[p. 342]{Leoni Giovanni 2009}.
The following lemma gives the necessary and sufficient conditions for the
stability of the scheme in variations (\ref{NDS55}).

\begin{lemma}
\label{LinSchStab}Linear Scheme (\ref{NDS55}) will be stable \emph{iff}\
there exist positive $s_{0}$,$C=const$ such that 
\begin{equation}
\left\| \mathbf{F}^{\prime }\right\| \leq C=const,\quad \forall \ \mathbf{x}%
\in \Omega _{F},\quad \forall \ \mathbf{s}:\ \left| \mathbf{s}\right| \leq
s_{0}.  \label{NDS65}
\end{equation}
\end{lemma}

The proof of Lemma \ref{LinSchStab} can be found in \cite{Borisov and Mond
2010a} (see also, e.g., \cite{Ganzha and Vorozhtsov 1996b}, \cite{Richtmyer
and Morton 1967}, \cite{Samarskii 2001}, \cite{Samarskiy and Gulin 1973}).
Let us remind the well-known theorem that ensures that Lipschitz domains are
extension domains (see \cite[pp. 320, 356]{Leoni Giovanni 2009}, \cite[pp.
149, 285]{Burenkov V.I. 1998}).

\begin{theorem}
\label{LipschitzExtension}Let $\Omega \subset \mathbb{R}^{N}$ be an open set
with uniformly Lipschitz boundary. Then there exists a continuous linear
operator $\frak{L}\,$: $W^{1,\infty }\left( \Omega \right) \rightarrow
W^{1,\infty }\left( \mathbb{R}^{N}\right) $ such that $\frak{L}(u)(\mathbf{x}%
)=u(\mathbf{x})$ for all $u\in W^{1,\infty }\left( \Omega \right) $.
\end{theorem}

The interconnection between stability of the nonlinear scheme (\ref{NSCS30})
and stability of its scheme in variations is established by the following
theorem. It is assumed that $\Omega _{F}$ in (\ref{NSCS30}) need not be
open. It is also assumed that the set $\Omega _{F}$ is bounded, i.e. there
exists $r_{0}<\infty $ such that $\Omega _{F}\subset B_{0}\equiv B\left(
0,r_{0}\right) $. Let $\mathbf{F\equiv }\left\{ F_{1}\right. ,$ $F_{2},$ $%
\ldots ,$ $\left. F_{N}\right\} ^{T}$ and let $\nabla F_{i}$, $i=1,2,\ldots
,N$, denote the distributional gradient of $F_{i}$. Definition of the
distributional gradient (or distributional partial derivative) of $F_{i}$ at 
$\mathbf{x}\in int\left( \Omega _{F}\right) $ can be found, e.g., in \cite
{Leoni Giovanni 2009}. If $\mathbf{x}\in \partial \Omega _{F}$, while $%
\Omega _{F}$ is Lipschitz, then the distributional gradient, $\nabla F_{i}$,
is assumed to be equal to $\nabla \frak{L}(F_{i})(\mathbf{x})$.

\begin{theorem}
\label{iffStab}Consider the scheme (\ref{NSCS30}). Let $\mathbf{F}$ be
bounded and let $\Omega _{F}$ be internal path-connected, bounded, and
Lipschitz. Then, the scheme, (\ref{NSCS30}), will be stable in terms of
Definition \ref{DefStability} \emph{iff}\ its scheme in variations, (\ref
{NDS55}), will be stable.
\end{theorem}

\begin{proof}
Let the scheme (\ref{NSCS30}) be stable, then, in view of Theorem \ref
{LipStabTheorem}, $\mathbf{F}$\ will be locally Lipschitz for a common
constant $C$, and, hence, in view of Theorem \ref{SobolevLipschitz}, $%
F_{i}\in W^{1,\infty }\left( int\left( \Omega _{F}\right) \right) $, $%
i=1,2,\ldots ,N$. By virtue of Theorem \ref{LipschitzExtension} we can find $%
F_{i}^{\ast }\in W^{1,\infty }\left( \mathbb{R}^{N}\right) $ such that $%
F_{i}^{\ast }(\mathbf{x})=F_{i}(\mathbf{x})$ on $int\left( \Omega
_{F}\right) $. Hence $F_{i}^{\ast }\in W^{1,\infty }\left( B_{0}\right) $,
where $B_{0}$ is an open ball such that $\Omega _{F}\subset B_{0}\subset 
\mathbb{R}^{N}$. Notice, $\left\| \nabla F_{i}^{\ast }\right\| _{\infty }$, $%
i=1,2,\ldots ,N$, is bounded on the ball $B_{0}$ and, naturally, on the
domain $\Omega _{F}\subset B_{0}$, since $F_{i}^{\ast }\in W^{1,\infty
}\left( B_{0}\right) $. In such a case there exists a common constant $C$
such that $\left\| \mathbf{F}^{\prime }\right\| \leq C$ for all $\mathbf{x}%
\in \Omega _{F}$, where $\left\| \mathbf{F}^{\prime }\right\| \equiv \left\| 
\mathbf{f}_{F}\right\| $, $\mathbf{f}_{F}\equiv \left\{ \left\| \nabla
F_{1}\right\| _{\infty }\right. $ $,$ $\left\| \nabla F_{2}\right\| _{\infty
},$ $\ldots ,$ $\left. \left\| \nabla F_{N}\right\| _{\infty }\right\} ^{T}$%
. Hence, in view of Lemma \ref{LinSchStab}, the scheme in variations, (\ref
{NDS55}), is stable.

Conversely, suppose that the scheme in variations, (\ref{NDS55}), is stable.
In view of Lemma \ref{LinSchStab}, there exists a common constant $C$ such
that $\left\| \mathbf{F}^{\prime }\right\| \leq C$ on $\Omega _{F}$, and,
hence, $\left\| \nabla F_{i}\right\| _{\infty }$, $i=1,2,\ldots ,N$, is
bounded on $\Omega _{F}$. Consequently, $F_{i}\in W^{1,\infty }\left(
int\left( \Omega _{F}\right) \right) $, $i=1,2,\ldots ,N$. Since $int\left(
\Omega _{F}\right) $ is an extension domain for $W^{1,\infty }\left(
int\left( \Omega _{F}\right) \right) $, we can find, in view of Theorem \ref
{LipschitzExtension}, $F_{i}^{\ast }\in W^{1,\infty }\left( \mathbb{R}%
^{N}\right) $ such that $F_{i}^{\ast }(\mathbf{x})=F_{i}(\mathbf{x})$ on $%
int\left( \Omega _{F}\right) $. In view of Theorem \ref{SobolevLipschitz}, $%
F_{i}^{\ast }$, $i=1,2,\ldots ,N$,\ will be locally Lipschitz on $\mathbb{R}%
^{N}$ (and, hence, on $\Omega _{F}$)\ for a common constant. Then $F_{i}$, $%
i=1,2,\ldots ,N$,\ and hence $\mathbf{F}$\ will be locally Lipschitz on $%
\Omega _{F}$, since $\Omega _{F}$ being Lipschitz\ satisfies the cone
property. Thus, in view of Theorem \ref{LipStabTheorem}, the scheme (\ref
{NSCS30}) will be stable.
\end{proof}

Notice, if $\mathbf{F}$ in the scheme (\ref{NSCS30}) is
Gateaux-differentiable, then $\nabla F_{i}$\ (see Theorem \ref{iffStab})\
denotes the classical gradient of $F_{i}$, and, hence, it may be taken that $%
\mathbf{f}_{F}=\left\{ \left\| \nabla F_{1}\right\| \right. $ $,$ $\left\|
\nabla F_{2}\right\| ,$ $\ldots ,$ $\left. \left\| \nabla F_{N}\right\|
\right\} ^{T}$.

\begin{remark}
Definition \ref{DefStab1} is ``restored in its rights'' for domains being
internal path-connected, bounded, and Lipschitz (see the proof of Theorem 
\ref{iffStab}). Namely, considering the scheme (\ref{NSCS30}) with such a
domain, we find that the scheme will be stable in terms of Definition \ref
{DefStability} \emph{iff} it will be stable in terms of Definition \ref
{DefStab1}.
\end{remark}

\section{Construction of stable central schemes\label{COSN}}

In this section we consider explicit schemes on a uniform grid with time
step $\Delta t$ and spatial mesh size $\Delta x$, as applied to 1-D
hyperbolic equations. In view of the CFL condition \cite{LeVeque 2002}, we
assume for the explicit schemes that $\Delta t=O\left( \Delta x\right) $.
Moreover, we will also assume that $\Delta x=O\left( \Delta t\right) $,
since a central scheme generates, in general, a conditional approximation
(see \cite{Borisov and Mond 2010b}). In such a case, the following
inequalities will be valid, for sufficiently small $\Delta t$ and $\Delta x$%
, 
\begin{equation}
\nu _{0}\Delta t\leq \Delta x\leq \mu _{0}\Delta t,\quad \nu _{0},\mu
_{0}=const,\ 0<\nu _{0}\leq \mu _{0}.  \label{CA05}
\end{equation}
We will focus on the following 1-D version of the problem (\ref{INA10}): 
\begin{equation}
\frac{\partial \mathbf{u}}{\partial t}+\frac{\partial }{\partial x}\mathbf{f}%
\left( \mathbf{u}\right) =0,\ t_{n}<t\leq t_{n+1}\equiv t_{n}+\Delta t,\quad 
\mathbf{u}\left( x,t_{n}\right) =\mathbf{u}^{n}\left( x\right) ,  \label{C10}
\end{equation}
Since the system (\ref{C10}) is hyperbolic, the Jacobian matrix of $\mathbf{f%
}\left( \mathbf{u}\right) $ possesses $M$ linearly independent eigenvectors
(see, e.g., \cite{Godlewski and Raviart 1996}). In addition, it is also
assumed that 
\begin{equation}
\underset{\mathbf{u\in }\Omega _{\mathbf{u}}}{\sup }\left\| \mathbf{A}%
\right\| _{2}\leq \lambda _{\max }<\infty ,\quad \mathbf{A=}\frac{\partial 
\mathbf{f}\left( \mathbf{u}\right) }{\partial \mathbf{u}}.  \label{IN80}
\end{equation}

Using the central differencing, we write 
\begin{equation}
\left. \frac{\partial \mathbf{u}}{\partial t}\right| _{t=t_{n+0.25},\
x=x_{i+0.5}}=\frac{\mathbf{u}_{i+0.5}^{n+0.5}-\mathbf{u}_{i+0.5}^{n}}{%
0.5\Delta t}+O\left( \left( \Delta t\right) ^{2}\right) ,  \label{C24}
\end{equation}
\begin{equation}
\left. \frac{\partial \mathbf{f}}{\partial x}\right| _{t=t_{n+0.25},\
x=x_{i+0.5}}=\frac{\mathbf{f}_{i+1}^{n+0.25}-\mathbf{f}_{i}^{n+0.25}}{\Delta
x}+O\left( \left( \Delta x\right) ^{2}\right) .  \label{C25}
\end{equation}
By virtue of (\ref{C24})-(\ref{C25}) we approximate (\ref{C10}) on the cell $%
\left[ x_{i},x_{i+1}\right] \times \left[ t_{n},t_{n+0.5}\right] $ by the
following difference equation 
\begin{equation}
\mathbf{v}_{i+0.5}^{n+0.5}=\mathbf{v}_{i+0.5}^{n}-\frac{\Delta t}{2\Delta x}%
\left( \mathbf{g}_{i+1}^{n+0.25}-\mathbf{g}_{i}^{n+0.25}\right) .
\label{C30}
\end{equation}
As usual, the mathematical treatment for the second step (i.e., on the cell $%
\left[ x_{i-0.5},x_{i+0.5}\right] \times \left[ t_{n+0.5},t_{n+1}\right] $)
of a staggered scheme will, in general, not be included in the text, because
it is quite similar to the treatment for the first step.

Considering that (\ref{C30}) approximates (\ref{C10}) with the accuracy $%
O(\left( \Delta x\right) ^{2}+\left( \Delta t\right) ^{2})$, the next
problem is to approximate $\mathbf{v}_{i+0.5}^{n}$ and $\mathbf{g}%
_{i}^{n+0.25}$ in such a way as to retain the accuracy of the approximation.
For instance, the following approximations 
\begin{equation}
\mathbf{v}_{i+0.5}^{n}=0.5\left( \mathbf{v}_{i}^{n}+\mathbf{v}%
_{i+1}^{n}\right) +O\left( \left( \Delta x\right) ^{2}\right) ,\quad \mathbf{%
g}_{i}^{n+0.25}=\mathbf{f}\left( \mathbf{v}_{i}^{n}\right) +O\left( \Delta
t\right) ,  \label{C50}
\end{equation}
leads to the staggered form of the famed LxF scheme that is of the
first-order approximation (see, e.g., \cite[p. 170]{Godlewski and Raviart
1996}). One way to obtain a higher-order scheme is to use a higher order
interpolation. At the same time it is required of the interpolant to be
monotonicity preserving. Notice, the classic cubic spline \cite{Press
William 1988} does not possess such a property (see Figure \ref{Fritsch}a).
In the following, subsection \ref{COS}, we consider the problem of
high-order interpolation of $\mathbf{v}_{i+0.5}^{n}$ in (\ref{C30}) with
closer inspection.

\subsection{Monotone $C^{1}$ piecewise cubics in construction of explicit
schemes\label{COS}}

We will consider some theoretical aspects for high-order interpolation and
employment of monotone $C^{1}$ piecewise cubics (see, e.g., \cite{Fritsch
and Carlson 1980}, \cite{Kocic and Milovanovic 1997}) in construction of
monotone explicit schemes.

Let $\mathbf{p}=\mathbf{p}\left( x\right) \equiv \left\{ p^{1}\left(
x\right) ,\ldots ,p^{k}\left( x\right) ,\ldots ,p^{m}\left( x\right)
\right\} ^{T}$ be an interpolant, and let 
\begin{equation*}
\mathbf{p}_{i}=\mathbf{p}\left( x_{i}\right) ,\quad \mathbf{p}_{i}^{\prime }=%
\mathbf{p}^{\prime }\left( x_{i}\right) ,\quad \Delta \mathbf{p}_{i}=\mathbf{%
p}_{i+1}-\mathbf{p}_{i},
\end{equation*}
\begin{equation}
\mathbf{p}_{i}^{\prime }=\mathbb{A}_{i}\cdot \frac{\Delta \mathbf{p}_{i}}{%
\Delta x},\quad \mathbf{p}_{i+1}^{\prime }=\mathbb{B}_{i}\cdot \frac{\Delta 
\mathbf{p}_{i}}{\Delta x},  \label{C80}
\end{equation}
where $\mathbf{p}_{i}^{\prime }$ denotes the derivative of the interpolant
at $x=x_{i}$. The diagonal matrices $\mathbb{A}_{i}$ and $\mathbb{B}_{i}$\
in (\ref{C80})\ are defined as follows 
\begin{equation}
\mathbb{A}_{i}=diag\left\{ \alpha _{i}^{1},\alpha _{i}^{2},\ldots ,\alpha
_{i}^{m}\right\} ,\ \mathbb{B}_{i}=diag\left\{ \beta _{i}^{1},\beta
_{i}^{2},\ldots ,\beta _{i}^{m}\right\} .  \label{C85}
\end{equation}
If $\mathbf{p}=\mathbf{p}\left( x\right) $ is a $C^{1}$ piecewise cubic
interpolant (see, e.g., \cite{Fritsch and Carlson 1980}, \cite{Kocic and
Milovanovic 1997}), then it will be component-wise monotone on $\left[
x_{i},x_{i+1}\right] $ \emph{iff} one of the following conditions (see \cite
{Fritsch and Carlson 1980}, \cite{Kocic and Milovanovic 1997}) is satisfied: 
\begin{equation}
\left( \alpha _{i}^{k}-1\right) ^{2}+\left( \alpha _{i}^{k}-1\right) \left(
\beta _{i}^{k}-1\right) +\left( \beta _{i}^{k}-1\right) ^{2}-3\left( \alpha
_{i}^{k}+\beta _{i}^{k}-2\right) \leq 0,  \label{C90}
\end{equation}
\begin{equation}
\alpha _{i}^{k}+\beta _{i}^{k}\leq 3,\quad \alpha _{i}^{k}\geq 0,\ \beta
_{i}^{k}\geq 0,\quad \forall i,k.  \label{C100}
\end{equation}
The region of monotonicity is shown in Figure \ref{Fritsch}b. The results of
implementing a monotone $C^{1}$ piecewise cubic interpolation when compared
with the classic cubic spline interpolation, are depicted in Figure \ref
{Fritsch}a. 
\begin{figure}[h]
\centerline{\includegraphics[width=11.50cm,height=9.6cm]{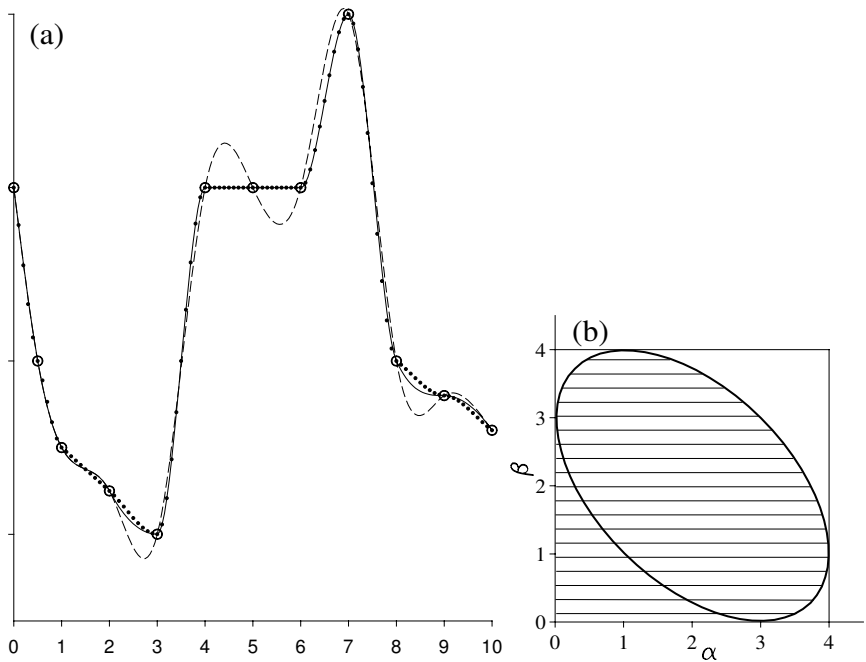}}
\caption{Monotone piecewise cubic interpolation. (a) Interpolation of a 1-D
tabulated function. Circles: prescribed tabulated values; Dashed line:
classic cubic spline; Solid line: monotone piecewise cubic under $\protect%
\alpha ,\protect\beta \leq 3$; Dotted line: monotone piecewise cubic under $%
\protect\alpha ,\protect\beta \leq 1$. (b) Necessary and sufficient
conditions for monotonicity. Horizontal hatching: region of monotonicity;
Unshaded: cubic is non-monotone. }
\label{Fritsch}
\end{figure}

We note (Figure \ref{Fritsch}a) that the constructed function produces
monotone interpolation and has a small, or even practically zero, deviation
from the classic cubic spline at some sections where the classic cubic
spline is monotone.

Let us consider the problem of monotone high-order approximation of $\mathbf{%
v}_{i+0.5}^{n}$ in (\ref{C30}). First, given the values, $\mathbf{p}_{i}$
and $\mathbf{p}_{i+1}$, of the interpolant $\mathbf{p}=\mathbf{p}\left(
x\right) $ and the approximate estimates, $\mathbf{d}_{i}$ and $\mathbf{d}%
_{i+1}$, of its derivatives, namely 
\begin{equation}
\mathbf{p}_{i}^{\prime }=\mathbf{d}_{i}+O\left( \left( \Delta x\right)
^{s}\right) ,  \label{CC110}
\end{equation}
we obtain the following interpolation formula 
\begin{equation}
\mathbf{p}_{i+0.5}=0.5\left( \mathbf{p}_{i}+\mathbf{p}_{i+1}\right) -\frac{%
\Delta x}{8}\left( \mathbf{d}_{i+1}-\mathbf{d}_{i}\right) +O\left( \left(
\Delta x\right) ^{r}\right) ,  \label{CC120}
\end{equation}
where 
\begin{equation}
r=\min \left( 4,s+1\right) .  \label{CC130}
\end{equation}
Actually, let $\mathbf{p}\left( x\right) $ be sufficiently smooth, then,
using the Taylor series expansion of the function $\mathbf{p}\left( x\right) 
$, we obtain: %\begin{equation}
%\mathbf{p}_{i+1}=\mathbf{p}_{i+0.5}+\mathbf{p}_{i+05}^{\prime }\frac{\Delta x%
%}{2}+\frac{\mathbf{p}_{i+05}^{\prime \prime }}{2}\left( \frac{\Delta x}{2}%
%\right) ^{2}+\frac{\mathbf{p}_{i+05}^{\prime \prime \prime }}{6}\left( \frac{%
%\Delta x}{2}\right) ^{3}+\ \ldots \ ,  \label{CC140}
%\end{equation}
%\begin{equation}
%\mathbf{p}_{i}=\mathbf{p}_{i+0.5}-\mathbf{p}_{i+05}^{\prime }\frac{\Delta x}{%
%2}+\frac{\mathbf{p}_{i+05}^{\prime \prime }}{2}\left( \frac{\Delta x}{2}%
%\right) ^{2}-\frac{\mathbf{p}_{i+05}^{\prime \prime \prime }}{6}\left( \frac{%
%\Delta x}{2}\right) ^{3}+\ \ldots \ .  \label{CC150}
%\end{equation}
%Combining the equalities (\ref{CC140}) and (\ref{CC150}) we obtain
\begin{equation}
\mathbf{p}_{i+1}+\mathbf{p}_{i}=2\mathbf{p}_{i+0.5}+\mathbf{p}%
_{i+05}^{\prime \prime }\left( \frac{\Delta x}{2}\right) ^{2}+\frac{\mathbf{p%
}_{i+05}^{\prime \prime \prime \prime }}{12}\left( \frac{\Delta x}{2}\right)
^{4}+O\left( \left( \Delta x\right) ^{6}\right) .  \label{CC160}
\end{equation}
%\begin{equation*}
%\mathbf{p}_{i+0.5}=0.5\left( \mathbf{p}_{i}+\mathbf{p}_{i+1}\right) -\frac{%
%\mathbf{p}_{i+05}^{\prime \prime }}{2}\left( \frac{\Delta x}{2}\right) ^{2}-%
%\frac{\mathbf{p}_{i+05}^{\prime \prime \prime \prime }}{24}\left( \frac{%
%\Delta x}{2}\right) ^{4}+O\left( \left( \Delta x\right) ^{6}\right) .
%\end{equation*}
In a similar manner, using the Taylor series expansion of the function $%
\mathbf{p}^{\prime }\left( x\right) $, we can show that %\begin{equation}
%\mathbf{p}_{i+1}^{\prime }=\mathbf{p}_{i+05}^{\prime }+\mathbf{p}%
%_{i+05}^{\prime \prime }\frac{\Delta x}{2}+\frac{\mathbf{p}_{i+05}^{\prime
%\prime \prime }}{2}\left( \frac{\Delta x}{2}\right) ^{2}+\frac{\mathbf{p}%
%_{i+05}^{\prime \prime \prime \prime }}{6}\left( \frac{\Delta x}{2}\right)
%^{3}+\ \ldots \ ,  \label{CC170}
%\end{equation}
%\begin{equation}
%\mathbf{p}_{i}^{\prime }=\mathbf{p}_{i+05}^{\prime }-\mathbf{p}%
%_{i+05}^{\prime \prime }\frac{\Delta x}{2}+\frac{\mathbf{p}_{i+05}^{\prime
%\prime \prime }}{2}\left( \frac{\Delta x}{2}\right) ^{2}-\frac{\mathbf{p}%
%_{i+05}^{\prime \prime \prime \prime }}{6}\left( \frac{\Delta x}{2}\right)
%^{3}+\ \ldots \ .  \label{CC180}
%\end{equation}
\begin{equation}
\mathbf{p}_{i+1}^{\prime }-\mathbf{p}_{i}^{\prime }=\mathbf{p}%
_{i+05}^{\prime \prime }\Delta x+\frac{\mathbf{p}_{i+05}^{\prime \prime
\prime \prime }}{3}\left( \frac{\Delta x}{2}\right) ^{3}+O\left( \left(
\Delta x\right) ^{5}\right)  \label{CC185}
\end{equation}
%Subtracting the equations (\ref{CC170}) and (\ref{CC180}), we obtain
%\begin{equation}
%\mathbf{p}_{i+05}^{\prime \prime }=\frac{\mathbf{p}_{i+1}^{\prime }-\mathbf{p%
%}_{i}^{\prime }}{\Delta x}-\frac{\mathbf{p}_{i+05}^{\prime \prime \prime
%\prime }}{6}\left( \frac{\Delta x}{2}\right) ^{2}+O\left( \left( \Delta
%x\right) ^{4}\right) .  \label{CC190}
%\end{equation}
%\begin{equation*}
%\mathbf{p}_{i+0.5}=0.5\left( \mathbf{p}_{i}+\mathbf{p}_{i+1}\right) -0.5(%
%\frac{\mathbf{p}_{i+1}^{\prime }-\mathbf{p}_{i}^{\prime }}{\Delta x}-\frac{%
%\mathbf{p}_{i+05}^{\prime \prime \prime \prime }}{6}\left( \frac{\Delta x}{2}%
%\right) ^{2}+O\left( \left( \Delta x\right) ^{4}\right) )\left( \frac{\Delta
%x}{2}\right) ^{2}-\frac{\mathbf{p}_{i+05}^{\prime \prime \prime \prime }}{24}%
%\left( \frac{\Delta x}{2}\right) ^{4}+O\left( \left( \Delta x\right)
%^{6}\right) =
%\end{equation*}
By virtue of (\ref{CC185}) we obtain from (\ref{CC160}) that 
\begin{equation}
\mathbf{p}_{i+0.5}=0.5\left( \mathbf{p}_{i}+\mathbf{p}_{i+1}\right) -\frac{%
\Delta x}{8}\left( \mathbf{p}_{i+1}^{\prime }-\mathbf{p}_{i}^{\prime
}\right) +\frac{\mathbf{p}_{i+05}^{\prime \prime \prime \prime }}{24}\left( 
\frac{\Delta x}{2}\right) ^{4}+O\left( \left( \Delta x\right) ^{6}\right)
\label{CC195}
\end{equation}
In view of (\ref{CC110}) we obtain from (\ref{CC195}) the following
interpolation formula 
\begin{equation}
\mathbf{p}_{i+0.5}=0.5\left( \mathbf{p}_{i}+\mathbf{p}_{i+1}\right) -\frac{%
\Delta x}{8}\left( \mathbf{d}_{i+1}-\mathbf{d}_{i}\right) +O\left( \left(
\Delta x\right) ^{4}+\left( \Delta x\right) ^{s+1}\right) .  \label{CC200}
\end{equation}
Hence, by virtue of (\ref{CC200}) and (\ref{CC120}), we conclude that (\ref
{CC130}) is, in general, valid. It is clear that if (\ref{CC110}) is valid
and $\mathbf{d}_{i+1}-\mathbf{d}_{i}$ $=$ $\left( \mathbf{p}_{i+1}^{\prime }-%
\mathbf{p}_{i}^{\prime }\right) $ $+$ $O(\left( \Delta x\right) ^{s+1})$,
then $r=\min \left( 4,s+2\right) $ in (\ref{CC120}). Thus, if $\mathbf{p}%
\left( x\right) $ has a continuous fourth derivative and $\mathbf{p}%
_{i}^{\prime }$ can be estimated with accuracy $O(\left( \Delta x\right)
^{3})$ (see, e.g., \cite[p. 112]{Kahaner et al. 1989}), then $r=4$ in (\ref
{CC120}). %In particular, the fourth order
%approximation is provided by using the cubic spline routines for the
%estimation $\mathbf{p}_{i}^{\prime }$, see e.g. \cite[p. 111]{Kahaner et al.
%1989}.

Let us note that instead of point values employed in the construction of the
scheme (\ref{C30}), it can be used the cell averages (see, e.g., \cite
{Balaguer and Conde 2005}, \cite{Kurganov and Tadmor 2000}, \cite{LeVeque
2002}). Then, given the values, $\overline{\mathbf{p}}_{i}$, of the averaged
interpolant $\overline{\mathbf{p}}=\overline{\mathbf{p}}\left( x\right) $
and the estimates, $\overline{\mathbf{d}}_{i}$ and $\overline{\mathbf{d}}%
_{i+1}$, of its derivatives, namely $\overline{\mathbf{p}}_{i}^{\prime }$ $=$
$\overline{\mathbf{d}}_{i}+O\left( \left( \Delta x\right) ^{s}\right) $ we
obtain the same formula (\ref{CC120}) written for the averaged values. 
%, i.e. $\overline{\mathbf{p}}_{i+0.5}$, $\overline{\mathbf{p}}_{i}$, $%
%\overline{\mathbf{p}}_{i+1}$, $\overline{\mathbf{d}}_{i}$, and $\overline{%
%\mathbf{d}}_{i+1}$ are used, respectively, instead\ of $\mathbf{p}_{i+0.5}$,
%$\mathbf{p}_{i}$, $\mathbf{p}_{i+1}$, $\mathbf{d}_{i}$, and $\mathbf{d}_{i+1}
%$.
Further, to estimate the flux at $x_{i}$, the point value $\mathbf{p}_{i}$
can be reconstructed using a piecewise polynomial interpolation (e.g., \cite
{Balaguer and Conde 2005}, \cite{Kurganov and Tadmor 2000}). Notice, we need
not any reconstruction when dealing with first- and second-order schemes,
since the point values agree with the corresponding cell averages to $%
O(\left( \Delta x\right) ^{2})$ (e.g., \cite{Kurganov and Tadmor 2000}, \cite
{LeVeque 2002}). We will therefore omit the bar notation when dealing with
such schemes.

Instead of (\ref{CC120}), we will employ the following interpolation formula 
\begin{equation}
\mathbf{p}_{i+0.5}=0.5\left( \mathbf{p}_{i}+\mathbf{p}_{i+1}\right)
-\varkappa \frac{\Delta x}{8}\left( \mathbf{d}_{i+1}-\mathbf{d}_{i}\right)
,\quad 0\leq \varkappa \leq 1,  \label{CC210}
\end{equation}
where, for some reason, a new parameter $\varkappa $ is used. Consider, for
instance, the case when the point value $\mathbf{v}_{i+0.5}^{n}$ in the
scheme (\ref{C30}) is estimated by the cell average calculated on the basis
of the monotone $C^{1}$ piecewise cubics. Note that this estimation is of
second-order. Such an approach leads to the interpolation of the point value
by formula (\ref{CC210}) with the parameter $\varkappa =2\diagup 3$.
Obviously, the interpolation (\ref{CC210}) is of second-order under $0\leq
\varkappa <1$, and (\ref{CC210}) coincides with (\ref{CC120}) under $%
\varkappa =1$ to give fourth order accuracy. Notice, if $0<\varkappa <1$,
then $\mathbf{p}_{i+0.5}$ can be estimated with accuracy $O(\left( \Delta
x\right) ^{4})$ by (\ref{CC210})\ provided that $1-\varkappa =O(\left(
\Delta x\right) ^{2})$. Actually, by virtue of (\ref{CC185}) and (\ref{CC110}%
), we obtain that 
\begin{equation*}
\varkappa \frac{\Delta x}{8}\left( \mathbf{d}_{i+1}-\mathbf{d}_{i}\right) =%
\frac{\Delta x}{8}\left( \mathbf{d}_{i+1}-\mathbf{d}_{i}\right)
+(1-\varkappa )\frac{\left( \Delta x\right) ^{2}}{8}\frac{\mathbf{d}_{i+1}-%
\mathbf{d}_{i}}{\Delta x}=\frac{\Delta x}{8}\left( \mathbf{d}_{i+1}-\mathbf{d%
}_{i}\right)
\end{equation*}
\begin{equation}
+\frac{O(\left( \Delta x\right) ^{4})}{8}\mathbf{p}_{i+05}^{\prime \prime
}+O\left( \left( \Delta x\right) ^{s+3}\right) =\frac{\Delta x}{8}\left( 
\mathbf{d}_{i+1}-\mathbf{d}_{i}\right) +O(\left( \Delta x\right) ^{4}).
\label{CC220}
\end{equation}
%Thus, in view of (\ref{CC220}), the interpolation (\ref{CC210}) is of
%fourth-order provided that $0<\varkappa \leq 1$, $s\geq 3$ in (\ref{CC110}),
%and $1-\varkappa =O(\left( \Delta x\right) ^{2})$ under $0<\varkappa <1$.

The approximation of derivatives $\mathbf{p}_{i}^{\prime }$ can be done by
the following three steps \cite{Fritsch and Carlson 1980}: (i)
initialization of the derivatives $\mathbf{p}_{i}^{\prime }=\mathbf{d}%
_{i}+O\left( \left( \Delta x\right) ^{s}\right) $; (ii) choice of subregion
of monotonicity; (iii) modification of the initialized derivatives, $\mathbf{%
d}_{i}$, to produce a monotone interpolant.

The matter of initialization of the derivatives is the most subtle issue of
this algorithm. Thus, using the two-point or the three-point (centered)
difference formula (e.g. \cite{Kocic and Milovanovic 1997}, \cite{Pareschi
Lorenzo 2001}) we obtain, in general, second-order approximation. Performing
the initialization of the derivatives $\mathbf{p}_{i}^{\prime }=\mathbf{d}%
_{i}+O\left( \left( \Delta x\right) ^{s}\right) $ in the interpolation
formula (\ref{CC210}) by the classic cubic spline \cite{Press William 1988}
interpolation, we obtain $s=3$ (e.g., \cite{Kahaner et al. 1989}, \cite
{Kocic and Milovanovic 1997}). The same accuracy can be achieved by using
the four-point approximation \cite{Kocic and Milovanovic 1997}. However, the
efficiency of the algorithm based on the classic cubic spline interpolation
is comparable with the one based on the four-point approximation, as the
number of multiplications and divisions per one node is approximately the
same for both algorithms.

Obviously, for each interval $\left[ x_{i},x_{i+1}\right] $ in which the
initialized derivatives $\mathbf{d}_{i}$, $\mathbf{d}_{i+1}$ such that at
least one point ($\alpha _{i}^{k}$, $\beta _{i}^{k}$) does not belong to the
region of monotonicity (\ref{C90})-(\ref{C100}), the derivatives $\mathbf{d}%
_{i}$, $\mathbf{d}_{i+1}$ must be modified to $\widetilde{\mathbf{d}}_{i}$, $%
\widetilde{\mathbf{d}}_{i+1}$ such that the point ($\widetilde{\alpha }%
_{i}^{k}$, $\widetilde{\beta }_{i}^{k}$) will be in the region of
monotonicity. Moreover, it is desirable to modify the derivatives $\mathbf{d}
$ $\equiv $ $\left\{ \ldots ,\mathbf{d}_{i-1}^{T},\mathbf{d}_{i}^{T},\mathbf{%
d}_{i+1}^{T},\ldots \right\} ^{T}$ in such a way that $\widetilde{\mathbf{d}}
$ $\equiv $ $\left\{ \ldots ,\widetilde{\mathbf{d}}_{i-1}^{T},\widetilde{%
\mathbf{d}}_{i}^{T},\widetilde{\mathbf{d}}_{i+1}^{T},\ldots \right\} ^{T}$
would be the solution to the following mathematical programming problem: 
\begin{equation}
\left\| \mathbf{d-}\widetilde{\mathbf{d}}\right\| \rightarrow \underset{%
\widetilde{\mathbf{d}}}{\min }.  \label{CA170}
\end{equation}
The modification of the initialized derivatives, would be much simplified if
we take a square as a subregion of monotonicity. In connection with this, we
will make use the subregions of monotonicity represented in the following
form: 
\begin{equation}
0\leq \alpha _{i}^{k}\leq 4\aleph ,\quad 0\leq \beta _{i}^{k}\leq 4\aleph
,\quad \forall i,k,  \label{CC230}
\end{equation}
where $\aleph $ is a monotonicity parameter. If $\aleph =0.75$, then the
subregion (\ref{CC230}), called de Boor-Swartz box \cite{Kocic and
Milovanovic 1997}, coincides with the one used by Fritsch and Carlson \cite
{Fritsch and Carlson 1980}. It is proven in \cite{Borisov and Mond 2010a}
that the interpolation (\ref{CC210}) will be monotone, i.e. the value of an
arbitrary component of $\mathbf{p}_{i+0.5}$ will be between the
corresponding components of $\mathbf{p}_{i}$ and $\mathbf{p}_{i+1}$, \emph{%
iff} (\ref{CC230}) will be valid provided that $0\leq $ $\aleph $ $\leq 1$.
Notice, in such a case the point ($\alpha _{i}^{k}$, $\beta _{i}^{k}$) may
be out of the region of monotonicity shown in Figure \ref{Fritsch}b.

To fulfill the conditions of monotonicity (\ref{CC230}), the modification of
derivatives $\mathbf{d}_{i}=\left\{ d_{i}^{1},d_{i}^{2},\ldots
,d_{i}^{m}\right\} $ can be done by the following algorithm suggested, in
fact, by Fritsch and Carlson \cite{Fritsch and Carlson 1980} (see also \cite
{Kocic and Milovanovic 1997}): 
\begin{equation}
S_{i}^{k}:=4\aleph \min \func{mod}(\Delta _{i-1}^{k},\Delta _{i}^{k}),\quad 
\widetilde{d}_{i}^{k}:=\min \func{mod}(d_{i}^{k},S_{i}^{k}),\quad \aleph
=const,  \label{CC240}
\end{equation}
where $\Delta _{i}^{k}=\left( p_{i+1}^{k}-p_{i}^{k}\right) \diagup \Delta x$%
, the function $\min \func{mod}(x,y)$ is defined (e.g., \cite{Kocic and
Milovanovic 1997}, \cite{Kurganov and Tadmor 2000}, \cite{Morton 2001}, \cite
{Pareschi Lorenzo 2001}, \cite{Serna and Marquina 2005}) as follows 
\begin{equation}
\min \func{mod}(x,y)\equiv \frac{1}{2}\left[ sgn(x)+sgn(y)\right] \min
\left( \left| x\right| ,\left| y\right| \right) .  \label{CC250}
\end{equation}

\subsection{Construction of nonstaggered central schemes\label{Construction}}

We will consider explicit central schemes on a uniform grid with time step $%
\Delta t$ and spatial mesh size $\Delta x$. In view of the interpolation
formula (\ref{CC210}), the staggered scheme (\ref{C30}) becomes 
\begin{equation}
\mathbf{v}_{i+0.5}^{n+0.5}=0.5\left( \mathbf{v}_{i+1}^{n}+\mathbf{v}%
_{i}^{n}\right) -\varkappa \frac{\Delta x}{8}\left( \mathbf{d}_{i+1}^{n}-%
\mathbf{d}_{i}^{n}\right) -\frac{\Delta t}{2}\frac{\mathbf{f}\left( \mathbf{v%
}_{i+1}^{n}\right) -\mathbf{f}\left( \mathbf{v}_{i}^{n}\right) }{\Delta x},
\label{CC260}
\end{equation}
where $\mathbf{d}_{i}^{n}$ denotes the derivative of the interpolant at $%
x=x_{i}$, the range of values for the parameter $\varkappa $ is the segment $%
0\leq \varkappa \leq 1$. If $\varkappa =0$, then the scheme (\ref{CC260})
coincides with the LxF scheme. It is shown in \cite{Borisov and Mond 2010a}
that the first-order, $O(\Delta t+\left( \Delta x\right) ^{r}\diagup \Delta
t+\left( \Delta x\right) ^{2})$, scheme (\ref{CC260}) generates a
conditional approximation, because it approximates (\ref{C10}) only if $%
\left( \Delta x\right) ^{r}\diagup \Delta t\rightarrow 0$ as $\Delta
x\rightarrow 0$ and $\Delta t\rightarrow 0$, where $r$\ is the order of
approximation of $\mathbf{v}_{i+0.5}^{n}$ by the interpolation formula (\ref
{CC210}). The scheme (\ref{CC260}) is abbreviated in \cite{Borisov and Mond
2010b} as COS1.

Let us develop new nonoscillatory central schemes, which are based on
regular, nonstaggered spatial grids. Using the central differencing, we
approximate (\ref{C10}) on the cell $\left[ x_{i-1},x_{i+1}\right] \times %
\left[ t_{n},t_{n+1}\right] $ by the following difference equation 
\begin{equation}
\mathbf{v}_{i}^{n+1}=\mathbf{v}_{i}^{n}-\frac{\Delta t}{2\Delta x}\left( 
\mathbf{g}_{i+1}^{n+0.5}-\mathbf{g}_{i-1}^{n+0.5}\right) .  \label{CC270}
\end{equation}
Notice, the approximation $\mathbf{g}_{i}^{n+0.5}=\mathbf{f}\left( \mathbf{v}%
_{i}^{n}\right) +O\left( \Delta t\right) $, leads to the first-order, $%
O(\Delta t+\left( \Delta x\right) ^{2})$, scheme 
\begin{equation}
\mathbf{v}_{i}^{n+1}=\mathbf{v}_{i}^{n}-\frac{\Delta t}{2\Delta x}\left( 
\mathbf{f}_{i+1}^{n}-\mathbf{f}_{i-1}^{n}\right) ,  \label{CC275}
\end{equation}
which is absolutely unstable \cite[p. 113]{Ganzha and Vorozhtsov 1996b}, 
\cite[p. 71]{LeVeque 2002}. Let us construct a stable scheme on the basis of
(\ref{CC270}). The interpolation formula (\ref{CC210}), as applied to the
cell $\left[ x_{i-1},x_{i+1}\right] \times \left[ t_{n},t_{n+1}\right] $,
becomes 
\begin{equation}
\mathbf{p}_{i}=0.5\left( \mathbf{p}_{i-1}+\mathbf{p}_{i+1}\right) -\varkappa 
\frac{\Delta x}{4}\left( \mathbf{d}_{i+1}-\mathbf{d}_{i-1}\right) .
\label{CC280}
\end{equation}
Equality (\ref{CC280}), after elementary transformations, leads to 
\begin{equation}
\mathbf{p}_{i}=0.25\left( \mathbf{p}_{i-1}+2\mathbf{p}_{i}+\mathbf{p}%
_{i+1}\right) -\varkappa \frac{\Delta x}{8}\left( \mathbf{d}_{i+1}-\mathbf{d}%
_{i-1}\right) .  \label{CC285}
\end{equation}
Using (\ref{CC285}) we obtain from (\ref{CC275}) the following nonstaggered
central scheme 
\begin{equation}
\mathbf{v}_{i}^{n+1}=\frac{\mathbf{v}_{i+1}^{n}+2\mathbf{v}_{i}^{n}+\mathbf{v%
}_{i-1}^{n}}{4}-\varkappa \frac{\Delta x}{8}\left( \mathbf{d}_{i+1}^{n}-%
\mathbf{d}_{i-1}^{n}\right) -\frac{\Delta t}{2\Delta x}\left( \mathbf{f}%
_{i+1}^{n}-\mathbf{f}_{i-1}^{n}\right) .  \label{CC290}
\end{equation}
where $0\leq \varkappa \leq 1$. In view of (\ref{CC110})-(\ref{CC130}), as
applied to the cell $\left[ x_{i-1},x_{i+1}\right] \times \left[
t_{n},t_{n+1}\right] $, the local truncation error, $\psi $, on a
sufficiently smooth solution $\mathbf{u}\left( x,t\right) $ to (\ref{C10})
is found to be 
\begin{equation}
\psi =O(\Delta t+\left( \Delta x\right) ^{2})+O\left( \frac{\left( \Delta
x\right) ^{r}}{\Delta t}\right) +(1-\varkappa )O\left( \frac{\left( \Delta
x\right) ^{2}}{\Delta t}\right) ,  \label{CC300}
\end{equation}
where $r=\min \left( 4,s+1\right) $, $s$ denotes the order of approximation
in (\ref{CC110}). Note that the scheme (\ref{CC290}) will be $O(\Delta
t+\left( \Delta x\right) ^{2})$ accurate if $r\geq 3$ as well as $%
1-\varkappa =O(\Delta x)$. The first-order nonstaggered central scheme (\ref
{CC290}), being a monotone approximation of the 1-D equation (\ref{C10}) by $%
C^{1}$ cubics, will be abbreviated to MAC1.

Let us consider the transformation from the scheme (\ref{CC275}) to (\ref
{CC290}) with closer inspection. We will use the, so called, \emph{first
differential approximation} of the scheme (\ref{CC275}) (\cite[p. 45]{Ganzha
and Vorozhtsov 1996}, \cite[p. 376]{Samarskiy and Gulin 1973}; see also
`modified equations' in \cite[p. 45]{Ganzha and Vorozhtsov 1996}, \cite
{LeVeque 2002}, \cite{Morton 1996}). As reported in \cite{Ganzha and
Vorozhtsov 1996}, \cite{Samarskiy and Gulin 1973}, this heuristic method was
originally presented by Hirt (1968) (see \cite[p. 45]{Ganzha and Vorozhtsov
1996}) as well as by Shokin and Yanenko (1968) (see \cite[p. 376]{Samarskiy
and Gulin 1973}), and has since been widely employed in the development of
stable difference schemes for PDEs.

The first differential approximation of the scheme (\ref{CC275}) is the
following: 
\begin{equation}
\frac{\partial \mathbf{u}}{\partial t}+\frac{\partial }{\partial x}\mathbf{f}%
\left( \mathbf{u}\right) =-\frac{\Delta t}{2}\frac{\partial }{\partial x}%
\left( \mathbf{A}^{2}\cdot \frac{\partial \mathbf{u}}{\partial x}\right)
,\quad \mathbf{A}^{2}\equiv \frac{\partial \mathbf{f}}{\partial \mathbf{u}}%
\cdot \frac{\partial \mathbf{f}}{\partial \mathbf{u}}.  \label{CC305}
\end{equation}
The negative diffusion term in (\ref{CC305}) is associated with instability,
as this therm is the source of energy resulting in an unlimited growth of
the amplitude of the solution (see, e.g., \cite{Strunin 2003}, \cite{Zhao
and Tang 2000}). Notice, for the sake of convenience we use the same
notation in (\ref{C10}) and in (\ref{CC305}) in spite of the fact that these
equations are different. We will use such an approach if it does not lead to
confusion.

To obtain (\ref{CC290}), we, in fact, added to the right-hand side of (\ref
{CC275}) the value 
\begin{equation}
\mathbf{E}_{i}\equiv \frac{\mathbf{v}_{i+1}^{n}-2\mathbf{v}_{i}^{n}+\mathbf{v%
}_{i-1}^{n}}{4}-\varkappa \frac{\Delta x}{8}\left( \mathbf{d}_{i+1}^{n}-%
\mathbf{d}_{i-1}^{n}\right) .  \label{CC310}
\end{equation}
Let the interpolant $\mathbf{p}=\mathbf{p}\left( x\right) $ be such that $%
\mathbf{p}_{i}=\mathbf{v}_{i}^{n}$. If $\varkappa =1$ and $\mathbf{d}_{i}=%
\mathbf{p}_{i}^{\prime }$, then, by virtue of (\ref{CC110})-(\ref{CC195}) as
applied to the cell $\left[ x_{i-1},x_{i+1}\right] \times \left[
t_{n},t_{n+1}\right] $, we find that 
\begin{equation}
\mathbf{E}_{i}=-\frac{\left( \Delta x\right) ^{4}}{48}\mathbf{p}_{i}^{\prime
\prime \prime \prime }+O\left( \left( \Delta x\right) ^{6}\right) .
\label{CC320}
\end{equation}
Thus, when we add $\mathbf{E}_{i}$ of (\ref{CC310}) to the right-hand side
of (\ref{CC275}), we, in fact, add the negative fourth-order diffusion term
of (\ref{CC320}) to the right-hand side of the first differential
approximation (\ref{CC305}). Such a term, in contrast to a negative
second-order diffusion term, stabilizes the system and produces dissipative
effect (see, e.g., \cite{Strunin 2003}, \cite{Zhao and Tang 2000}).

If $\varkappa =1$ and $\mathbf{p}_{i}^{\prime }$ is estimated by the
central-difference derivative, i.e. $\mathbf{d}_{i}=(\mathbf{p}_{i+1}-%
\mathbf{p}_{i-1})\diagup (2\Delta x)$, then 
\begin{equation}
\mathbf{p}_{i}^{\prime }=\mathbf{d}_{i}-\frac{\left( \Delta x\right) ^{2}}{6}%
\mathbf{p}_{i}^{\prime \prime \prime }+O\left( \left( \Delta x\right)
^{4}\right) .  \label{CC330}
\end{equation}
In such a case we obtain, instead of (\ref{CC320}), that 
\begin{equation}
\mathbf{E}_{i}=-\frac{3\left( \Delta x\right) ^{4}}{48}\mathbf{p}%
_{i}^{\prime \prime \prime \prime }+O\left( \left( \Delta x\right)
^{5}\right) .  \label{CC340}
\end{equation}
We can see from (\ref{CC340}) that the estimation of the derivative $\mathbf{%
p}_{i}^{\prime }$ by the central-difference derivative leads to the same
order of accuracy as it was in the previous example when $\mathbf{d}_{i}=%
\mathbf{p}_{i}^{\prime }$. Furthermore, the scheme (\ref{CC290}) will be
more dissipative in such a case, since the coefficient of the fourth-order
diffusion term in (\ref{CC340}) is three times bigger than the similar
coefficient in (\ref{CC320}).

If $0\leq \varkappa \leq 1$ and $\mathbf{p}_{i}^{\prime }$ is estimated by
the central-difference derivative, then we find 
\begin{equation}
\mathbf{E}_{i}=\left( 1-\varkappa \right) \frac{\left( \Delta x\right) ^{2}}{%
4}\mathbf{p}_{i}^{\prime \prime }+\left( 1-4\varkappa \right) \frac{\left(
\Delta x\right) ^{4}}{48}\mathbf{p}_{i}^{\prime \prime \prime \prime
}+O\left( \left( \Delta x\right) ^{5}\right) .  \label{CC350}
\end{equation}

Note that the estimations (\ref{CC320}), (\ref{CC340}), and (\ref{CC350})
will be valid if the derivative $\mathbf{d}_{i}$ need not be modified, i.e. (%
\ref{CC230}) will be valid without implementing the algorithm (\ref{CC240}).
Otherwise, there is a risk to add negative second-order diffusion terms to
the right-hand side of (\ref{CC305}). Let us consider, for instance, the
case when $\varkappa =1$ and the derivative $\mathbf{p}_{i+1}^{\prime }$ is
estimated by the right-difference derivative, i.e. $\mathbf{d}_{i+1}=(%
\mathbf{p}_{i+2}-\mathbf{p}_{i+1})\diagup \Delta x$, however $\mathbf{p}%
_{i-1}^{\prime }$ is estimated by the central-difference derivative. In such
a case we obtain: 
\begin{equation}
\mathbf{E}_{i}=-\frac{\left( \Delta x\right) ^{2}}{16}\mathbf{p}%
_{i+1}^{\prime \prime }-\frac{1.25\left( \Delta x\right) ^{4}}{16}\mathbf{p}%
_{i+1}^{\prime \prime \prime \prime }+O\left( \left( \Delta x\right)
^{5}\right) .  \label{CC360}
\end{equation}
In view of (\ref{CC360}) and (\ref{CC350}), we conclude that the part played
by the parameter $\varkappa $ in suppressing spurious oscillations produced
by negative second-order diffusion terms can be very important.

Interestingly, there is a possibility to improve the scheme (\ref{CC260}) by
introducing an additional positive numerical viscosity, i.e., using the
vanishing viscosity method \cite{Godlewski and Raviart 1996}, \cite{LeVeque
2002}, such that the scheme's order of accuracy would increase up to $%
O((\Delta t)^{2}+\left( \Delta x\right) ^{2})$. Such an approach \cite
{Borisov and Mond 2010b} leads to the following second order central scheme 
\begin{equation*}
\mathbf{v}_{i+0.5}^{n+0.5}=0.5\left( \mathbf{v}_{i+1}^{n}+\mathbf{v}%
_{i}^{n}\right) -\varkappa \frac{\Delta x}{8}\left( \mathbf{d}_{i+1}^{n}-%
\mathbf{d}_{i}^{n}\right) -\frac{\Delta t}{2}\frac{\mathbf{f}\left( \mathbf{v%
}_{i+1}^{n}\right) -\mathbf{f}\left( \mathbf{v}_{i}^{n}\right) }{\Delta x}
\end{equation*}
\begin{equation}
+\frac{\left( \Delta t\right) ^{2}}{8\Delta x}\left[ \left( \mathbf{A}%
_{i+1}^{n}\right) ^{2}\cdot \mathbf{d}_{i+1}^{n}-\left( \mathbf{A}%
_{i}^{n}\right) ^{2}\cdot \mathbf{d}_{i}^{n}\right] ,\quad \mathbf{A\equiv }%
\frac{\partial \mathbf{f}}{\partial \mathbf{u}},  \label{SA30}
\end{equation}
where $\mathbf{d}_{i}^{n}$ is the derivative of the interpolant at $x=x_{i}$%
. The scheme (\ref{SA30}) is abbreviated in \cite{Borisov and Mond 2010b} as
COS2.

Let us develop a nonstaggered central scheme that will be of second order.
Using the central differencing, we approximate (\ref{C10}) at the point $%
x=x_{i}$, $t=t_{n+05}$ by the following equation: 
\begin{equation}
\mathbf{v}_{i}^{n+1}=\mathbf{v}_{i}^{n}-\Delta t\left. \frac{\partial 
\mathbf{f}}{\partial x}\right| _{x=x_{i}\text{, }t=t_{n+05}}.  \label{CC380}
\end{equation}
Using Taylor series expansion, we approximate $\mathbf{g}_{i}^{n+0.5}$ in (%
\ref{CC270}) with the accuracy $O(\left( \Delta t\right) ^{2})$: 
\begin{equation}
\mathbf{g}_{i}^{n+0.5}=\mathbf{f}\left( \mathbf{v}_{i}^{n}\right) +\left. 
\frac{\partial \mathbf{f}\left( \mathbf{v}_{i}^{n}\right) }{\partial t}%
\right| _{t=t_{n}}\frac{\Delta t}{2}+O\left( \Delta t^{2}\right) .
\label{SA10}
\end{equation}
By virtue of the PDE system, (\ref{C10}), we find 
\begin{equation}
\frac{\partial \mathbf{f}}{\partial t}=\frac{\partial \mathbf{f}}{\partial 
\mathbf{u}}\cdot \frac{\partial \mathbf{u}}{\partial t}=-\frac{\partial 
\mathbf{f}}{\partial \mathbf{u}}\cdot \frac{\partial \mathbf{f}}{\partial 
\mathbf{u}}\cdot \frac{\partial \mathbf{u}}{\partial x}=-\left( \frac{%
\partial \mathbf{f}}{\partial \mathbf{u}}\right) ^{2}\cdot \frac{\partial 
\mathbf{u}}{\partial x}.  \label{SA20}
\end{equation}
By virtue of (\ref{CC285}) and (\ref{SA10})-(\ref{SA20}), we find 
\begin{equation*}
\mathbf{v}_{i}^{n+1}=\frac{\mathbf{v}_{i-1}^{n}+2\mathbf{v}_{i}^{n}+\mathbf{v%
}_{i+1}^{n}}{4}-\varkappa \frac{\Delta x}{8}\left( \mathbf{d}_{i+1}^{n}-%
\mathbf{d}_{i-1}^{n}\right) -\Delta t\left. \frac{\partial \mathbf{f}}{%
\partial x}\right| _{x=x_{i}\text{, }t=t_{n}}
\end{equation*}
\begin{equation}
+\frac{\left( \Delta t\right) ^{2}}{2}\left. \frac{\partial }{\partial x}%
\left( \mathbf{A}^{2}\frac{\partial \mathbf{u}}{\partial x}\right) \right|
_{x=x_{i}\text{, }t=t_{n}},\quad \mathbf{A}^{2}\equiv \frac{\partial \mathbf{%
f}}{\partial \mathbf{u}}\cdot \frac{\partial \mathbf{f}}{\partial \mathbf{u}}%
.  \label{CC390}
\end{equation}
Then, approximating the last two terms in the right-hand side of (\ref{CC390}%
) with the second order accuracy, we obtain 
\begin{equation*}
\mathbf{v}_{i}^{n+1}=\frac{\mathbf{v}_{i-1}^{n}+2\mathbf{v}_{i}^{n}+\mathbf{v%
}_{i+1}^{n}}{4}-\varkappa \frac{\Delta x}{8}\left( \mathbf{d}_{i+1}^{n}-%
\mathbf{d}_{i-1}^{n}\right) -\frac{\Delta t}{2\Delta x}\left( \mathbf{f}%
_{i+1}^{n}-\mathbf{f}_{i-1}^{n}\right)
\end{equation*}
\begin{equation}
+\frac{\varsigma \left( \Delta t\right) ^{2}}{2\left( \Delta x\right) ^{2}}%
\left[ \mathbf{D}_{i+0.5}^{n}\cdot \left( \mathbf{v}_{i+1}^{n}-\mathbf{v}%
_{i}^{n}\right) -\mathbf{D}_{i-0.5}^{n}\cdot \left( \mathbf{v}_{i}^{n}-%
\mathbf{v}_{i-1}^{n}\right) \right] ,\ \varsigma =const\geq 0,  \label{CC400}
\end{equation}
where $\mathbf{D}_{i+0.5}^{n}=0.5\left[ \left( \mathbf{A}_{i}^{n}\right)
^{2}+\left( \mathbf{A}_{i+1}^{n}\right) ^{2}\right] $, the parameter $%
\varsigma $ is introduced by analogy with $\varkappa $ in (\ref{CC260}).
Scheme (\ref{CC400}) coincides with the scheme MAC1, (\ref{CC290}), provided
that $\varsigma =0$. The last term in right-hand side of (\ref{CC400}) can
be seen as the non-negative numerical viscosity introduced into the first
order scheme (\ref{CC290}). Thus, we are dealing with the vanishing
viscosity method \cite{Godlewski and Raviart 1996}, \cite{LeVeque 2002} and,
hence, in view of \cite[Theorem 3.3]{Godlewski and Raviart 1996}, the
scheme, (\ref{CC400}), satisfies the entropy condition. To make this point a
little bit more clear, we consider a viscous perturbation of the system (\ref
{C10}). We associate with (\ref{C10}) the following parabolic system 
\begin{equation}
\frac{\partial \mathbf{u}_{\epsilon }}{\partial t}+\frac{\partial }{\partial
x}\mathbf{f}\left( \mathbf{u}_{\epsilon }\right) =\epsilon \frac{\partial }{%
\partial x}\left( \mathbf{A}^{2}\cdot \frac{\partial \mathbf{u}_{\epsilon }}{%
\partial x}\right) ,\quad \epsilon >0,  \label{SA40}
\end{equation}
where the right-hand side can be viewed as a viscosity term. It is assumed
that Mock's assumption of admissibility \cite[p. 32]{Godlewski and Raviart
1996} is valid, i.e. the matrix $U^{\prime \prime }\cdot \mathbf{A}^{2}$ is
positive-definite. Here $U=U\left( \mathbf{u}_{\epsilon }\right) $ denotes a
strictly convex entropy, $U^{\prime \prime }$ denotes the Hessian matrix of $%
U$. We recover solutions of (\ref{C10}) as the limits of solutions to (\ref
{SA40}) as $\epsilon \rightarrow 0$. If we divide the scheme COS1, (\ref
{CC290}), by $\Delta t$ and group all the terms of this scheme together in
its left-hand side, then this group, in view of the above, is the finite
difference approximation of the left-hand side of Eq. (\ref{SA40}) on the
cell $\left[ x_{i-1},x_{i+1}\right] \times \left[ t_{n},t_{n+1}\right] $.
The right-hand side of Eq. (\ref{SA40}) is approximated as following 
\begin{equation}
\epsilon \frac{\partial }{\partial x}\left( \mathbf{A}^{2}\cdot \frac{%
\partial \mathbf{u}_{\epsilon }}{\partial x}\right) \approx \epsilon \frac{%
\mathbf{D}_{i+0.5}^{n}\cdot \left( \mathbf{v}_{i+1}^{n}-\mathbf{v}%
_{i}^{n}\right) -\mathbf{D}_{i-0.5}^{n}\cdot \left( \mathbf{v}_{i}^{n}-%
\mathbf{v}_{i-1}^{n}\right) }{\left( \Delta x\right) ^{2}},  \label{SA42}
\end{equation}
where $\mathbf{D}_{i\pm 0.5}^{n}=0.5\left[ \left( \mathbf{A}_{i}^{n}\right)
^{2}+\left( \mathbf{A}_{i\pm 1}^{n}\right) ^{2}\right] $. If $\epsilon
=\varsigma \Delta t\diagup 2$, then we obtain the scheme (\ref{CC400}) as
the approximation of Eq. (\ref{SA40}) on the cell $\left[ x_{i-1},x_{i+1}%
\right] \times \left[ t_{n},t_{n+1}\right] $. Thus, owing to the last term
in the right-hand side of (\ref{CC400}), the grid function $\mathbf{v}%
_{i}^{n}$ in (\ref{CC400}) can be considered as Lipschitz-continuous.
Moreover, owing to this term, the scheme (\ref{CC400}) is $O(\left( \Delta
x\right) ^{2}+\left( \Delta t\right) ^{2})$ accurate provided that $%
\varkappa =\varsigma =1$. The nonstaggered central scheme (\ref{CC400}),
being a monotone approximation of the 1-D equation (\ref{C10}) with the
second order, will be abbreviated to MAC2.

\subsection{Stability of the developed schemes\label%
{Stability of the developed schemes}}

It is proven in \cite{Borisov and Mond 2010b} that the scheme COS2, (\ref
{SA30}), will be, in general, stable if 
\begin{equation}
\left| \varkappa -\varsigma C_{r}^{2}\right| \aleph +C_{r}\leq 1,\quad C_{r}=%
\frac{\Delta t\lambda _{\max }}{\Delta x}.  \label{SSA10}
\end{equation}
Let us note that the stability condition (\ref{SSA10}) is not exactly
correct, namely, there exist infrequent situations when (\ref{SSA10}) is
valid nevertheless the scheme is not stable provided $\varkappa ,\varsigma
,\aleph =const$. In particular, there could be a grid node $(x_{\ast
},t_{\ast \ast })$ where the following inequality must be, but not valid: 
\begin{equation}
\varsigma \aleph C_{r}^{2}+C_{r}\leq 1,\quad C_{r}=\frac{\Delta t\lambda
_{\max }}{\Delta x}.  \label{SSA20}
\end{equation}
However, there is a possibility to take the parameters $\varkappa =\varkappa
(x,t),$ $\varsigma =\varsigma (x,t),$ and $\aleph =\aleph (x,t)$ such that
the scheme COS2, (\ref{SA30}), will be stable under (\ref{SSA10}). For
instance, if we take $\varkappa ,\varsigma ,C_{r}=1$, $\aleph =0.5$ at $%
(x,t)\neq (x_{\ast },t_{\ast \ast })$ and $\aleph =0$ at $(x,t)=(x_{\ast
},t_{\ast \ast })$, then (\ref{SSA20}) will be also valid.

In this subsection we will use analogous technic \cite{Borisov and Mond
2010b} to prove the stability of the developed nonstaggered central schemes.
In view of Theorem \ref{iffStab}, the stability of the scheme MAC2, (\ref
{CC400}), will be investigated on the basis of its variational scheme.\ It
is assumed that the bounded operator $\mathbf{A}$ $(=\partial \mathbf{f}%
\left( \mathbf{u}\right) \diagup \partial \mathbf{u)}$ in (\ref{C10}) is
Fr\'{e}chet-differentiable on the set $\Omega _{\mathbf{u}}$ $\subset $ $%
\mathbb{R}^{M}$, and its derivative is bounded on $\Omega _{\mathbf{u}}$.
Considering that $\mathbf{v}_{i}^{n}$ in (\ref{SA30}) is
Lipschitz-continuous, we write 
\begin{equation}
\left\| \mathbf{v}_{i+1}^{n}-\mathbf{v}_{i}^{n}\right\| _{2}\leq C_{v}\Delta
x,\quad C_{v}=const.  \label{SB40}
\end{equation}
By virtue of (\ref{C80}), the second term in right-hand side of (\ref{CC400}%
) can be written in the form 
\begin{equation}
\varkappa \frac{\Delta x}{8}\left( \mathbf{d}_{i+1}^{n}-\mathbf{d}%
_{i-1}^{n}\right) =\frac{\varkappa }{8}\left[ \mathbb{B}_{i}^{n}\cdot \left( 
\mathbf{v}_{i+1}^{n}-\mathbf{v}_{i}^{n}\right) -\mathbb{A}_{i-1}^{n}\cdot
\left( \mathbf{v}_{i}^{n}-\mathbf{v}_{i-1}^{n}\right) \right] .  \label{SB10}
\end{equation}
Then, the variational scheme corresponding to (\ref{CC400}) is the following 
\begin{equation*}
\delta \mathbf{v}_{i}^{n+1}-0.25\left( \delta \mathbf{v}_{i+1}^{n}+2\delta 
\mathbf{v}_{i}^{n}+\delta \mathbf{v}_{i-1}^{n}\right)
\end{equation*}
\begin{equation*}
+\frac{\Delta t}{2\Delta x}\left( \mathbf{A}_{i+1}^{n}\cdot \delta \mathbf{v}%
_{i+1}^{n}-\mathbf{A}_{i-1}^{n}\cdot \delta \mathbf{v}_{i-1}^{n}\right)
\end{equation*}
\begin{equation*}
+\frac{\varkappa }{8}\left[ \mathbb{B}_{i}^{n}\cdot \left( \delta \mathbf{v}%
_{i+1}^{n}-\delta \mathbf{v}_{i}^{n}\right) -\mathbb{A}_{i-1}^{n}\cdot
\left( \delta \mathbf{v}_{i}^{n}-\delta \mathbf{v}_{i-1}^{n}\right) \right]
\end{equation*}
\begin{equation*}
-\varsigma \frac{\left( \Delta t\right) ^{2}}{2\Delta x^{2}}\left[ \mathbf{D}%
_{i+0.5}^{n}\cdot \left( \delta \mathbf{v}_{i+1}^{n}-\delta \mathbf{v}%
_{i}^{n}\right) -\mathbf{D}_{i-0.5}^{n}\cdot \left( \delta \mathbf{v}%
_{i}^{n}-\delta \mathbf{v}_{i-1}^{n}\right) \right] =
\end{equation*}
\begin{equation*}
+\varsigma \frac{\left( \Delta t\right) ^{2}}{2\Delta x^{2}}\left[ \left(
\delta \mathbf{D}_{i+0.5}^{n}\right) \cdot \left( \mathbf{v}_{i+1}^{n}-%
\mathbf{v}_{i}^{n}\right) -\left( \delta \mathbf{D}_{i-0.5}^{n}\right) \cdot
\left( \mathbf{v}_{i}^{n}-\mathbf{v}_{i-1}^{n}\right) \right]
\end{equation*}
\begin{equation}
-\frac{\varkappa }{8}\left[ \left( \delta \mathbb{B}_{i}^{n}\right) \cdot
\left( \mathbf{v}_{i+1}^{n}-\mathbf{v}_{i}^{n}\right) -\left( \delta \mathbb{%
A}_{i-1}^{n}\right) \cdot \left( \mathbf{v}_{i}^{n}-\mathbf{v}%
_{i-1}^{n}\right) \right] .  \label{SSB20}
\end{equation}
By virtue of (\ref{CC230}) and (\ref{IN80}), we find 
\begin{equation}
\left\| \delta \mathbb{A}_{i}^{n}\right\| _{2},\left\| \delta \mathbb{B}%
_{i}^{n}\right\| _{2}\leq 4\aleph ,\quad \left\| \left( \mathbf{A}%
_{i}^{n}\right) ^{2}\right\| _{2},\left\| \left( \mathbf{A}_{i+1}^{n}\right)
^{2}\right\| _{2}\leq \lambda _{\max }^{2}.  \label{SSB30}
\end{equation}
Thus, we may write that 
\begin{equation}
\left\| \delta \mathbf{D}_{i+0.5}^{n}\right\| _{2},\left\| \delta \left[ 
\mathbf{D}_{i-0.5}^{n}\right] \right\| _{2}\leq 2\lambda _{\max }^{2}.
\label{SSB40}
\end{equation}
Then, by virtue of (\ref{CA05}), (\ref{SB40}), (\ref{SSB30})-(\ref{SSB40}),
and since $0\leq $ $\varkappa ,\varsigma ,\aleph $ $\leq 1$, we find the
following estimate for the right-hand side of (\ref{SSB20}): 
\begin{equation*}
\varsigma \frac{\left( \Delta t\right) ^{2}}{2\Delta x^{2}}\left\| \left(
\delta \mathbf{D}_{i+0.5}^{n}\right) \cdot \left( \mathbf{v}_{i+1}^{n}-%
\mathbf{v}_{i}^{n}\right) -\left( \delta \mathbf{D}_{i-0.5}^{n}\right) \cdot
\left( \mathbf{v}_{i}^{n}-\mathbf{v}_{i-1}^{n}\right) \right\| _{2}
\end{equation*}
\begin{equation*}
+\left\| \frac{\varkappa }{8}\left[ \left( \delta \mathbb{B}_{i}^{n}\right)
\cdot \left( \mathbf{v}_{i+1}^{n}-\mathbf{v}_{i}^{n}\right) -\left( \delta 
\mathbb{A}_{i-1}^{n}\right) \cdot \left( \mathbf{v}_{i}^{n}-\mathbf{v}%
_{i-1}^{n}\right) \right] \right\| _{2}\leq
\end{equation*}
\begin{equation*}
\varsigma \frac{\left( \Delta t\right) ^{2}}{2\Delta x^{2}}\left( \left\|
\delta \mathbf{D}_{i+0.5}^{n}\right\| _{2}\left\| \mathbf{v}_{i+1}^{n}-%
\mathbf{v}_{i}^{n}\right\| _{2}+\left\| \delta \mathbf{D}_{i-0.5}^{n}\right%
\| _{2}\cdot \left\| \mathbf{v}_{i}^{n}-\mathbf{v}_{i-1}^{n}\right\|
_{2}\right)
\end{equation*}
\begin{equation*}
+\frac{1}{8}\left( \left\| \delta \mathbb{B}_{i}^{n}\right\| _{2}\cdot
\left\| \mathbf{v}_{i+1}^{n}-\mathbf{v}_{i}^{n}\right\| _{2}+\left\| \delta 
\mathbb{A}_{i-1}^{n}\right\| _{2}\cdot \left\| \mathbf{v}_{i}^{n}-\mathbf{v}%
_{i-1}^{n}\right\| _{2}\right) \leq
\end{equation*}
\begin{equation}
\left( 1+2C_{r}^{2}\right) C_{v}\mu _{0}\Delta t,\quad C_{r}=\frac{\Delta
t\lambda _{\max }}{\Delta x}.  \label{SSB50}
\end{equation}
Since the uniform stability with respect to the initial data implies the
stability of scheme \cite[pp. 390-392]{Samarskii 2001} (see also \cite[Sec.
5]{Borisov and Mond 2010b}), we conclude, in view of (\ref{SSB50}), that the
scheme (\ref{SSB20}) will be stable if the following scheme (i.e. (\ref
{SSB20}) without the right-hand side) will be stable 
\begin{equation*}
\delta \mathbf{v}_{i}^{n+1}=0.25\left( \delta \mathbf{v}_{i+1}^{n}+2\delta 
\mathbf{v}_{i}^{n}+\delta \mathbf{v}_{i-1}^{n}\right) -\frac{\Delta t}{%
2\Delta x}\left( \mathbf{A}_{i+1}^{n}\cdot \delta \mathbf{v}_{i+1}^{n}-%
\mathbf{A}_{i-1}^{n}\cdot \delta \mathbf{v}_{i-1}^{n}\right)
\end{equation*}
\begin{equation*}
+\varsigma \frac{\left( \Delta t\right) ^{2}}{2\Delta x^{2}}\left[ \mathbf{D}%
_{i+0.5}^{n}\cdot \left( \delta \mathbf{v}_{i+1}^{n}-\delta \mathbf{v}%
_{i}^{n}\right) -\mathbf{D}_{i-0.5}^{n}\cdot \left( \delta \mathbf{v}%
_{i}^{n}-\delta \mathbf{v}_{i-1}^{n}\right) \right]
\end{equation*}
\begin{equation}
-\frac{\varkappa }{8}\left[ \mathbb{B}_{i}^{n}\cdot \left( \delta \mathbf{v}%
_{i+1}^{n}-\delta \mathbf{v}_{i}^{n}\right) -\mathbb{A}_{i-1}^{n}\cdot
\left( \delta \mathbf{v}_{i}^{n}-\delta \mathbf{v}_{i-1}^{n}\right) \right] .
\label{SSB60}
\end{equation}
Since the bounded operator $\mathbf{A}$ $(=\partial \mathbf{f}\left( \mathbf{%
u}\right) \diagup \partial \mathbf{u)}$ in (\ref{C10}) is
Fr\'{e}chet-differentiable on the set $\Omega _{\mathbf{u}}$ $\subset $ $%
\mathbb{R}^{M}$, and its derivative is bounded on $\Omega _{\mathbf{u}}$,
the operator $\mathbf{A}$ is Lipschitz. In view of (\ref{SB40}), we can
write that $\left\| (\Delta t\diagup \Delta x)\left( \mathbf{A}_{i-1}^{n}-%
\mathbf{A}_{i+1}^{n}\right) \cdot \delta \mathbf{v}_{i}^{n}\right\| \leq
C_{a}\Delta t\left\| \delta \mathbf{v}_{i}^{n}\right\| $, $C_{a}=const$.
Thus, adding the term $(\Delta t\diagup \Delta x)\left( \mathbf{A}_{i-1}^{n}-%
\mathbf{A}_{i+1}^{n}\right) \cdot \delta \mathbf{v}_{i}^{n}$ to the
right-hand side of (\ref{SSB60}) has no effect on the stability of (\ref
{SSB60}). To investigate the stability of (\ref{SSB60}) we rewrite it in the
following form, where the term $0.5(\Delta t\diagup \Delta x)\left( \mathbf{A%
}_{i-1}^{n}-\mathbf{A}_{i+1}^{n}\right) \cdot \delta \mathbf{v}_{i}^{n}$ is
added. 
\begin{equation*}
\delta \mathbf{v}_{i}^{n+1}=\left( 0.25\mathbf{I-E}_{i,i-1}^{n}\right) \cdot
\delta \mathbf{v}_{i-1}^{n}
\end{equation*}
\begin{equation}
+\left( 0.5\mathbf{I}+\mathbf{E}_{i,i-1}^{n}+\mathbf{E}_{i,i+1}^{n}\right)
\cdot \delta \mathbf{v}_{i}^{n}+\left( 0.25\mathbf{I-E}_{i,i+1}^{n}\right)
\cdot \delta \mathbf{v}_{i+1}^{n},  \label{SSB80}
\end{equation}
where 
\begin{equation}
\mathbf{E}_{i,i-1}^{n}=\frac{\varkappa }{8}\mathbb{A}_{i-1}^{n}-\varsigma 
\frac{\left( \Delta t\right) ^{2}}{2\Delta x^{2}}\mathbf{D}_{i-0.5}^{n}-%
\frac{\Delta t}{2\Delta x}\mathbf{A}_{i-1}^{n},  \label{SSB90}
\end{equation}
\begin{equation}
\mathbf{E}_{i,i+1}^{n}=\frac{\varkappa }{8}\mathbb{B}_{i}^{n}-\varsigma 
\frac{\left( \Delta t\right) ^{2}}{2\Delta x^{2}}\mathbf{D}_{i+0.5}^{n}+%
\frac{\Delta t}{2\Delta x}\mathbf{A}_{i+1}^{n}.  \label{SSB100}
\end{equation}
We write, in view of (\ref{IN80}) and (\ref{CC230}), that the spectrum $%
s\left( \mathbf{E}_{i,j}^{n}\right) \subseteq \left[ L_{\min ,i}^{n},L_{\max
,i}^{n}\right] $, where $j=i\pm 1$ and 
\begin{equation}
L_{\min ,i}^{n}\geq -0.5\varsigma C_{r}^{2}-0.5C_{r},\ L_{\max ,i}^{n}\leq
0.5\varkappa \aleph -0.25\varsigma C_{r}^{2}+0.5C_{r}.  \label{SSB110}
\end{equation}
Hence, by virtue of \cite[Theorem 2.10]{Borisov and Mond 2010a} we find that
the scheme (\ref{SSB80}) will be stable if 
\begin{equation}
\underset{\lambda \in \left[ L_{\min ,i}^{n},L_{\max ,i}^{n}\right] }{\max }%
\left( \left| 0.25-{\lambda }\right| +\left| 0.25+{\lambda }\right| \right)
\leqslant 0.5,\quad \forall i,n.  \label{SSB120}
\end{equation}
We obtain from (\ref{SSB120}) the following sufficient conditions for the
stability of the variational scheme (\ref{SSB20}) 
\begin{equation}
-0.5\varsigma C_{r}^{2}-0.5C_{r}\geq -0.25,  \label{SSB125}
\end{equation}
\begin{equation}
0.5\varkappa \aleph -0.25\varsigma C_{r}^{2}+0.5C_{r}\leq 0.25,\quad C_{r}=%
\frac{\Delta t\lambda _{\max }}{\Delta x}.  \label{SSB130}
\end{equation}
Let us note that the value of $L_{\max ,i}^{n}$ in (\ref{SSB110}) can rarely
be attained. It could be the case at the point $(x_{i},t_{n})$ only if $%
\left\| \mathbb{B}_{i}^{n}\right\| _{2}=4\aleph $, $\left\| \mathbf{A}%
_{i+1}^{n}\right\| _{2}=\lambda _{\max }$, and $\left\| \mathbf{A}%
_{i}^{n}\right\| _{2}=0$ (or, what is the same, $\left\| \mathbb{A}%
_{i-1}^{n}\right\| _{2}=4\aleph $, $\left\| \mathbf{A}_{i-1}^{n}\right\|
_{2}=\lambda _{\max }$, and $\left\| \mathbf{A}_{i}^{n}\right\| _{2}=0$).
Moreover, considering the modification of the scheme (\ref{SSB80}), where $%
\mathbf{D}_{i-0.5}^{n}$ is replaced by $(\mathbf{A}_{i-1}^{n})^{2}$, namely 
\begin{equation}
\mathbf{E}_{i,i-1}^{n}=\frac{\varkappa }{8}\mathbb{A}_{i-1}^{n}-\varsigma 
\frac{\left( \Delta t\right) ^{2}}{2\Delta x^{2}}(\mathbf{A}_{i-1}^{n})^{2}-%
\frac{\Delta t}{2\Delta x}\mathbf{A}_{i-1}^{n},  \label{SSB90a}
\end{equation}
\begin{equation}
\mathbf{E}_{i,i+1}^{n}=\frac{\varkappa }{8}\mathbb{B}_{i}^{n}-\varsigma 
\frac{\left( \Delta t\right) ^{2}}{2\Delta x^{2}}(\mathbf{A}_{i+1}^{n})^{2}+%
\frac{\Delta t}{2\Delta x}\mathbf{A}_{i+1}^{n},  \label{SSB100a}
\end{equation}
we note (see Proof in \cite[Theorem 2.10]{Borisov and Mond 2010a}) that the
stability of the modified scheme, (\ref{SSB80}), (\ref{SSB90a})-(\ref
{SSB100a}), implies the stability of the scheme (\ref{SSB80})-(\ref{SSB100}%
), since the operator $\mathbf{A}$ is Lipschitz. Hence, instead of (\ref
{SSB130}), we will consider less rigid requirement: 
\begin{equation}
0.5\varkappa \aleph -0.5\varsigma C_{r}^{2}+0.5C_{r}\leq 0.25.
\label{SSB131}
\end{equation}
Let $\varsigma =1$, then, by virtue of (\ref{SSB125}) and (\ref{SSB131}), we
find the stability condition for the variational scheme (\ref{SSB20}): 
\begin{equation}
C_{r}\leq 0.5\left( \sqrt{3}-1\right) ,\quad \varkappa \aleph \leq 2-\sqrt{3}%
.  \label{SSB132}
\end{equation}
Thus, in view of Theorem \ref{iffStab}, the scheme MAC2, (\ref{CC400}), will
be stable if (\ref{SSB132}) will be valid.

Notice, if $\varsigma \rightarrow 0$ in (\ref{SSB125}) and (\ref{SSB130}),
then we obtain the stability condition for the scheme MAC1, (\ref{CC290}): 
\begin{equation}
\varkappa \aleph +C_{r}\leq 0.5,\quad C_{r}=\frac{\Delta t\lambda _{\max }}{%
\Delta x}.  \label{SSB135}
\end{equation}

It is significant that the stability conditions (\ref{SSB125}), (\ref{SSB131}%
), and hence (\ref{SSB132}), is found by virtue of \cite[Theorem 2.10]
{Borisov and Mond 2010a} and, therefore, it is assumed that $\mathbf{E}%
_{i,j}^{n}$ is Lipschitz, i.e. $\mathbf{E}_{i,i+1}^{n}=\mathbf{E}%
_{i,i-1}^{n}+O(\Delta x)$. For the sake of brevity, we associate the Landau
symbol $O(\Delta x)$ with Lipschitz-continuity of a grid function whose
increment can be estimated in norm by a constant times $\Delta x$ for $%
\Delta x$ small enough. Note that $\mathbf{E}_{i,j}^{n}$ is Lipschitz if $%
\mathbf{A}_{i}^{n}$, $\mathbb{A}_{i}^{n}$, and $\mathbb{B}_{i}^{n}$ are
Lipschitz. It is easy to see that $\mathbf{A}_{i}^{n}$ will be Lipschitz if $%
\mathbf{u}_{\epsilon }^{n}\left( x\right) $ (the solution to (\ref{SA40}) at 
$t=t_{n}$) will be Lipschitz-continuous. Actually, since the bounded
operator $\mathbf{A}\left( \mathbf{u}_{\epsilon }\right) $ $(\left.
=(\partial \mathbf{f}\diagup \partial \mathbf{u)}\right| _{\mathbf{u}=%
\mathbf{u}_{\epsilon }}\mathbf{)}$ is Fr\'{e}chet-differentiable on the set $%
\Omega _{\mathbf{u}}$ $\subset $ $\mathbb{R}^{M}$, and its derivative is
bounded on $\Omega _{\mathbf{u}}$, we write: $\left\| \mathbf{A}_{i+1}^{n}-%
\mathbf{A}_{i}^{n}\right\| $ $\leq $ $C_{A}\left\| \mathbf{u}_{\epsilon
,i+1}^{n}-\mathbf{u}_{\epsilon ,i}^{n}\right\| $ $\leq $ $C_{A}C_{u}\Delta x$%
, where $C_{A}$, $C_{u}$ $=$ $const$. We may guarantee that $\mathbb{A}%
_{i}^{n}$ and $\mathbb{B}_{i}^{n}$ will be Lipschitz if $\mathbf{u}%
_{\epsilon }^{n}\left( x\right) $ will be differentiable. Actually, we may
write, by virtue of (\ref{C80}), that $\mathbb{A}_{i}=\mathbf{I}+O(\Delta x)$%
, and $\mathbb{B}_{i}=\mathbf{I}+O(\Delta x)$. Hence, Lipschitz-continuity
of $\mathbb{A}_{i}^{n}$ and $\mathbb{B}_{i}^{n}$ is obvious. However, if $%
\mathbf{u}_{\epsilon }^{n}\left( x\right) $ is only Lipschitz-continuous,
then we can not guarantee that $\mathbb{A}_{i}^{n}$ and $\mathbb{B}_{i}^{n}$
are also Lipschitz-continuous. In such a case a generalization of 
\cite[Theorem 2.10]{Borisov and Mond 2010a} can be useful, namely the case
when the scheme coefficients depend on several matrices. Then we obtain that
(\ref{SSB80}), (\ref{SSB90a})-(\ref{SSB100a}) will be stable if 
\begin{equation}
\underset{\lambda _{1},\lambda _{2}\in \left[ L_{\min ,i}^{n},L_{\max ,i}^{n}%
\right] }{\max }\left( \left| 0.25{-\lambda }_{1}\right| +\left| 0.25{%
-\lambda }_{2}\right| +\left| 0.5{+\lambda }_{1}+{\lambda }_{2}\right|
\right) \leqslant 1,\quad \forall i,n.  \label{SSB140}
\end{equation}
By virtue of (\ref{SSB140}), we find that the scheme MAC2, (\ref{CC400}),
will be stable under the same conditions (\ref{SSB125}), (\ref{SSB131}),
and, hence, under (\ref{SSB132}).

Let us note that the stability conditions, (\ref{SSB132}), are, in general,
apt to be unduly rigid. Assuming that $\aleph =\aleph (x_{i},t_{n})$ as well
as $C_{r}=C_{r}(t_{n})$, and using Weyl's inequalities (see, e.g., 
\cite[Chap. 3]{Horn and Johnson 1991}), we can find less rigid stability
conditions than (\ref{SSB132}). Let $\lambda _{i}^{n,k}$ denote the $k-th$
singular value of $\mathbf{A}_{i}^{n}$, $\lambda _{\max ,i}^{n}=\underset{k}{%
\max }\lambda _{i}^{n,k}$, $\lambda _{\max }^{n}=\underset{i}{\max }\lambda
_{\max ,i}^{n}$, and let $\alpha _{i}^{n,k}$ denote the $k-th$ singular
value of $\mathbb{A}_{i}^{n}$, $\alpha _{\min ,i}^{n}=\underset{k}{\min }%
\alpha _{i}^{n,k}$. Then, given the pre-assigned $\aleph $ $=$ $\widehat{%
\aleph }$ $\equiv $ $const$ and $C_{r}$ $=$ $\widehat{C}_{r}$ $\equiv $ $%
const\leq 0.5$, we find the sought-after stability conditions at $t=t_{n}$,
employing the following two-step procedure.

\begin{algorithm}
\label{Algorithm01}
\end{algorithm}

Step 1. Given $C_{r}$ $=$ $\widehat{C}_{r}$ $\leq $ $0.5$ and $\tau
^{n}=\Delta x\widehat{C}_{r}\diagup \lambda _{\max }^{n}$, we, by virtue of
Weyl's inequalities \cite[Chap. 3]{Horn and Johnson 1991}, estimate $L_{\max
,i}^{n}$: 
\begin{equation}
L_{\max ,i}^{n}\leq 0.5\varkappa \aleph _{i}^{n}-0.5\left( \frac{\tau ^{n}}{%
\Delta x}\right) ^{2}\left( \lambda _{\max ,i}^{n}\right) ^{2}+0.5\frac{\tau
^{n}}{\Delta x}\lambda _{\max ,i}^{n}.  \label{SSB150}
\end{equation}
In view of (\ref{SSB140}) and (\ref{SSB150}), we find 
\begin{equation}
\aleph _{i}^{n}=\min \left[ \widehat{\aleph },\left( \frac{\tau ^{n}}{\Delta
x}\right) ^{2}\frac{\left( \lambda _{\max ,i}^{n}\right) ^{2}}{\varkappa }-%
\frac{\tau ^{n}}{\varkappa \Delta x}\lambda _{\max ,i}^{n}+\frac{1}{%
2\varkappa }\right] .  \label{SSB160}
\end{equation}

Step 2.Given $\aleph _{i}^{n}$, (\ref{SSB160}), we find the derivative $%
\mathbf{d}_{i}^{n}$ using (\ref{CC240}) and, by virtue of (\ref{C80}), the
diagonal matrix $\mathbb{A}_{i}^{n}$, and hence $\alpha _{\min ,i}^{n}$. We,
by virtue of Weyl's inequalities \cite[Chap. 3]{Horn and Johnson 1991},
estimate $L_{\min ,i}^{n}$: 
\begin{equation}
L_{\min ,i}^{n}\geq \frac{\varkappa \alpha _{\min ,i}^{n}}{8}-0.5\left( 
\frac{\tau ^{n}}{\Delta x}\right) ^{2}\left( \lambda _{\max ,i}^{n}\right)
^{2}-0.5\frac{\tau ^{n}}{\Delta x}\lambda _{\max ,i}^{n}.  \label{SSB170}
\end{equation}
Let $\tau _{i}^{n}$ be used instead of $\tau ^{n}$ in (\ref{SSB170}). In
view of (\ref{SSB140}) and (\ref{SSB170}), the following inequality must be
valid. 
\begin{equation}
\frac{\varkappa \alpha _{\min ,i}^{n}}{4}-\left( \frac{\tau _{i}^{n}}{\Delta
x}\right) ^{2}\left( \lambda _{\max ,i}^{n}\right) ^{2}-\frac{\tau _{i}^{n}}{%
\Delta x}\lambda _{\max ,i}^{n}\geq -0.5.  \label{SSB180}
\end{equation}
We obtain from (\ref{SSB180}): 
\begin{equation}
\tau _{i}^{n}\frac{\lambda _{\max ,i}^{n}}{\Delta x}\leq \frac{\sqrt{%
3+\varkappa \alpha _{\min ,i}^{n}}-1}{2}.  \label{SSB185}
\end{equation}

Thus, given $\tau ^{n}$ ($=\Delta x\widehat{C}_{r}\diagup \lambda _{\max
}^{n}$) and $\tau _{i}^{n}$, (\ref{SSB185}), we find the time increment, $%
\Delta t^{n}$, at $t=t_{n}$: 
\begin{equation}
\Delta t^{n}=\min \left( \tau ^{n},\underset{i}{\min }\tau _{i}^{n}\right) \
\Rightarrow \ C_{r}^{n}=\frac{\Delta t^{n}\lambda _{\max }^{n}}{\Delta x}.
\label{SSB200}
\end{equation}

\section{Operator splitting schemes\label{OSS}}

By virtue of the operator-splitting idea \cite{Bereux and Sainsaulieu 1997}, 
\cite{Du Tao et al. 2003}, \cite{Gosse L. 2000}, \cite{LeVeque 2002} (see
also LOS in \cite{Samarskii 2001}), the following chain of equations
corresponds to the problem (\ref{INA10}) 
\begin{equation}
\frac{1}{2}\frac{\partial \mathbf{U}}{\partial t}=\frac{1}{\tau }\mathbf{q}%
\left( \mathbf{U}\right) ,\quad t_{n}<t\leq t_{n+0.5},\quad \mathbf{U}\left( 
\mathbf{x},t_{n}\right) =\mathbf{U}^{n}\left( \mathbf{x}\right) ,
\label{OS20}
\end{equation}
\begin{equation}
\frac{1}{2}\frac{\partial \mathbf{U}}{\partial t}+\sum_{j=1}^{N}\frac{%
\partial }{\partial x_{j}}\mathbf{f}_{j}\left( \mathbf{U}\right)
=0,\,t_{n+0.5}<t\leq t_{n+1},\,\mathbf{U}\left( \mathbf{x},t_{n+0.5}\right) =%
\mathbf{U}^{n+0.5}\left( \mathbf{x}\right) ,  \label{OS10}
\end{equation}
where $\mathbf{U}^{n}\left( \mathbf{x}\right) $ denotes the solution to (\ref
{OS10}) at $t=t_{n}$, $\mathbf{U}^{n+0.5}\left( \mathbf{x}\right) $ denotes
the solution to (\ref{OS20}) at $t=t_{n+0.5}$. If a high-resolution method
is used directly for the homogeneous conservation law (\ref{OS10}), then it
is natural to use a high-order scheme for (\ref{OS20}). As applied to, in
general, stiff ($\tau \ll 1$) System (\ref{INA10}), the second order schemes
can be constructed on the basis of operator-splitting techniques with ease
if (\ref{OS20}) will be approximated by an implicit scheme and (\ref{OS10})
by an explicit one, see Proposition 4.2 in \cite{Borisov and Mond 2010a}. As
an example, let us develop a central scheme for a 1-D version of (\ref{INA10}%
). After operator-splitting, the 1-D equation can be represented by the
chain of equations, namely (\ref{OS20}) and 
\begin{equation}
\frac{1}{2}\frac{\partial \mathbf{U}}{\partial t}+\frac{\partial }{\partial x%
}\mathbf{f}\left( \mathbf{U}\right) =0,\ t_{n+0.5}<t\leq t_{n+1},\ \mathbf{U}%
\left( x,t_{n+0.5}\right) =\mathbf{U}^{n+0.5}\left( x\right) .  \label{OS40}
\end{equation}
Let us first consider the case when the following first-order implicit
scheme be used for (\ref{OS20}) 
\begin{equation}
\mathbf{v}_{i}^{n+0.5}=\mathbf{v}_{i}^{n}+\frac{\Delta t}{\tau }\mathbf{q}%
\left( \mathbf{v}_{i}^{n+0.5}\right) ,  \label{OS30}
\end{equation}
and a central scheme with nonstaggered grid cells will be used for (\ref
{OS40}). We rewrite the scheme MAC1, (\ref{CC290}), to read 
\begin{equation*}
\mathbf{v}_{i}^{n+1}=0.25\left( \mathbf{v}_{i-1}^{n+0.5}+2\mathbf{v}%
_{i}^{n+0.5}+\mathbf{v}_{i+1}^{n+0.5}\right) -\varkappa \frac{\Delta x}{8}%
\left( \mathbf{d}_{i+1}^{n+0.5}-\mathbf{d}_{i-1}^{n+0.5}\right)
\end{equation*}
\begin{equation}
-\frac{\Delta t}{2\Delta x}\left( \mathbf{f}_{i+1}^{n+0.5}-\mathbf{f}%
_{i-1}^{n+0.5}\right) ,\quad \varkappa =const,\ 0\leq \varkappa \leq 1,
\label{OS60}
\end{equation}
where $\mathbf{d}_{i}^{n+0.5}$ denotes the derivative of interpolant at $%
x=x_{i}$. It is clear that the scheme (\ref{OS60}) approximates (\ref{OS40})
with the accuracy $O(\Delta t+\left( \Delta x\right) ^{2})$, however, in
view of Proposition 4.2 in \cite{Borisov and Mond 2010a}, the scheme (\ref
{OS30})-(\ref{OS60}), taken as a whole, is of the second order approximation
for the 1-D version of (\ref{INA10}).

Let us develop another nonstaggered central scheme approximating a 1-D
version of (\ref{INA10}) with the accuracy $O(\left( \Delta t\right)
^{2}+\left( \Delta x\right) ^{2})$ and such that its components (after
operator splitting) will be of the second order. It can be done on the basis
of the second order scheme (\ref{OS30}), (\ref{OS60}) with ease. Notice,
adding to and subtracting from Equation (4.10) in \cite{Borisov and Mond
2010a} (rewritten for $t_{n}<t\leq t_{n+1}$) the same quantity is equivalent
to adding this quantity to (\ref{OS60}) and subtracting it from (\ref{OS30}%
). Let $0.125\left( \Delta t\right) ^{2}(\partial ^{2}\mathbf{U\diagup }%
\partial t^{2})_{i}^{n+0.5}$ be this quantity, then we obtain the following
scheme, instead of (\ref{OS30}), (\ref{OS60}), 
\begin{equation}
\mathbf{v}_{i}^{n+0.5}=\mathbf{v}_{i}^{n}+\frac{\Delta t}{\tau }\mathbf{q}%
_{i}^{n+0.5}-\frac{\left( \Delta t\right) ^{2}}{8}\left( \frac{\partial ^{2}%
\mathbf{U}}{\partial t^{2}}\right) _{i}^{n+0.5},  \label{SOS10}
\end{equation}
\begin{equation*}
\mathbf{v}_{i}^{n+1}=0.25\left( \mathbf{v}_{i-1}^{n+0.5}+2\mathbf{v}%
_{i}^{n+0.5}+\mathbf{v}_{i+1}^{n+0.5}\right) -\varkappa \frac{\Delta x}{8}%
\left( \mathbf{d}_{i+1}^{n+0.5}-\mathbf{d}_{i-1}^{n+0.5}\right)
\end{equation*}
\begin{equation}
-\frac{\Delta t}{2\Delta x}\left( \mathbf{f}_{i+1}^{n+0.5}-\mathbf{f}%
_{i-1}^{n+0.5}\right) +\frac{\left( \Delta t\right) ^{2}}{8}\left( \frac{%
\partial ^{2}\mathbf{U}}{\partial t^{2}}\right) _{i}^{n+0.5}.  \label{SOS20}
\end{equation}
Thus, the scheme (\ref{SOS10}) as well as the scheme (\ref{SOS20}) are of
the second order, and the scheme (\ref{SOS10})-(\ref{SOS20}), taken as a
whole, is of the second order as well. Using Taylor series expansion, and
central differencing, we find 
\begin{equation*}
\mathbf{v}_{i}^{n+0.25}=\mathbf{v}_{i}^{n+0.5}-\frac{\Delta t}{4}\left( 
\frac{\partial \mathbf{U}}{\partial t}\right) _{i}^{n+0.5}+
\end{equation*}
\begin{equation}
\frac{1}{2}\left( \frac{\Delta t}{4}\right) ^{2}\left( \frac{\partial ^{2}%
\mathbf{U}}{\partial t^{2}}\right) _{i}^{n+0.5}+O\left( \left( \Delta
t\right) ^{3}\right) ,  \label{SOS25}
\end{equation}
\begin{equation}
\mathbf{v}_{i}^{n+0.5}=\mathbf{v}_{i}^{n}+\frac{\Delta t}{2}\left( \frac{%
\partial \mathbf{U}}{\partial t}\right) _{i}^{n+0.25}+O\left( \left( \Delta
t\right) ^{3}\right) .  \label{SOS27}
\end{equation}
We obtain, by virtue of (\ref{SA20}), (\ref{OS40}), that 
\begin{equation}
\frac{\partial ^{2}\mathbf{U}}{\partial t^{2}}=-2\frac{\partial }{\partial t}%
\left( \frac{\partial \mathbf{f}}{\partial x}\right) =-2\frac{\partial }{%
\partial x}\left( \frac{\partial \mathbf{f}}{\partial t}\right) =4\frac{%
\partial }{\partial x}\left( \mathbf{A}^{2}\cdot \frac{\partial \mathbf{U}}{%
\partial x}\right) ,  \label{SOS40}
\end{equation}
where $\mathbf{A=}\partial \mathbf{f\diagup }\partial \mathbf{U}$. Then 
\begin{equation*}
\left[ \frac{\partial }{\partial x}\left( \mathbf{A}^{2}\cdot \frac{\partial 
\mathbf{U}}{\partial x}\right) \right] _{i}^{n+0.5}=\frac{1}{\Delta x}\left[ 
\mathbf{D}_{i+0.5}^{n+0.5}\cdot \frac{\mathbf{v}_{i+1}^{n+0.5}-\mathbf{v}%
_{i}^{n+0.5}}{\Delta x}\right.
\end{equation*}
\begin{equation}
-\left. \mathbf{D}_{i-0.5}^{n+0.5}\cdot \frac{\mathbf{v}_{i}^{n+0.5}-\mathbf{%
v}_{i-1}^{n+0.5}}{\Delta x}\right] +O\left( \left( \Delta x\right)
^{2}\right) ,  \label{SOS50}
\end{equation}
where $\mathbf{D}_{i+0.5}^{n+0.5}=0.5\left( (\mathbf{A}_{i+1}^{n+0.5})^{2}+(%
\mathbf{A}_{i}^{n+0.5})^{2}\right) $. By virtue of (\ref{OS20}), (\ref{SOS25}%
)-(\ref{SOS50}), we rewrite Scheme (\ref{SOS10})-(\ref{SOS20}) to read 
\begin{equation}
\mathbf{v}_{i}^{n+0.25}=\mathbf{v}_{i}^{n+0.5}-\frac{\Delta t}{4\tau }\left( 
\mathbf{q}_{i}^{n+0.25}+\mathbf{q}_{i}^{n+0.5}\right) ,  \label{SOS55}
\end{equation}
\begin{equation}
\mathbf{v}_{i}^{n+0.5}=\mathbf{v}_{i}^{n}+\frac{\Delta t}{\tau }\mathbf{q}%
_{i}^{n+0.25},  \label{SOS60}
\end{equation}
\begin{equation*}
\mathbf{v}_{i}^{n+1}=0.25\left( \mathbf{v}_{i-1}^{n+0.5}+2\mathbf{v}%
_{i}^{n+0.5}+\mathbf{v}_{i+1}^{n+0.5}\right) -\varkappa \frac{\Delta x}{8}%
\left( \mathbf{d}_{i+1}^{n+0.5}-\mathbf{d}_{i-1}^{n+0.5}\right)
\end{equation*}
\begin{equation*}
+\frac{\varsigma \left( \Delta t\right) ^{2}}{2\left( \Delta x\right) ^{2}}%
\left[ \mathbf{D}_{i+0.5}^{n+0.5}\cdot \left( \mathbf{v}_{i+1}^{n+0.5}-%
\mathbf{v}_{i}^{n+0.5}\right) \right. -\left. \mathbf{D}_{i-0.5}^{n+0.5}%
\cdot \left( \mathbf{v}_{i}^{n+0.5}-\mathbf{v}_{i-1}^{n+0.5}\right) \right]
\end{equation*}
\begin{equation}
-\frac{\Delta t}{2\Delta x}\left( \mathbf{f}_{i+1}^{n+0.5}-\mathbf{f}%
_{i-1}^{n+0.5}\right) .  \label{SOS70}
\end{equation}
If $\varsigma =0$, then (\ref{SOS70}) coincides with (\ref{OS60}), being $%
O(\Delta t+\left( \Delta x\right) ^{2})$ accurate. If $\varkappa =1$ and $%
\varsigma =1$, then Scheme (\ref{SOS55})-(\ref{SOS70}) approximates the 1-D
version of (\ref{INA10}) with the accuracy $O(\left( \Delta t\right)
^{2}+\left( \Delta x\right) ^{2})$. The scheme (\ref{SOS70}) coincides, in
fact, with the scheme MAC2, (\ref{CC400}), and hence the scheme (\ref{SOS70}%
) will be stable if (\ref{SSB132}) will be valid. Notice, the less rigid
conditions for the stability of the scheme (\ref{SOS70}) can be found by
means of Algorithm \ref{Algorithm01}.

Let us note that in practice (e.g., \cite{LeVeque 2002}, \cite{Samarskii
2001}) the operator-splitting techniques find a wide range of application in
designing economical schemes for Eq. (\ref{OS10}) in the domains of
complicated geometry. The resulting method, in general, will be only
first-order accurate in time because of the splitting \cite{LeVeque 2002}, 
\cite{Samarskii 2001}. Thus, in line with established practice we will
replace the multidimensional Eq. (\ref{OS10}) by the chain of the
one-dimensional equations: 
\begin{equation}
\frac{1}{2N}\frac{\partial \mathbf{U}_{j}}{\partial t}+\frac{\partial }{%
\partial x_{j}}\mathbf{f}_{j}\left( \mathbf{U}_{j}\right) =0,\
t_{n+0.5+(j-1)\diagup (2N)}<t\leq t_{n+0.5+j\diagup (2N)},  \label{SOS90}
\end{equation}
where$\ \mathbf{U}_{j}\left( x,t_{n+0.5+(j-1)\diagup (2N)}\right) =\mathbf{U}%
_{j-1}\left( x,t_{n+0.5+(j-1)\diagup (2N)}\right) $, $\ j=1,2,\ldots N$, $%
\mathbf{U}_{0}\left( x,t_{n+0.5}\right) $ denotes the solution to (\ref{OS20}%
) at $t=t_{n+0.5}$. Eq. (\ref{OS20}) will be approximated by the first-order
implicit scheme, (\ref{OS30}), or a second-order implicit Runge-Kutta
scheme. In particular, it can be used (\ref{SOS55})-(\ref{SOS60}) or the
following implicit Runge-Kutta scheme 
\begin{equation}
\mathbf{v}_{i}^{n+0.25}=\mathbf{v}_{i}^{n+0.5}-\frac{\Delta t}{2\tau }%
\mathbf{q}_{i}^{n+0.5},\ \mathbf{v}_{i}^{n+0.5}=\mathbf{v}_{i}^{n}+\frac{%
\Delta t}{\tau }\mathbf{q}_{i}^{n+0.25},  \label{SOS100}
\end{equation}
since these schemes possess a discrete analogy to the continuous asymptotic
limit (see, e.g., \cite{Jin Shi 1995}, \cite{Naldi and Pareschi 2000}, \cite
{Pareschi Lorenzo 2001}). Let us note that the scheme (\ref{SOS100}) is an
implicit analogue to the well known Runge-Kutta scheme, which is referred,
originally due to Runge, as modified Euler method \cite{Griffiths and Higham
2010}:

\begin{equation}
\mathbf{v}_{i}^{n+0.25}=\mathbf{v}_{i}^{n}+\frac{\Delta t}{2\tau }\mathbf{q}%
_{i}^{n},\ \mathbf{v}_{i}^{n+0.5}=\mathbf{v}_{i}^{n}+\frac{\Delta t}{\tau }%
\mathbf{q}_{i}^{n+0.25}.  \label{SOS110}
\end{equation}

\section{Examples\label{Exemplification and discussion}}

In this section, we are mainly concerned with verification of the second
order central schemes COS2, (\ref{SA30}) and MAC2, (\ref{CC400}), as well as
the splitting schemes (\ref{SOS55})-(\ref{SOS70}), and (\ref{SOS90}).

\subsection{Scalar non-linear equation}

As the first stage in the verification, we will focus on the following
scalar 1-D version of the problem (\ref{INA10}): 
\begin{equation}
\frac{\partial u}{\partial t}+\frac{\partial }{\partial x}f\left( u\right)
=0,\quad x\in \mathbb{R},\ 0<t\leq T_{\max };\quad \left. u\left( x,t\right)
\right| _{t=0}=u^{0}\left( x\right) .  \label{VS10}
\end{equation}

Let us first compare the schemes MAC2, (\ref{CC400}), and COS2, (\ref{SA30}%
), and demonstrate that these schemes are of the second order. We consider
the linear transport equation, i.e. (\ref{VS10}) with $f\left( u\right)
\equiv u$, subject to the initial data: $u^{0}\left( x\right) =\sin \left(
\pi x\right) $. The numerical solutions were computed under $\varkappa
=\varsigma =1$. The scheme COS2, (\ref{SA30}), will be stable, in view of (%
\ref{SSA10}), (\ref{SSA20}), if we take the CFL number $Cr$ $=$ $\sqrt{3}-1$%
, and $\aleph =2-\sqrt{3}$. The scheme MAC2, (\ref{CC400}), will be stable,
in view of (\ref{SSB132}), under $Cr=0.5(\sqrt{3}-1)$ and $\aleph =2-\sqrt{3}
$. We will also verify the scheme COS1, (\ref{CC260}), under $Cr$ $=$ $0.5$
and $\aleph =0.25$, as well as the scheme MAC1, (\ref{CC290}), under $Cr$ $=$
$0.25$ and $\aleph =0.25$. Notice, in the case of schemes COS2 and COS1 the
CFL numbers are two times higher than in the case of schemes MAC2 and MAC1,
respectively. The reason is that the staggered schemes, COS2 and COS1, are
solved twice during the time increment, $\Delta t$, in contrast to the
nonstaggered schemes MAC2 and MAC1. $L_{1}$ errors, at $t=10$, $0\leq x\leq
2 $, versus the number of nodes are depicted in Table \ref{Table01}.

\begin{table}[tbph]
\caption{$L_{1}$ errors versus the number of nodes ($N$)}
\label{Table01}$
\begin{tabular}{|c|c|c|c|c|c|}
\hline
$N$ & $1280$ & $640$ & $320$ & $160$ & $80$ \\ \hline
COS1 & $2.4\times 10^{-3}$ & $4.9\times 10^{-3}$ & $9.8\times 10^{-3}$ & $%
2.0\times 10^{-2}$ & $3.9\times 10^{-2}$ \\ \hline
COS2 & $1.5\times 10^{-5}$ & $6.2\times 10^{-5}$ & $2.5\times 10^{-4}$ & $%
1.0\times 10^{-3}$ & $3.9\times 10^{-3}$ \\ \cline{1-1}\cline{2-6}
MAC1 & $2.5\times 10^{-3}$ & $4.9\times 10^{-3}$ & $1.0\times 10^{-2}$ & $%
2.1\times 10^{-2}$ & $4.3\times 10^{-2}$ \\ \cline{1-1}\cline{2-6}
MAC2 & $5.6\times 10^{-5}$ & $2.4\times 10^{-4}$ & $9.8\times 10^{-4}$ & $%
3.9\times 10^{-3}$ & $1.5\times 10^{-2}$ \\ \hline
\end{tabular}
$%
\end{table}
Note (Table \ref{Table01}) that $L_{1}$ errors in the case of the scheme
COS2 are approximately coincide with those appeared in the case of the
scheme MAC2 under the double number of nodes. It is associated with the fact
that in the case of the staggered scheme COS2, (\ref{SA30}), the space
increment is, actually, two times less than in the case of the nonstaggered
scheme MAC2, (\ref{CC400}).

Now we will solve the inviscid Burgers equation (i.e. $f\left( u\right)
\equiv u^{2}\diagup 2$) with the following initial condition 
\begin{equation}
u\left( x,0\right) =\left\{ 
\begin{array}{cc}
u_{0}, & x\in \left( h_{L},h_{R}\right) \\ 
0, & x\notin \left( h_{L},h_{R}\right)
\end{array}
\right. ,\quad h_{R}>h_{L},\ u_{0}=const\neq 0.  \label{VS70}
\end{equation}
The exact solution to (\ref{VS10}), (\ref{VS70}) is given by 
\begin{equation}
u\left( x,t\right) =\left\{ 
\begin{array}{cc}
u_{1}\left( x,t\right) , & 0<t\leq T \\ 
u_{2}\left( x,t\right) , & t>T
\end{array}
\right. ,  \label{VS80}
\end{equation}
where $T=2S\diagup u_{0}$, $S=h_{R}-h_{L}$, 
\begin{equation}
u_{1}\left( x,t\right) =\left\{ 
\begin{array}{cc}
\frac{x-h_{L}}{b-h_{L}}u_{0}, & h_{L}<x\leq b,\ b=u_{0}t+h_{L} \\ 
u_{0}, & b<x\leq 0.5u_{0}t+h_{R} \\ 
0, & x\leq h_{L}\ or\ x>0.5u_{0}t+h_{R}
\end{array}
\right. ,  \label{VS90}
\end{equation}
\begin{equation}
u_{2}\left( x,t\right) =\left\{ 
\begin{array}{cc}
\frac{2S\left( x-h_{L}\right) }{\left( L-h_{L}\right) ^{2}}u_{0}, & 
h_{L}<x\leq L \\ 
0, & x\leq h_{L}\ or\ x>L
\end{array}
\right. ,  \label{VS100}
\end{equation}
\begin{equation}
L=2\sqrt{S^{2}+0.5u_{0}S\left( t-T\right) }+h_{L}.  \label{VS110}
\end{equation}

First, it will be used the first order in time schemes MAC1, (\ref{CC290}),
and COS1, (\ref{CC260}). The numerical solutions were computed on a uniform
grid with spatial increments of $\Delta x=0.01$, the velocity $u_{0}=1$ in (%
\ref{VS70}), $h_{L}=0.2$, $h_{R}=1$. In view of the stability condition, (%
\ref{SSB135}), for the scheme MAC1, we take the CFL number $Cr$ $=$ $0.25$,
the monotonicity parameter $\aleph =0.25$, and the parameter $\varkappa =1$.
Having regard to the stability condition (\ref{SSA10}) under $\varsigma =0$,
we verify the scheme COS1, (\ref{CC260}), under $Cr$ $=$ $0.5$, $\aleph
=0.25 $, and $\varkappa =1$. The results of simulations are depicted with
the exact solution in Figure \ref{COSMAC1}. 
\begin{figure}[h]
\centerline{\includegraphics[width=11.50cm,height=3.8cm]{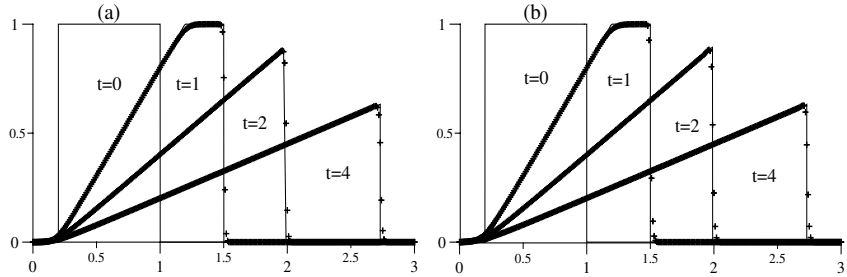}}
\caption{Inviscid Burgers equation. The schemes COS1 (a) at given $C_{r}=0.5$
and MAC1 (b) at given $C_{r}=0.25$ versus the analytical solution. Crosses:
numerical solution; Solid line: analytical solution and initial data. $%
\Delta x=0.01$, $\aleph =0.25$.}
\label{COSMAC1}
\end{figure}

We note (Figure \ref{COSMAC1}) that the schemes COS1, (\ref{CC260}), and
MAC1, (\ref{CC290}), exhibit similar results in spite of the fact that in
the case of the scheme COS1 the space increments are, virtually, two times
less than those in the case of MAC1. It should also be mentioned the absence
of spurious oscillations in the numerical solutions (see Figure \ref{COSMAC1}%
). However, if we take $\aleph =0.5$, then the spurious oscillations can be
produced by the scheme COS1 (see \cite[p. 2804]{Borisov and Mond 2010b}).
Using Algorithm \ref{Algorithm01} at given $\widehat{\aleph }=0.5$, in the
case of scheme MAC1 we obtain very slight but significant spurious
oscillations, which indicate that the boundary maximum principle is
violated. For instance, the maximum positive value of the dependent
variable, $v$, at $t=1$ is $0.4\%$ above the maximum value of $v$ at the
boundary $t=0$. These spurious oscillations can be eradicated by decreasing
the parameter $\varkappa $. Particularly, the spurious oscillations
disappear if $\varkappa =0.7$, however, this introduces an additional
numerical smearing. The results of simulations are not depicted here.

To test the schemes MAC2, (\ref{CC400}), and COS2, (\ref{SA30}), the
inviscid Burgers equation was solved under the initial condition (\ref{VS70}%
). The numerical solutions were computed under the same values of parameters
as in the case of the schemes MAC1 and COS1, but $C_{r}$ and $\aleph $. In
view of the stability condition, (\ref{SSB132}), for the scheme MAC2, we
take the CFL number $Cr$ $=$ $0.5(\sqrt{3}-1)$. It is also used Algorithm 
\ref{Algorithm01} at given $\widehat{\aleph }=0.5$. The scheme COS2 will be
stable, in view of (\ref{SSA10}), (\ref{SSA20}), if we take $Cr$ $=$ $\sqrt{3%
}-1$ and $\aleph =0.5$. The results of simulation are depicted with the
exact solution in Figure \ref{COSMAC2}. 
\begin{figure}[h]
\centerline{\includegraphics[width=11.50cm,height=3.8cm]{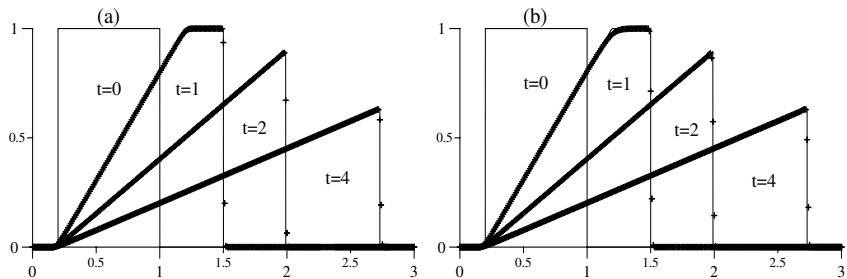}}
\caption{Inviscid Burgers equation. The schemes COS2 (a) at given $C_{r}=%
\protect\sqrt{3}-1$ and MAC2 (b) at given $C_{r}=0.5\left( \protect\sqrt{3}%
-1\right) $ versus the analytical solution. Crosses: numerical solution;
Solid line: analytical solution and initial data. $\Delta x=0.01$, $\aleph =%
\widehat{\aleph }=0.5$. }
\label{COSMAC2}
\end{figure}

We note (Figure \ref{COSMAC2}) that the schemes COS2 and MAC2 exhibit a
typical second-order nature without any spurious oscillations. The scheme
MAC2 exhibits more numerical smearing than the scheme COS2. Such a
phenomenon has, at least, two reasons. First, in the case of the staggered
scheme COS2 the space increment is, in fact, two times less than in the case
of the nonstaggered scheme MAC2. Second place, using Algorithm \ref
{Algorithm01} at given $\widehat{\aleph }=0.5$ we, in general, obtain $%
\aleph _{i}^{n}<\widehat{\aleph }$ at the grid node $(x_{i},t_{n})$,
resulting in an additional numerical smearing.

\subsection{Hyperbolic conservation laws with relaxation}

Let us consider the model system of hyperbolic conservation laws with
relaxation developed in \cite{Pember 1993a}: 
\begin{equation}
\frac{\partial w}{\partial t}+\frac{\partial }{\partial x}\left( \frac{1}{2}%
u^{2}+aw\right) =0,  \label{VS240}
\end{equation}
\begin{equation}
\frac{\partial z}{\partial t}+\frac{\partial }{\partial x}az=\frac{1}{\tau }%
Q(w,z),  \label{VS250}
\end{equation}
where 
\begin{equation}
Q(w,z)=z-m(u-u_{0}),\quad u=w-q_{0}z,  \label{VS260}
\end{equation}
$\tau $ denotes the relaxation time of the system, $q_{0}$, $m$, $a$, and $%
u_{0}$ are constants. The Jacobian, $\mathbf{A}$, can be written in the form 
\begin{equation}
\mathbf{A=}\left\{ 
\begin{array}{cc}
w-q_{0}z+a & -q_{0}\left( w-q_{0}z\right) \\ 
0 & a
\end{array}
\right\} .  \label{VS270}
\end{equation}
The system (\ref{VS240})-(\ref{VS250}) has the following frozen \cite{Pember
1993a} characteristic speeds $\lambda _{1}$ $=$ $a$, $\lambda _{2}$ $=$ $u+a$%
. The equilibrium equation for (\ref{VS240})-(\ref{VS250}) is 
\begin{equation}
\frac{\partial w}{\partial t}+\frac{\partial }{\partial x}\left( \frac{1}{2}%
u_{\ast }^{2}+aw\right) =0,  \label{VS280}
\end{equation}
where 
\begin{equation}
u_{\ast }=w-q_{0}z_{\ast },\quad z_{\ast }=\frac{m}{1+mq_{0}}\left(
w-u_{0}\right) .  \label{VS290}
\end{equation}
The equilibrium characteristic speed $\lambda _{\ast }$ can be written in
the form 
\begin{equation}
\lambda _{\ast }\left( w\right) =\frac{u_{\ast }\left( w\right) }{1+mq_{0}}%
+a.  \label{VS300}
\end{equation}

Pember's rarefaction test problem is to find the solution $\left\{
w,z\right\} $ to (\ref{VS240})-(\ref{VS250}), and hence the function $%
u=u\left( x,t\right) $, under $\tau \rightarrow 0$, and where 
\begin{equation}
\left\{ w,z\right\} =\left\{ 
\begin{array}{cc}
\left\{ w_{L},z_{\ast }\left( w_{L}\right) \right\} , & x<x_{0} \\ 
\left\{ w_{R},z_{\ast }\left( w_{R}\right) \right\} , & x>x_{0}
\end{array}
\right. ,  \label{VS310}
\end{equation}
\begin{equation}
0<u_{L}=w_{L}-q_{0}z_{\ast }\left( w_{L}\right) <u_{R}=w_{R}-q_{0}z_{\ast
}\left( w_{R}\right) .  \label{VS320}
\end{equation}
The analytical solution of this problem can be found in \cite{Pember 1993a}.
The parameters of the model system are assumed as follows: $q_{0}=-1$, $m=-1$%
, $u_{0}=3$, $a=1$, $\tau =10^{-8}$. The initial conditions of the
rarefaction problem are defined by 
\begin{equation}
u_{L}=2,\ \Longrightarrow \ z_{L}=m\left( u_{L}-u_{0}\right) =1,\
w_{L}=u_{L}+q_{0}z_{L}=1,  \label{VS330}
\end{equation}
\begin{equation}
u_{R}=3,\ \Longrightarrow \ z_{R}=m\left( u_{R}-u_{0}\right) =0,\
w_{R}=u_{R}+q_{0}z_{R}=3.  \label{VS340}
\end{equation}
The position of the initial discontinuity, $x_{0}$, is set according to the
value of $a$ so that the solutions of all the rarefaction problems are
identical \cite{Pember 1993a}. Let a position, $x_{R}^{t}$, of leading edge
or a position, $x_{L}^{t}$, of trailing edge of the rarefaction be known
(e.g., $x_{R}^{t}=0.85$, $x_{L}^{t}=0.7$ in \cite{Pember 1993a}), then 
\begin{equation}
x_{0}=x_{R}^{t}-\left( \frac{u_{R}}{1+mq_{0}}+a\right) t=x_{L}^{t}-\left( 
\frac{u_{L}}{1+mq_{0}}+a\right) t.  \label{VS350}
\end{equation}
At $t=0.3$, under (\ref{VS330})-(\ref{VS340}) we have \cite{Pember 1993a} 
\begin{equation}
u=\left\{ 
\begin{array}{cc}
2, & x\leq 0.7 \\ 
2+\frac{x-0.7}{0.85-0.7}, & 0.7<x<0.85 \\ 
3, & x\geq 0.85
\end{array}
\right. .  \label{VS360}
\end{equation}

The results of simulations, based upon the schemes MAC2, (\ref{CC400}), and
COS2, (\ref{SA30}), together with (\ref{SOS100}), under different values of
a grid spacing ($\Delta x=10^{-3}$, $5\times 10^{-4}$) are depicted in
Figure \ref{RELAX01}. In view of the stability condition, (\ref{SSB132}),
for the scheme MAC2, we take the CFL number $Cr$ $=$ $0.5(\sqrt{3}-1)$. It
is also used Algorithm \ref{Algorithm01} at given $\widehat{\aleph }=0.5$.
The scheme COS2 will be stable, in view of (\ref{SSA10}), (\ref{SSA20}), if
we take $Cr$ $=$ $\sqrt{3}-1$ and $\aleph =0.5$.

\begin{figure}[tbph]
\centerline{\includegraphics[width=14.cm,height=6.cm]{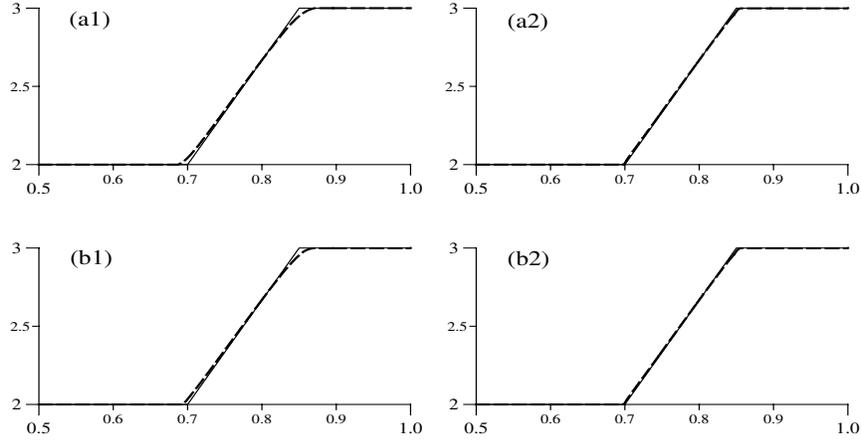}}
\caption{ Pember's rarefaction test problem. The schemes COS2 (a1, a2) at
given $C_{r}=\protect\sqrt{3}-1$ and MAC2 (b1, b2) at given $C_{r}=0.5\left( 
\protect\sqrt{3}-1\right) $ versus the analytical solution for $u$. Dashed
line: numerical solution; Solid line: analytical solution. Time $t=0.3$, $%
\aleph =\widehat{\aleph }=0.5$, (a1, b1): $\Delta x=0.001$, (a2, b2): $%
\Delta x=0.0005$. }
\label{RELAX01}
\end{figure}

One can clearly see (Figure \ref{RELAX01}) that the schemes MAC2 together
with the Runge-Kutta scheme (\ref{SOS100}) as well as COS2 together with (%
\ref{SOS100}) are free from spurious oscillations. Recall that in the case
of the staggered scheme COS2 the space increment is, in fact, two times less
than in the case of the nonstaggered scheme MAC2. Moreover, using Algorithm 
\ref{Algorithm01} in the case of scheme MAC2, we, in general, obtain $\aleph
_{i}^{n}<\widehat{\aleph }$ at the grid node $(x_{i},t_{n})$, resulting in
an additional numerical smearing. Nevertheless, the scheme MAC2 can exhibit
(Figure \ref{RELAX01}) even less numerical viscosity than the scheme COS2.

\subsection{1-D Euler equation of gas dynamics}

In this subsection we apply the second order schemes COS2, (\ref{SA30}), and
MAC2, (\ref{CC400}), to the Euler equations of gamma-law gas: 
\begin{equation}
\frac{\partial \mathbf{u}\left( x,t\right) }{\partial t}+\frac{\partial }{%
\partial x}\mathbf{F}\left( \mathbf{u}\right) =0,\quad x\in \mathbb{R},\
t>0;\quad \mathbf{u}\left( x,0\right) =\mathbf{u}^{0}\left( x\right) ,
\label{VE10}
\end{equation}
\begin{equation}
\mathbf{u\equiv }\left\{ u_{1},u_{2},u_{3}\right\} ^{T}=\left\{ \rho ,\rho
v,e\right\} ^{T},\quad \mathbf{F}\left( \mathbf{u}\right) =\left\{ \rho
v,\rho v^{2}+p,\left( e+p\right) v\right\} ^{T},  \label{VE20}
\end{equation}
\begin{equation}
e=\frac{p}{\gamma -1}+\frac{1}{2}\rho v^{2},\quad \gamma =const,
\label{VE30}
\end{equation}
where $\rho $, $v$, $p$, $e$ denote the density, velocity, pressure, and
total energy respectively. We consider the Riemann problem subject to
Riemann initial data 
\begin{equation}
\mathbf{u}^{0}\left( x\right) =\left\{ 
\begin{array}{cc}
\mathbf{u}_{L} & x<x_{0} \\ 
\mathbf{u}_{R} & x>x_{0}
\end{array}
\right. ,\quad \mathbf{u}_{L},\mathbf{u}_{R}=const.  \label{VE40}
\end{equation}
The analytic solution to the Riemann problem can be found in \cite[Sec. 14]
{LeVeque 2002}.

We solve the shock tube problem (see, e.g., \cite{Balaguer and Conde 2005}, 
\cite{Jiang et al. 1998}, \cite{LeVeque 2002}, \cite{Liu and Tadmor 1998})
with Sod's initial data: 
\begin{equation}
\mathbf{u}_{L}=\left\{ 
\begin{array}{c}
1 \\ 
0 \\ 
2.5
\end{array}
\right\} ,\quad \mathbf{u}_{R}=\left\{ 
\begin{array}{c}
0.125 \\ 
0 \\ 
0.25
\end{array}
\right\} .  \label{VE50}
\end{equation}
Following Balaguer and Conde \cite{Balaguer and Conde 2005} as well as Liu
and Tadmor \cite{Liu and Tadmor 1998} we assume that the computational
domain is $0\leq x\leq 1$; the point $x_{0}$ is located at the middle of the
interval $\left[ 0,1\right] $, i.e. $x_{0}=0.5$; the equations (\ref{VE10})
are integrated up to $t=0.16$ on a spatial grid with 200 nodes as in \cite
{Balaguer and Conde 2005} and in \cite{Liu and Tadmor 1998}. Aiming to
compare the schemes COS2 and MAC2, we, in view of the stability conditions (%
\ref{SSA10}), (\ref{SSA20}), and (\ref{SSB132}), take the CFL number $Cr$ $=$
$\sqrt{3}-1$ (the scheme COS2), and $Cr$ $=$ $0.5(\sqrt{3}-1)$ (the scheme
MAC2). In both cases we take $\varkappa \aleph =2-\sqrt{3}$, where $%
\varkappa =0.94$. The reason why the parameter $\varkappa $ has to be less
than unity is discussed in Section \ref{Construction} (see also \cite[Sec. 4]
{Borisov and Mond 2010b}). The results of simulations are depicted in Figure 
\ref{C732D5} (left column) and Figure \ref{M366D5D} (left column).

\begin{figure}[tbph]
\centerline{\includegraphics[width=11.9cm,height=15.4cm]{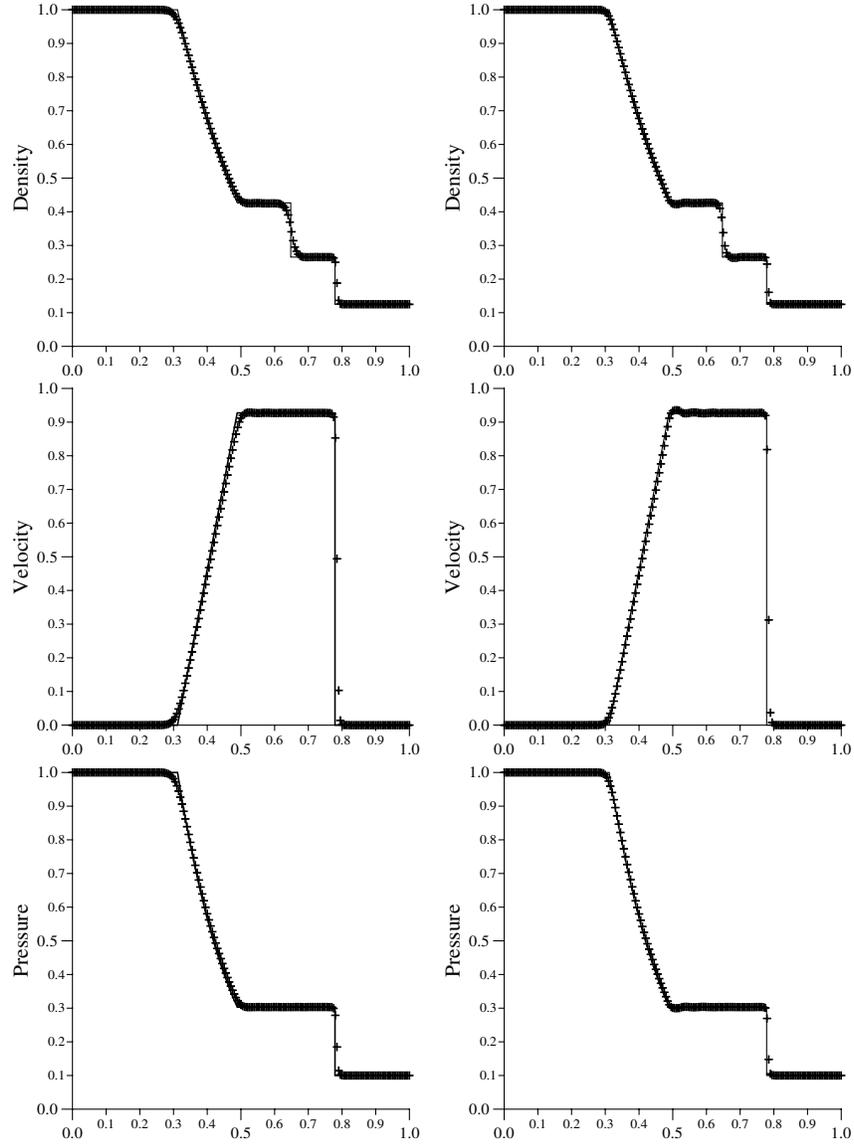}}
\caption{Sod's problem. Time $t=0.16$. The scheme COS2 versus the analytical
solution. $\varkappa \aleph =2-\protect\sqrt{3}$, $\varkappa =0.94$ (left
column) and $\aleph =0.35$, $\varkappa =1$ (right column). Crosses:
numerical solution; Solid line: analytical solution. $C_{r}=\protect\sqrt{3}%
-1$, $\Delta x=0.005$}
\label{C732D5}
\end{figure}

\begin{figure}[tbph]
\centerline{\includegraphics[width=11.9cm,height=15.4cm]{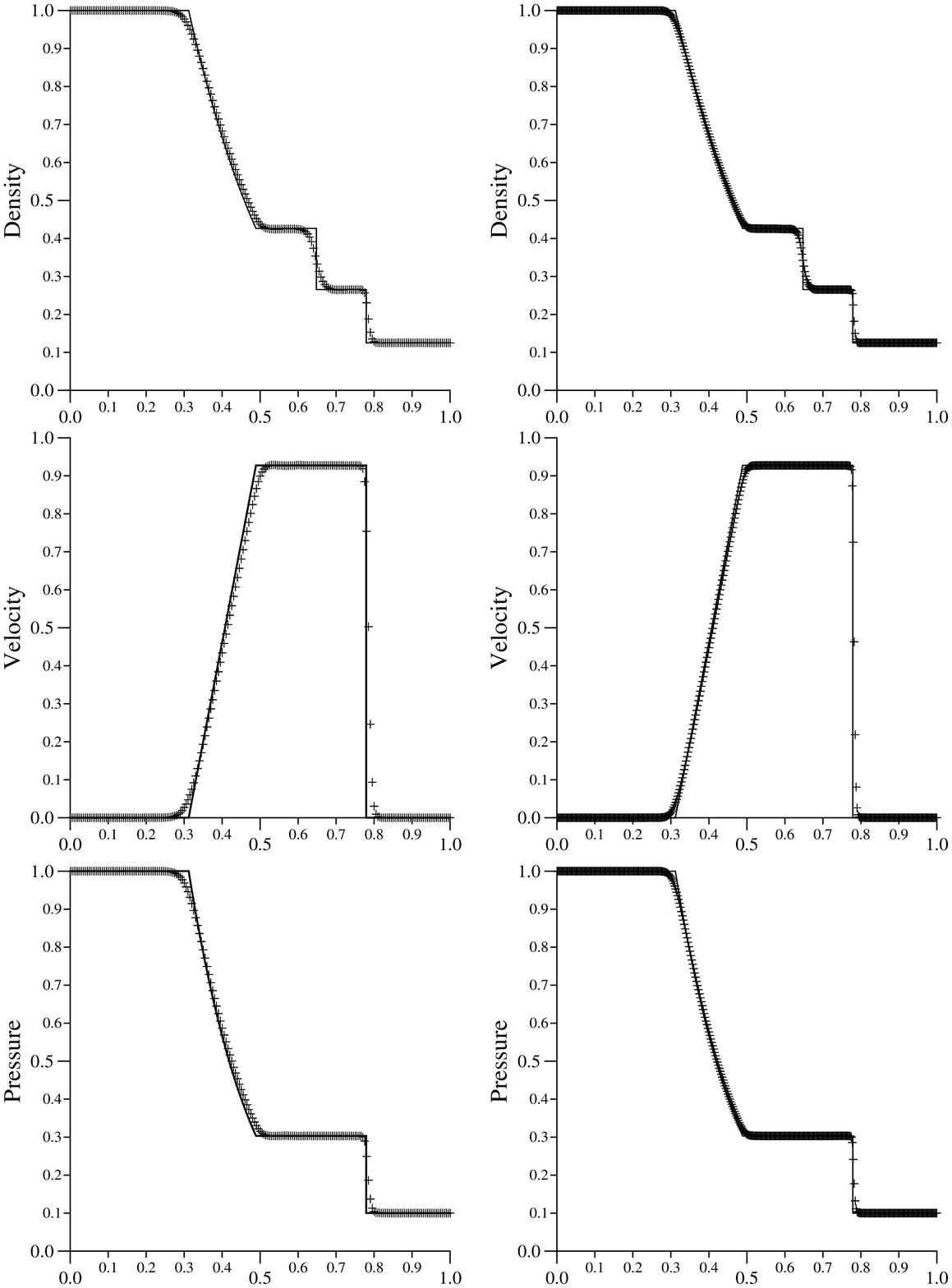}}
\caption{Sod's problem. The scheme MAC2 versus the analytical solution. Time 
$t=0.16$. $C_{r}=0.5(\protect\sqrt{3}-1)$, $\varkappa \aleph =2-\protect%
\sqrt{3}$, $\varkappa =0.94$. Crosses: numerical solution; Solid line:
analytical solution. $\Delta x=0.005$ (left column), $\Delta x=0.0025$
(right column).\ \ \ \ \ \ \ \ }
\label{M366D5D}
\end{figure}

The results using the scheme COS2 (Figure \ref{C732D5}, left column) are not
worse in comparison to the corresponding third-order central results of 
\cite[p. 418]{Liu and Tadmor 1998} as well as to the results obtained by the
fourth-order non-oscillatory scheme in \cite[p. 472]{Balaguer and Conde 2005}%
. The simulations in \cite{Balaguer and Conde 2005} and \cite{Liu and Tadmor
1998} were done under $\Delta t=0.1\Delta x$ (i.e. $0.13\lesssim Cr\lesssim
0.22$). The fourth-order scheme \cite[p. 472]{Balaguer and Conde 2005} gives
a better resolution but, in contrast to the scheme COS2 (Figure \ref{C732D5}%
, left column), can produce spurious oscillations. Let us note that the
second-order scheme COS2 can give results analogous to those obtained by the
fourth-order scheme \cite[p. 472]{Balaguer and Conde 2005}. Namely, the
higher will be the values of $\varkappa $ and $\aleph $, the better will be
the resolution (see, e.g., Figure \ref{C732D5}, right column). However, if $%
\varkappa =1$, then the scheme COS2 will be oscillatory \cite[p. 2816]
{Borisov and Mond 2010b}. It is apparent that $\varkappa $ should be
reasonably less than unity to suppress spurious oscillations, while it
should be as close to unity as possible to lose less in the order of
accuracy of the scheme \cite[p. 2816]{Borisov and Mond 2010b}. The results
of using the scheme MAC2 (Figure \ref{M366D5D}, left column) are not worse
in comparison to the corresponding second-order central results of \cite[p.
418]{Liu and Tadmor 1998}, in spite of the fact in the case of staggered
scheme the space increment is, in effect, two times less than in the case of
nonstaggered scheme. For the sake of illustration, the results using the
scheme MAC2 under $\Delta x=0.0025$ are depicted in \ Figure \ref{M366D5D}
(right column). It is easily seen that the scheme MAC2 gives a better
resolution than the scheme COS2 (Figure \ref{C732D5}, left column).

\subsection{3-D axial symmetric gas dynamics}

We consider an adiabatic expansion of a vapor-plasma plume, produced by a
laser beam on a target plane surface, into vacuum \cite{Anisimov et al. 1993}%
, i.e., the so called Anisimov's problem. It is assumed \cite{Anisimov et
al. 1993} that the expansion of the plasma plume may be described by the
ideal gas dynamic equations with a constant adiabatic exponent $\gamma $.
Let the target surface be perpendicular to the axis $z$ and located at $z=0$%
. Taking into account the symmetry of the plume with respect to the axis $z$%
, the gas-dynamic equations can be written as follows. 
\begin{equation}
\frac{\partial \rho }{\partial t}+\frac{1}{r}\frac{\partial \left( r\rho
v_{r}\right) }{\partial r}+\frac{\partial \left( \rho v_{z}\right) }{%
\partial z}=0,  \label{FA10}
\end{equation}
\begin{equation}
\frac{\partial }{\partial t}\left( \rho v_{r}\right) +\frac{1}{r}\frac{%
\partial }{\partial r}\left[ r\rho \left( v_{r}\right) ^{2}\right] +\frac{%
\partial }{\partial z}\left( \rho v_{z}v_{r}\right) +\frac{\partial p}{%
\partial r}=0,  \label{FA20}
\end{equation}
\begin{equation}
\frac{\partial }{\partial t}\left( \rho v_{z}\right) +\frac{\partial }{%
\partial z}\left[ \rho \left( v_{z}\right) ^{2}\right] +\frac{1}{r}\frac{%
\partial }{\partial r}\left( r\rho v_{z}v_{r}\right) +\frac{\partial p}{%
\partial z}=0,  \label{FA30}
\end{equation}
\begin{equation}
\frac{\partial \rho E}{\partial t}+\frac{1}{r}\frac{\partial }{\partial r}%
\left[ rv_{r}\left( \rho E+p\right) \right] +\frac{\partial }{\partial z}%
\left[ v_{z}\left( \rho E+p\right) \right] =0.  \label{FA40}
\end{equation}
\begin{equation}
\rho E=\frac{p}{\gamma -1}+0.5\rho v^{2},\ v^{2}=v_{r}^{2}+v_{z}^{2},\
\gamma =const.  \label{FA50}
\end{equation}
Let $R_{0}$ and $Z_{0}$\ will be the initial radius of the plume at $z=0$
and the initial height of the plume in $z$-direction, respectively. It is
also used in \cite{Anisimov et al. 1993} the mass $M_{P}$ of the plume and
its initial energy $E_{P}$ as the input data. We will use the following
values as the reference quantities: $l_{\ast }$ $=$ $R_{0}$, $v_{\ast }$ $=$ 
$\sqrt{\left( 5\gamma -3\right) E_{P}\diagup M_{P}}$, $t_{\ast }$ $=$ $%
l_{\ast }\diagup v_{\ast }$, $\rho _{\ast }$ $=$ $M_{P}\diagup \left(
R_{0}^{2}Z_{0}\right) $, $p_{\ast }$ $=$ $\rho _{\ast }v_{\ast }^{2}$. It is
assumed that the system (\ref{FA10})-(\ref{FA40}) is written in the
non-dimensional variables, since the dimensionless form of the gas dynamic
equations coincides with (\ref{FA10})-(\ref{FA40}), up to notations. The
pseudo-analytic solution to (\ref{FA10})-(\ref{FA40}) can be written in the
following form \cite{Anisimov et al. 1993}: 
\begin{equation}
\rho \left( r,z,t\right) =I_{1}^{-1}\left( \gamma \right) \frac{\sigma }{\xi
^{2}\eta }\left( 1-\frac{r^{2}}{\xi ^{2}}-\frac{z^{2}}{\eta ^{2}}\right)
^{1\diagup \left( \gamma -1\right) },\ \xi =\xi \left( t\right) ,\ \eta
=\eta \left( t\right) ,  \label{AS10}
\end{equation}
\begin{equation}
p\left( r,z,t\right) =\frac{I_{2}^{-1}\left( \gamma \right) }{5\gamma -3}%
\left( \frac{\sigma }{\xi ^{2}\eta }\right) ^{\gamma }\left( 1-\frac{r^{2}}{%
\xi ^{2}}-\frac{z^{2}}{\eta ^{2}}\right) ^{\gamma \diagup \left( \gamma
-1\right) },\quad \sigma =\frac{Z_{0}}{R_{0}},  \label{AS20}
\end{equation}
\begin{equation}
v_{r}\left( r,z,t\right) =r\frac{\dot{\xi}}{\xi },\ v_{z}\left( r,z,t\right)
=z\frac{\dot{\eta}}{\eta },\quad \dot{\xi}\equiv \frac{d\xi }{dt},\ \dot{\eta%
}\equiv \frac{d\eta }{dt},  \label{AS30}
\end{equation}
\begin{equation}
\xi \ddot{\xi}=\eta \ddot{\eta}=\left( \frac{\sigma }{\xi ^{2}\eta }\right)
^{\gamma -1},\quad \xi \left( 0\right) =1,\ \eta \left( 0\right) =\sigma ,\ 
\dot{\xi}\left( 0\right) =\dot{\eta}\left( 0\right) =0,  \label{AS35}
\end{equation}
where 
\begin{equation}
I_{1}\left( \gamma \right) =\frac{\pi ^{3\diagup 2}}{2}\Gamma \left( \frac{%
\gamma }{\gamma -1}\right) \diagup \,\Gamma \left( \frac{\gamma }{\gamma -1}+%
\frac{3}{2}\right) ,  \label{AS50}
\end{equation}
\begin{equation}
I_{2}\left( \gamma \right) =\frac{\pi ^{3\diagup 2}}{2\left( \gamma
-1\right) }\Gamma \left( \frac{\gamma }{\gamma -1}+1\right) \diagup \,\Gamma
\left( \frac{\gamma }{\gamma -1}+\frac{5}{2}\right) ,  \label{AS60}
\end{equation}
$\Gamma (\cdot )$ is the Gamma-function. Out of the plume, i.e. provided $%
(r\diagup \xi )^{2}$ $+$ $(z\diagup \eta )^{2}$ $>$ $1$, the values of
density ($\rho $) and pressure ($p$) as well as the components of velocity $%
\left( v_{r},v_{z}\right) $ are equal to zero.

Thus, the problem, (\ref{FA10})-(\ref{FA40}), is reduced to the system of
ordinary differential equations (ODEs) (\ref{AS35}). The system of ODEs, (%
\ref{AS35}), is solved numerically by the Runge-Kutta method with the
adaptive step-size control \cite{Press William 1988}. To check the accuracy
of the calculations, the integral of energy \cite{Anisimov et al. 1993} is
used: 
\begin{equation}
\dot{\xi}^{2}+0.5\dot{\eta}^{2}+\frac{1}{\gamma -1}\left( \frac{\sigma }{\xi
^{2}\eta }\right) ^{\gamma -1}=\frac{1}{\gamma -1}.  \label{AS90}
\end{equation}

Aiming to use the schemes COS2, (\ref{SA30}), and MAC2, (\ref{CC400}), for
solving the system (\ref{FA10})-(\ref{FA40}), we rewrite (\ref{FA10})-(\ref
{FA40}) to read 
\begin{equation}
\frac{\partial \rho }{\partial t}+\frac{\partial \left( \rho v_{r}\right) }{%
\partial r}+\frac{\partial \left( \rho v_{z}\right) }{\partial z}=-\frac{%
\rho v_{r}}{r},  \label{FA62}
\end{equation}
\begin{equation}
\frac{\partial }{\partial t}\left( \rho v_{r}\right) +\frac{\partial }{%
\partial r}\left[ \rho \left( v_{r}\right) ^{2}+p\right] +\frac{\partial }{%
\partial z}\left( \rho v_{z}v_{r}\right) =-\frac{\rho \left( v_{r}\right)
^{2}}{r},  \label{FA64}
\end{equation}
\begin{equation}
\frac{\partial }{\partial t}\left( \rho v_{z}\right) +\frac{\partial }{%
\partial z}\left[ \rho \left( v_{z}\right) ^{2}+p\right] +\frac{\partial }{%
\partial r}\left( \rho v_{z}v_{r}\right) =-\frac{\rho v_{z}v_{r}}{r},
\label{FA66}
\end{equation}
\begin{equation}
\frac{\partial \rho E}{\partial t}+\frac{\partial }{\partial r}\left[
v_{r}\left( \rho E+p\right) \right] +\frac{\partial }{\partial z}\left[
v_{z}\left( \rho E+p\right) \right] =-\frac{v_{r}\left( \rho E+p\right) }{r}.
\label{FA68}
\end{equation}

The initial conditions are the following (in details, see (\ref{AS10})-(\ref
{AS60}) under $t=0$): 
\begin{equation}
v_{r}=v_{z}=0,\ \rho =I_{1}^{-1}\left( \gamma \right) \left( 1-r^{2}-\left(
z\diagup \sigma \right) ^{2}\right) ^{1\diagup \left( \gamma -1\right) },\
p\diagup \rho ^{\gamma }=const.  \label{F160}
\end{equation}
The boundary conditions are established on the basis of the reflection
concept \cite{Anderson et al. 1984.}. At $r=0$ we assume that the axis $z$
is a reflection line. It prohibits any normal flux of mass through the
boundary $r=0$, i.e. $v_{r}=0$. Moreover, it is assumed that the pressure ($%
p $), density ($\rho $), and tangential velocity ($v_{z}$) are even
functions of normal distance to the axis $z$ while the normal velocity ($%
v_{r}$) is an odd function of $r$. It is also assumed that the plane $z=0$
is a reflection surface, i.e. the pressure ($p$), density ($\rho $), and
tangential velocity ($v_{r}$) are even functions of normal distance above
the target surface while the normal velocity ($v_{z}$) is an odd function of 
$z$.

By virtue of the operator-splitting method \cite{Bereux and Sainsaulieu 1997}%
, \cite{Du Tao et al. 2003}, \cite{Gosse L. 2000}, \cite{LeVeque 2002} (see
also LOS in \cite{Samarskii 2001}), the system (\ref{FA62})-(\ref{FA68}) may
be approximated by the following chain of equations: 
\begin{equation}
\frac{1}{4}\frac{\partial \mathbf{\check{u}}}{\partial t}+\frac{\partial }{%
\partial z}\mathbf{\check{f}}\left( \mathbf{\check{u}}\right) =0,\quad
t_{n}<t\leq t_{n+0.25},\ \left. \mathbf{\check{u}}\right| _{t=t_{n}}=\mathbf{%
\hat{u}}^{n},  \label{FA190}
\end{equation}
\begin{equation}
\frac{1}{4}\frac{\partial \mathbf{\bar{u}}}{\partial t}+\frac{\partial }{%
\partial r}\mathbf{\bar{f}}\left( \mathbf{\bar{u}}\right) =0,\
t_{n+0.25}<t\leq t_{n+0.5},\ \left. \mathbf{\bar{u}}\right| _{t=t_{n+0.25}}=%
\mathbf{\check{u}}^{n+0.25},  \label{FA300}
\end{equation}
\begin{equation}
\frac{1}{2}\frac{\partial \mathbf{\hat{u}}}{\partial t}=-\frac{1}{r}\mathbf{q%
}\left( \mathbf{\hat{u}}\right) ,\ t_{n+0.5}<t\leq t_{n+1},\ \left. \mathbf{%
\hat{u}}\right| _{t=t_{n+0.5}}=\mathbf{\bar{u}}^{n+0.5},  \label{FA360}
\end{equation}
where $\mathbf{\check{u}}=\left\{ \check{\rho}\right. ,$ $\check{\rho}\check{%
v}_{r},$ $\check{\rho}\check{v}_{z},$ $\left. \check{\rho}\check{E}\right\}
^{T}$, $\mathbf{\check{f}}\left( \mathbf{\check{u}}\right) =\left\{ \check{%
\rho}\check{v}_{z}\right. ,$ $\check{\rho}\check{v}_{z}\check{v}_{r},$ $%
\check{\rho}\left( \check{v}_{z}\right) ^{2}+\check{p},$ $\left. \check{v}%
_{z}(\check{\rho}\check{E}+\check{p})\right\} ^{T}$, $\mathbf{\bar{u}}%
=\left\{ \bar{\rho}\right. ,$ $\bar{\rho}\bar{v}_{r},$ $\bar{\rho}\bar{v}%
_{z},$ $\left. \bar{\rho}\bar{E}\right\} ^{T}$, $\mathbf{\bar{f}}\left( 
\mathbf{\bar{u}}\right) =\left\{ \bar{\rho}\bar{v}_{r}\right. ,$ $\bar{\rho}%
\left( \bar{v}_{r}\right) ^{2}+\bar{p},$ $\bar{\rho}\bar{v}_{z}\bar{v}_{r},$ 
$\left. \bar{v}_{r}\left( \bar{\rho}\bar{E}+\bar{p}\right) \right\} ^{T}$, $%
\mathbf{\hat{u}}=\left\{ \hat{\rho}\right. ,$ $\hat{\rho}\hat{v}_{r},$ $\hat{%
\rho}\hat{v}_{z},$ $\left. \hat{\rho}\hat{E}\right\} ^{T}$, $\mathbf{q}%
\left( \mathbf{\hat{u}}\right) =\left\{ \hat{\rho}\hat{v}_{r}\right. ,$ $%
\hat{\rho}\left( \hat{v}_{r}\right) ^{2},$ $\hat{\rho}\hat{v}_{z}\hat{v}%
_{r}, $ $\left. \hat{v}_{r}(\hat{\rho}\hat{E}+\hat{p})\right\} ^{T}$.

We will use the schemes COS2, (\ref{SA30}), and MAC2, (\ref{CC400}), fo
solving Eqs. (\ref{FA190})-(\ref{FA300}). To solve Eq. (\ref{FA360}) it will
be used the Runge-Kutta scheme (modified Euler method), (\ref{SOS110}). Let
us note that every point on the axis $r=0$ is a singular point for Eq. (\ref
{FA360}). Assuming that all terms at left-hand side of Equation (\ref{FA62})
are bounded values at a vicinity of $r=0$, we find that $v_{r}\rightarrow 0$
as $r\rightarrow 0$ and, hence, $\left. \mathbf{q}\left( \mathbf{\hat{u}}%
\right) \right| _{r=0}=0$. Thus 
\begin{equation}
\underset{r\rightarrow 0}{\lim }\frac{\left. \mathbf{q}\left( \mathbf{\hat{u}%
}\right) \right| _{r>0}}{r}=\underset{r\rightarrow 0}{\lim }\frac{\left. 
\mathbf{q}\left( \mathbf{\hat{u}}\right) \right| _{r>0}-\left. \mathbf{q}%
\left( \mathbf{\hat{u}}\right) \right| _{r=0}}{r}=\left. \frac{\partial 
\mathbf{q}\left( \mathbf{\hat{u}}\right) }{\partial r}\right| _{r=0}.
\label{FA362}
\end{equation}
Thus, we have at $r=0$: 
\begin{equation}
\frac{1}{2}\frac{\partial \mathbf{\hat{u}}}{\partial t}=-\frac{\partial 
\mathbf{q}\left( \mathbf{\hat{u}}\right) }{\partial r},\ t_{n+0.5}<t\leq
t_{n+1},\ \left. \mathbf{\hat{u}}\right| _{t=t_{n+0.5}}=\mathbf{\bar{u}}%
^{n+0.5}.  \label{FA364}
\end{equation}

Let $\Delta r$, $\Delta z$ denote the spatial increments and let $%
r_{i}=i\Delta r$, $z_{j}=j\Delta z$, $i,j=0$, $1$, $2$, $\ldots $ . The
Runge-Kutta scheme, for Eq. (\ref{FA360}) can be written in the form: 
\begin{equation*}
\mathbf{\hat{v}}_{i,j}^{n+0.75}=\mathbf{\hat{v}}_{i,j}^{n+0.5}-\frac{\Delta t%
}{2r_{i}}\mathbf{q}\left( \mathbf{\hat{v}}_{i,j}^{n+0.5}\right) ,
\end{equation*}
\begin{equation}
\mathbf{\hat{v}}_{i,j}^{n+1}=\mathbf{\hat{v}}_{i,j}^{n+0.5}-\frac{\Delta t}{%
r_{i}}\mathbf{q}\left( \mathbf{\hat{v}}_{i,j}^{n+0.75}\right) ,\
i=1,2,\ldots ,\ j=0,1,2,\ldots \ ,  \label{FA370}
\end{equation}
where $\mathbf{\hat{v}}_{i,j}^{n+\theta }=\left\{ \left( \hat{\rho}\right)
_{i,j}^{n+\theta }\right. ,$ $\left( \hat{\rho}\hat{v}_{r}\right)
_{i,j}^{n+\theta },$ $\left( \hat{\rho}\hat{v}_{z}\right) _{i,j}^{n+\theta
}, $ $\left. (\hat{\rho}\hat{E})_{i,j}^{n+\theta }\right\} ^{T}$, $\theta
=0.5$, $0.75$, $1$. If $i=0$, then, in view of (\ref{FA364}), the
Runge-Kutta scheme is the following 
\begin{equation}
\mathbf{\hat{v}}_{0,j}^{n+0.75}=\mathbf{\hat{v}}_{0,j}^{n+0.5}-\frac{\Delta t%
}{2}\left. \frac{\partial \mathbf{q}}{\partial r}\right| _{r=0}^{n+05},\ 
\mathbf{\hat{v}}_{0,j}^{n+1}=\mathbf{\hat{v}}_{0,j}^{n+0.5}-\Delta t\left. 
\frac{\partial \mathbf{q}}{\partial r}\right| _{r=0}^{n+0.75},\ j\geq 0\ .
\label{FA375}
\end{equation}

Notice, we are using, for the sake of convenience, the same notation for the
components of, in general, different vectors $\mathbf{\hat{u}}$ and $\mathbf{%
\hat{v}}$. To derive the scheme at $i=0$, we will use the reflection concept 
\cite{Anderson et al. 1984.}. We introduce a node $i=-1,$ i.e., $%
r_{-1}=-\Delta r$. At this node, in view of the reflection concept, we have: 
$\mathbf{\hat{v}}_{-1,j}^{n+\theta }=\left\{ \left( \hat{\rho}\right)
_{1,j}^{n+\theta }\right. ,$ $-\left( \hat{\rho}\hat{v}_{r}\right)
_{1,j}^{n+\theta },$ $\left( \hat{\rho}\hat{v}_{z}\right) _{1,j}^{n+\theta
}, $ $\left. (\hat{\rho}\hat{E})_{1,j}^{n+\theta }\right\} ^{T}$, $\theta
=0.5$, $0.75$, $j=0,1,2,\ldots $ . Then, using central differencing, we
obtain, by virtue of (\ref{FA375}), the following difference scheme for Eq. (%
\ref{FA364}): 
\begin{equation*}
\mathbf{\hat{v}}_{0,j}^{n+0.75}=\mathbf{\hat{v}}_{0,j}^{n+0.5}-\frac{\Delta t%
}{4\Delta r}\left[ \mathbf{q}\left( \mathbf{\hat{v}}_{1,j}^{n+0.5}\right) -%
\mathbf{q}\left( \mathbf{\hat{v}}_{-1,j}^{n+0.5}\right) \right] ,
\end{equation*}
\begin{equation}
\mathbf{\hat{v}}_{0,j}^{n+1}=\mathbf{\hat{v}}_{0,j}^{n+0.5}-\frac{\Delta t}{%
2\Delta r}\left[ \mathbf{q}\left( \mathbf{\hat{v}}_{1,j}^{n+0.75}\right) -%
\mathbf{q}\left( \mathbf{\hat{v}}_{-1,j}^{n+0.75}\right) \right] ,\ \
j=0,1,2,\ldots \ .  \label{FA380}
\end{equation}

The equations (\ref{FA10})-(\ref{FA40}) are integrated up to $t=1$ with $%
\sigma \equiv Z_{0}\diagup R_{0}=0.1$. It is assumed for the schemes COS2, (%
\ref{SA30}), and MAC2, (\ref{CC400}), that the monotonicity parameter $%
\aleph =0.2$ and the parameter $\varkappa =1$. In view of the stability
conditions, (\ref{SSA10}) we take $Cr$ $=$ $0.95$ for the scheme COS2.
Following to \cite{Jiang et al. 1998}, we take $Cr$ $=$ $0.475$ for the
scheme MAC2, in contrast to the sufficient (but not necessary) stability
conditions (\ref{SSB132}). It is assumed, in the case of scheme COS2, that
the spatial increments are the following: $\Delta r=10^{-3}$, $\Delta
z=5\times 10^{-4}$ if $0<t\leq 0.1$; $\Delta r=2\times 10^{-3}$, $\Delta
z=10^{-3}$ if $0.1<t\leq 0.4$; $\Delta r=2.0\times 10^{-3}$, $\Delta
z=2.0\times 10^{-3}$ if $0.4<t\leq 1$. Taking into account the fact that in
the case of the staggered scheme COS2 the space increment is, actually, two
times less than in the case of the nonstaggered scheme MAC2, we assume for
the scheme MAC2: $\Delta r=5\times 10^{-4}$, $\Delta z=2.5\times 10^{-4}$ if 
$0<t\leq 0.1$; $\Delta r=10^{-3}$, $\Delta z=5\times 10^{-4}$ if $0.1<t\leq
0.4$; $\Delta r=10^{-3}$, $\Delta z=10^{-3}$ if $0.4<t\leq 1$. The results
of simulations as well as the analytical solution are depicted in Figures 
\ref{DA2C475T04}, \ref{MA2C475T04}, \ref{DA2C475T1}, \ref{MRA2C475T1}. 
\begin{figure}[tbph]
\centerline{\includegraphics[width=12.75cm,height=16.5cm]{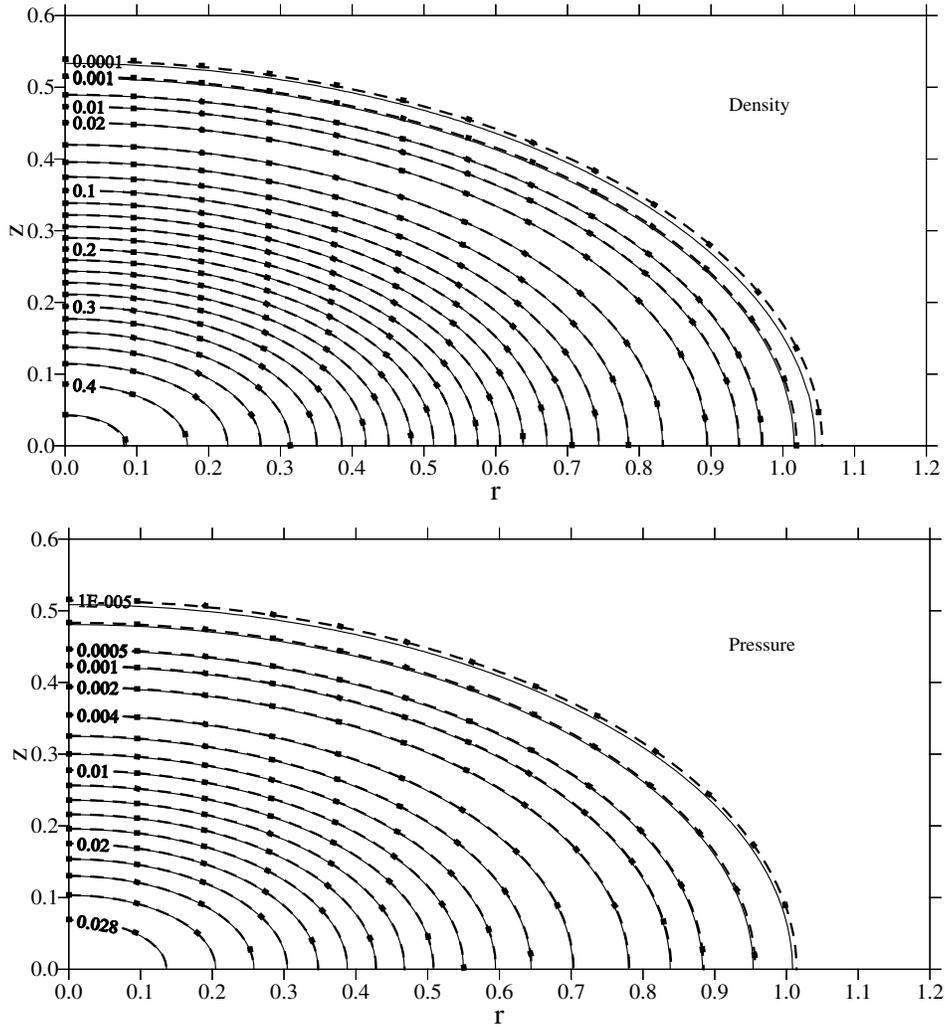}}
\caption{Anisimov's problem, density and pressure distribution. COS2 ($Cr$ $%
= $ $0.95$) and MAC2 ($Cr$ $=$ $0.475$) schemes versus the analytical
solution. $\protect\sigma \equiv Z_{0}\diagup R_{0}=0.1$, time $t=0.4$. The
monotonicity parameter $\aleph =0.2$ and the parameter $\varkappa =1$.
Dashed lines: numerical solution (COS2); Square dotted lines: numerical
solution (MAC2); Solid lines: analytical solution.}
\label{DA2C475T04}
\end{figure}
\begin{figure}[tbph]
\centerline{\includegraphics[width=12.75cm,height=16.5cm]{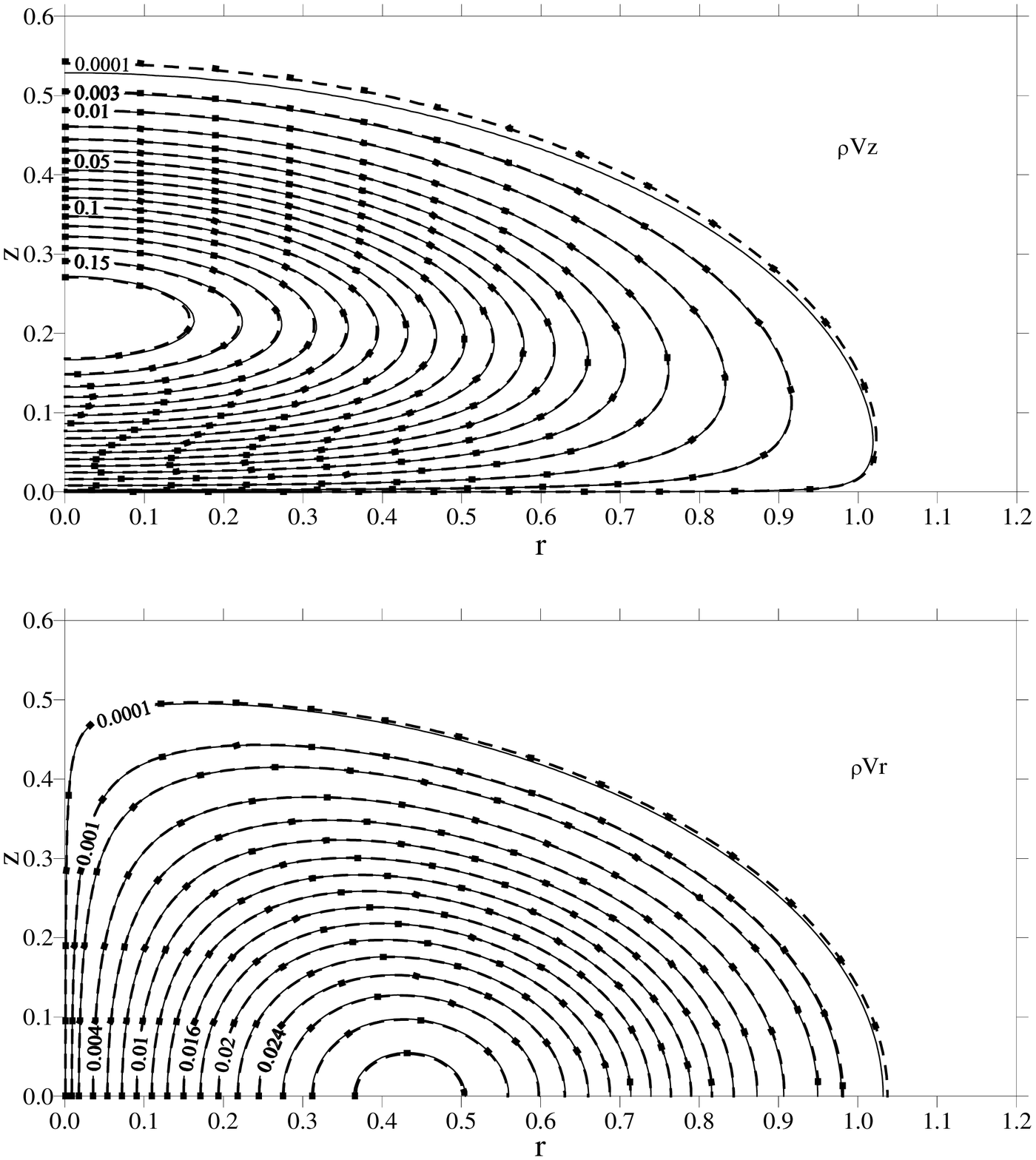}}
\caption{Anisimov's problem, momenta ($\protect\rho V_{z}$\ and $\protect%
\rho V_{r}$) distribution. COS2 ($Cr$ $=$ $0.95$) and MAC2 ($Cr$ $=$ $0.475$%
) schemes versus the analytical solution. $\protect\sigma \equiv
Z_{0}\diagup R_{0}=0.1$, time $t=0.4$. The monotonicity parameter $\aleph
=0.2$ and the parameter $\varkappa =1$. Dashed lines: numerical solution
(COS2); Square dotted lines: numerical solution (MAC2); Solid lines:
analytical solution.}
\label{MA2C475T04}
\end{figure}

\begin{figure}[tbph]
\centerline{\includegraphics[width=12.75cm,height=16.5cm]{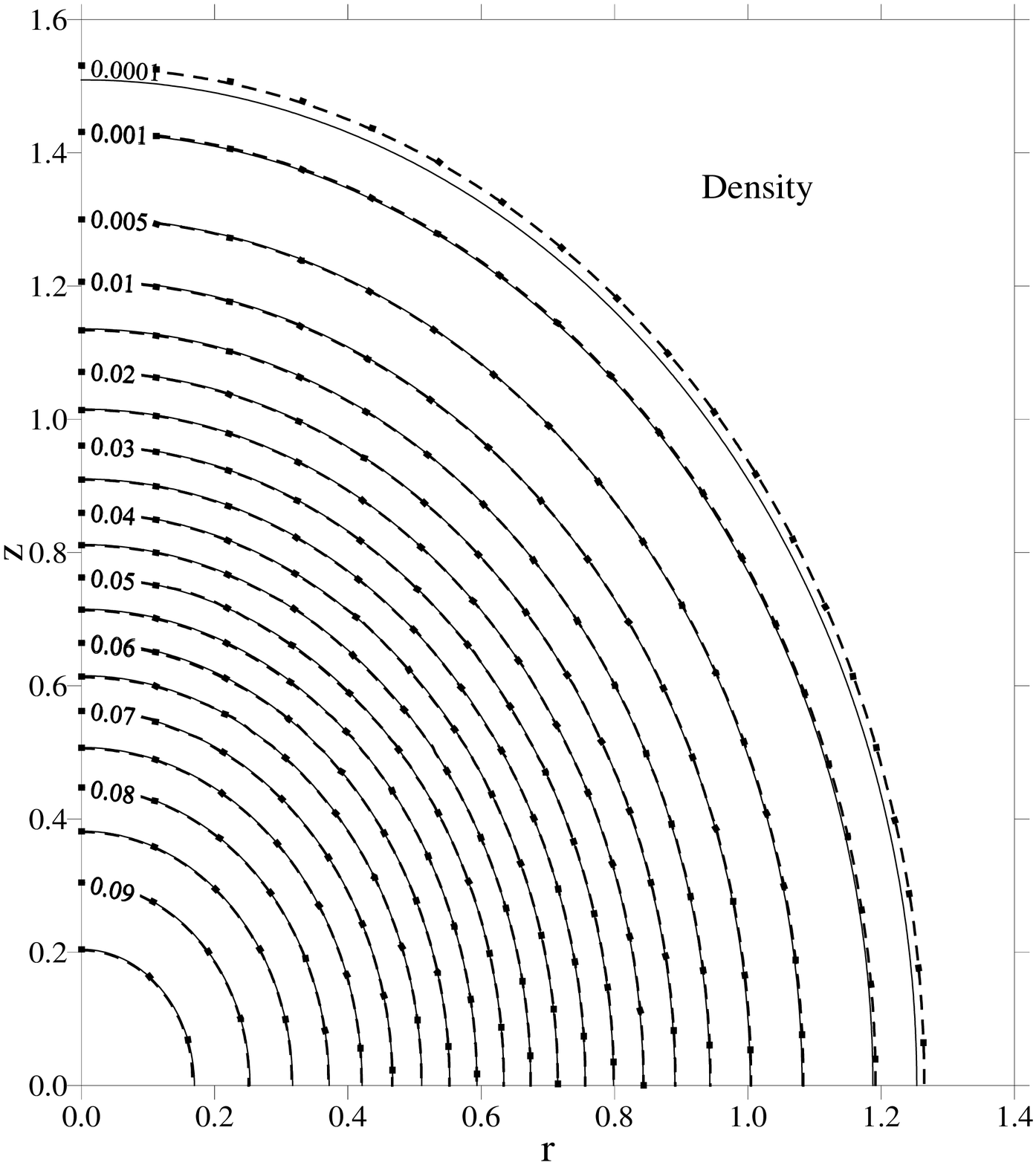}}
\caption{Anisimov's problem, density distribution. COS2 ($Cr$ $=$ $0.95$)
and MAC2 ($Cr$ $=$ $0.475$) schemes versus the analytical solution. $\protect%
\sigma \equiv Z_{0}\diagup R_{0}=0.1$, time $t=1$. The monotonicity
parameter $\aleph =0.2$ and the parameter $\varkappa =1$. Dashed lines:
numerical solution (COS2); Square dotted lines: numerical solution (MAC2);
Solid lines: analytical solution.}
\label{DA2C475T1}
\end{figure}
\begin{figure}[tbph]
\centerline{\includegraphics[width=12.75cm,height=16.5cm]{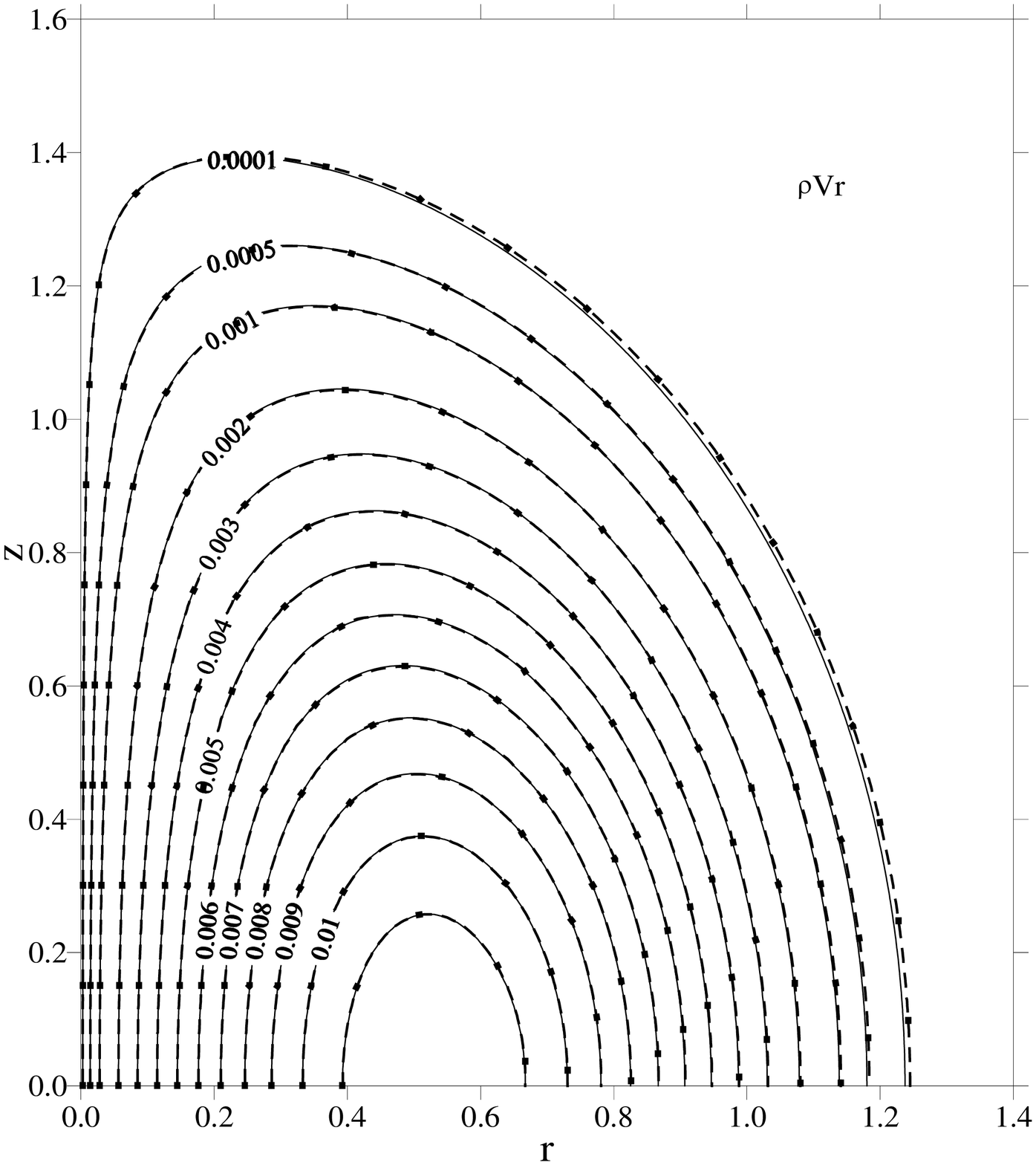}}
\caption{Anisimov's problem, momentum, $\protect\rho V_{r}$, distribution.
COS2 ($Cr$ $=$ $0.95$) and MAC2 ($Cr$ $=$ $0.475$) schemes versus the
analytical solution. $\protect\sigma \equiv Z_{0}\diagup R_{0}=0.1$, time $%
t=1$. The monotonicity parameter $\aleph =0.2$ and the parameter $\varkappa
=1$. Dashed lines: numerical solution (COS2); Square dotted lines: numerical
solution (MAC2); Solid lines: analytical solution.}
\label{MRA2C475T1}
\end{figure}

We observe (Figures \ref{DA2C475T04}, \ref{MA2C475T04}, \ref{DA2C475T1}, \ref
{MRA2C475T1}) that the numerical results obtained via the schemes COS2 and
MAC2 are practically coincide. We also observe that the relatively large
deviations between numerical and analytical solutions are occurred in the
area of the relatively large second derivatives (i.e. their absolute values)
of the primitive variables over the spatial coordinates. Hence, these
relatively large deviations may be occurred in the area of small values of
the variables in the vicinity of the front (see Figures \ref{DA2C475T04}, 
\ref{MA2C475T04}, \ref{DA2C475T1}, \ref{MRA2C475T1}) as well as in the area
of relatively large values as well (see Figure \ref{MA2C475T04}).

\end{document}